\documentclass[a4paper,USenglish]{llncs}
\usepackage{fullpage}
\usepackage{caption,subcaption}
\usepackage{microtype}
\usepackage{amsmath}
\usepackage{amssymb}
\usepackage{thm-restate}

\newcommand{\remove}[1]{}

\usepackage{dsfont}
\usepackage{xspace}
\usepackage{color}
\usepackage{soul}

\definecolor{blue}{rgb}{0.274,0.392,0.666}
\definecolor{darkblue}{rgb}{0.063,0.306,0.545}
\definecolor{red}{rgb}{1,0.3,0.3}
\definecolor{greennn}{rgb}{0,0.588,0.509}

\usepackage{hyperref}
\hypersetup{colorlinks,linkcolor={darkblue},citecolor={darkblue},urlcolor={darkblue}} 
\usepackage[nameinlink,capitalise]{cleveref}
\Crefname{observation}{Observation}{Observations}
\Crefname{algorithm}{Algorithm}{Algorithms}
\Crefname{section}{Section}{Sections}
\Crefname{lemma}{Lemma}{Lemmata}
\Crefname{claim}{Claim}{Claims}
\Crefname{figure}{Fig.}{Figs.}
\Crefname{figure}{Fig.}{Figs.}
\Crefname{property}{Property}{Properties}
\Crefname{enumi}{Condition}{Conditions.}
\Crefname{myclaim}{Claim}{Claims}
\usepackage{paralist}
\usepackage[color=yellow]{todonotes}
\usepackage{wrapfig}
\usepackage{thm-restate}
\usepackage{rotating}
\usepackage{algorithm}
\usepackage[noend]{algpseudocode}
\usepackage{multirow}
\usepackage{dsfont}
\usepackage{rotating}
\usepackage{stmaryrd}
\usepackage{rotating}

\newcommand{\skel}{sk}
\newcommand{\sk}{sk}

\usepackage{ntheorem}

\newtheorem{myclaim}{Claim}
\theoremseparator{.}

\let\doendproof\endproof
\renewcommand\endproof{~\hfill\qed\doendproof}

\renewcommand{\paragraph}[1]{\smallskip\noindent\textbf{#1}\xspace}

\graphicspath{{figs/}{img/}}

\title{Extending Upward Planar Graph Drawings}
\date{}
\author{Giordano {Da Lozzo}, Giuseppe {Di Battista}, and Fabrizio {Frati}
  \institute{
     Roma Tre University, Italy\\
    \{\href{mailto:dalozzo@dia.uniroma3.it}{dalozzo},\href{mailto:gdb@dia.uniroma3.it}{gdb},\href{mailto:frati@dia.uniroma3.it}{frati}\}\href{mailto:dalozzo@dia.uniroma3.it,gdb@dia.uniroma3.it,gdb@dia.uniroma3.it,frati@dia.uniroma3.it}{@dia.uniroma3.it}
    }
}

\authorrunning{{Da Lozzo, Di Battista, and Frati}} 

\begin{document}
\pagestyle{plain}
\maketitle

\begin{abstract}
In this paper we study the computational complexity of the {\sc Upward 
Planarity Extension} problem, which takes in input an upward planar 
drawing $\Gamma_H$ of a subgraph $H$ of a directed graph $G$ and asks 
whether $\Gamma_H$ can be extended to an upward planar drawing of $G$. 
Our study fits into the line of research on the extensibility of 
partial representations, which has recently become a mainstream in Graph 
Drawing.

We show the following results.

\begin{itemize}
	\item First, we prove that the {\sc Upward Planarity Extension} problem is NP-complete, even if $G$ has a prescribed upward embedding, the vertex set of $H$ coincides with the one of $G$, and $H$ contains no edge. 
	\item Second, we show that the {\sc Upward Planarity Extension} problem can be solved in $O(n \log n)$ time if $G$ is an $n$-vertex upward planar $st$-graph. This result improves upon a known $O(n^2)$-time algorithm, which however applies to all $n$-vertex single-source upward planar graphs. 
	\item Finally, we show how to solve in polynomial time a surprisingly difficult version of the {\sc Upward Planarity Extension} problem, in which $G$ is a directed path or cycle with a prescribed upward embedding, $H$ contains no edges, and no two vertices share the same $y$-coordinate~in~$\Gamma_H$.
\end{itemize}
\end{abstract}

\section{Introduction}\label{se:intro}

Testing whether a partial solution to a problem can be extended into a complete one is a classical algorithmic question. For instance, Kratochv\'il and Sebo~\cite{DBLP:journals/jgt/KratochvilS97} studied the vertex coloring problem when few vertices are already colored, whereas Fiala~\cite{DBLP:journals/jgt/Fiala03} considered the problem of extending a partial $3$-coloring of the edges of a graph.

The study of the extensibility of partial representations of graphs has recently become a mainstream in the graph drawing community; see, e.g., \cite{adf-tppeg-15,DBLP:conf/soda/BrucknerR17,cdk-crp-14,cfk-epr-13,cgg-pvrep-18,jkr-ktt-13,kkk-fgp-12,kko-epr-17,kko-scg-15,kko-ig-17,p-epsd-06}. Major contributions in this scenario are the result of Angelini et al.~\cite{adf-tppeg-15}, which states that the existence of a planar drawing of a graph $G$ extending a given planar drawing of a subgraph of $G$ can be tested in linear time, and the result of Br\"uckner and Rutter~\cite{DBLP:conf/soda/BrucknerR17}, which states that the problem of testing the extensibility of a given partial level planar drawing of a level graph (where each vertex -- including the ones whose drawing is not part of the input -- has a prescribed $y$-coordinate, called \emph{level}) is NP-complete. 


Upward planarity is the natural counterpart of planarity for directed graphs. In an upward planar drawing of a directed graph no two edges cross and an edge directed from a vertex $u$ to a vertex $v$ is represented by a curve monotonically increasing in the $y$-direction from $u$ to $v$; the latter property effectively conveys the information about the direction of the edges of the graph.
The study of upward planar drawings is a most prolific topic in the theory of graph visualization~\cite{DBLP:journals/algorithmica/AngeliniLBF17,DBLP:journals/algorithmica/BertolazziBD02,BertolazziBLM94,DBLP:journals/siamcomp/BertolazziBMT98,DBLP:journals/cj/BinucciD16,DBLP:journals/comgeo/Brandenburg14,ccc-pl-17,ddf-upm-18,dt-aprad-88,dtt-arsdpud-92,GargT01,DBLP:journals/cj/RextinH17}. Garg and Tamassia showed that deciding the existence of an upward planar drawing is an NP-complete problem~\cite{GargT01}. On the other hand, Bertolazzi et al.~\cite{BertolazziBLM94} showed that testing for the existence of an upward planar drawing belonging to a fixed isotopy class of planar embeddings can be done in polynomial time. Further, Di Battista et al.~\cite{dt-aprad-88} proved that any upward planar graph is a subgraph of a planar $st$-graph and as such it admits a straight-line upward planar drawing.

In this paper, we consider the extensibility of upward planar drawings of directed graphs. Namely, we introduce and study the complexity of the {\sc Upward Planarity Extension} (for short, {\sc UPE}) problem, which is defined as follows. The input is a triple $\langle G, H, \Gamma_H \rangle$, where $\Gamma_H$ is an upward planar drawing of a subgraph $H$ of a directed graph $G$; we call $H$ and $\Gamma_H$ the \emph{partial graph} and the \emph{partial drawing}, respectively. The {\sc UPE} problem asks whether $\Gamma_H$ can be extended to an upward planar drawing of $G$; or, equivalently, whether an upward planar drawing of $G$ exists which coincides with $\Gamma_H$ when restricted to \mbox{the vertices and edges of~$H$.} We also study the {\sc Upward Planarity Extension with Fixed Upward Embedding} (for short, \mbox{\sc UPE-FUE}) problem, which is the {\sc UPE} problem with the additional requirement that the drawing of $G$ we seek has to respect a given \emph{upward embedding}, i.e., a left-to-right order of the edges entering and exiting each vertex.


{\bf Related problems.} Level planar drawings are special upward planar drawings. Klemz and Rote studied the {\sc Ordered Level Planarity} ({\sc OLP}) problem~\cite{KlemzR17}, where a partial drawing of a level graph is given containing all the vertices and no edges. The problem asks for the existence of a level planar drawing of the graph extending the partial one. They show a tight border of tractability for the problem by proving NP-completeness even if no three vertices have the same $y$-coordinate and by providing a linear-time algorithm if no two vertices have the same $y$-coordinate. Br\"uckner and Rutter studied the {\sc Partial Level Planarity} ({\sc PLP}) problem~\cite{DBLP:conf/soda/BrucknerR17}, that is, the extensibility of a partial drawing of a level graph, which might contain (not necessarily all) vertices and edges. Beside proving NP-completeness even for connected graphs, they provided a quadratic-time algorithm for single-source graphs.  

The NP-hardness of the {\sc Upward Planarity Testing} problem~\cite{GargT01} directly implies the NP-hardness of the {\sc UPE} problem, as the former coincides with the special case of the latter in which the partial graph is the empty graph. Further, we have that any instance of the {\sc OLP} problem in which no $\lambda$ vertices have the same $y$-coordinate can be transformed in linear time into an equivalent instance of the {\sc UPE} problem in which the partial graph contains all the vertices and no edges, and no $\lambda$ vertices have the same $y$-coordinate in the partial drawing; moreover, a linear-time reduction can also be performed in the opposite direction. As a consequence of these reductions and of the cited results about the complexity of the {\sc OLP} problem, we obtain that the {\sc UPE} problem is NP-hard even if the partial graph contains all the vertices and no edges, and no three vertices share the same $y$-coordinate in the partial drawing, while it is linear-time solvable if the partial graph contains all the vertices and no edges, and no two vertices have the same $y$-coordinate in the partial drawing. 

Several constrained graph embedding problems that are NP-hard when the graph has a variable embedding are efficiently solvable in the fixed embedding setting; some examples are minimizing the number of bends in an \emph{orthogonal drawing}~\cite{GargT01,DBLP:journals/siamcomp/Tamassia87}, testing for the existence of an upward planar drawing~\cite{BertolazziBLM94,GargT01}, or testing for the existence of a \emph{windrose-planar drawing}~\cite{Angelini:2018:WPE:3266298.3239561}. Observe that the NP-hardness of the {\sc UPE} problem does not directly imply the NP-hardness of the {\sc UPE-FUE} problem. However, by providing a non-trivial extension of the cited NP-hardness proof of Br\"uckner and Rutter~\cite{DBLP:conf/soda/BrucknerR17} for the {\sc PLP} problem, we show that the {\sc UPE-FUE} problem is NP-hard even for connected instances whose partial graph contains all the vertices and no edges.
These proofs of NP-hardness are presented in \cref{se:complexity}.



{\bf Our contributions.} We now present an overview of our algorithmic results.

First, we identify two main factors that contribute to the complexity of the {\sc UPE} and {\sc UPE-FUE} problems:
\begin{inparaenum}[(i)]
\item The presence of edges in the partial graph and 
\item the existence of vertices with the same $y$-coordinate in the partial drawing. 
\end{inparaenum}
These two properties are strictly tied together. Namely, any instance of the {\sc UPE} or {\sc UPE-FUE} problems can be efficiently transformed into an equivalent instance $\langle G, H, \Gamma_H \rangle$ of the same problem in which $H$ contains no edges {\em or} no two vertices share the same $y$-coordinate in $\Gamma_H$ (see \cref{se:preliminaries}). Hence, the NP-hardness results for the {\sc UPE} and {\sc UPE-FUE} problems discussed above carry over to such instances, even when $V(G)=V(H)$. When the partial graph contains no edges {\em and} no two vertices share the same $y$-coordinate in the partial drawing, then the {\sc UPE} and {\sc UPE-FUE} problems appear to be more tractable. Indeed, although we can not establish their computational complexity in general, we can solve them for instances $\langle G, H, \Gamma_H \rangle$ such that $G$ is a directed path or cycle (see~\cref{se:pathsANDcycles}). In particular, in order to solve the {\sc UPE-FUE} problem for directed paths, we employ a sophisticated dynamic programming approach.


Second, we look at the {\sc UPE} and {\sc UPE-FUE} problems for instances $\langle G, H, \Gamma_H \rangle$ such that $G$ is an upward planar \emph{$st$-graph} (see~\cref{se:st-graphs}), i.e., it has a unique source $s$ and a unique sink $t$. The upward planarity of an $st$-graph is known to be decidable in $O(n)$ time~\cite{dt-aprad-88,dtt-arsdpud-92}, where $n$ is the size of the instance. We observe that a result of Br\"uckner and Rutter~\cite{DBLP:conf/soda/BrucknerR17} implies the existence of an $O(n^2)$-time algorithm to solve the {\sc UPE} problem for upward planar $st$-graphs; their algorithm works more in general for upward planar single-source graphs. We present $O(n \log n)$-time algorithms for the {\sc UPE} and {\sc UPE-FUE} problems for upward planar $st$-graphs. Notably, these results assume neither that the edge set of $H$ is empty, nor that any two vertices have distinct $y$-coordinates in $\Gamma_H$, nor that $V(G)=V(H)$.

\section{Preliminaries}\label{se:preliminaries}

In the first part of this section we give some preliminaries and definitions.

A drawing of a graph is \emph{planar} if no two edges intersect. A graph is \emph{planar} if it admits a planar drawing. A planar drawing partitions the plane into topologically connected regions, called \emph{faces}. The unique unbounded face is the \emph{outer face}, whereas the bounded faces are the \emph{internal faces}.
Two planar drawings of a connected planar graph are \emph{equivalent} if they have the same clockwise order of the edges around each vertex. A \emph{planar embedding} is an equivalence class of planar drawings of the same graph. 


For a directed graph $G$ we denote by $(u,v)$ an edge that is directed from a vertex $u$ to a vertex $v$; such an edge is \emph{incoming} at $v$ or \emph{enters} $v$, and is \emph{outgoing} from $u$ or \emph{exits} $u$. A \emph{source} of $G$ is a vertex $v$ with no incoming edges; a \emph{sink} of $G$ is a vertex $u$ with no outgoing edges. A path $(u_1,\dots,u_n)$ in $G$ is \emph{monotone} if the edge between $u_i$ and $u_{i+1}$ exits $u_i$ and enters $u_{i+1}$, for every $i=1,\dots,n-1$. A \emph{successor} (\emph{predecessor}) of a vertex $v$ in $G$ is a vertex $u$ such that there is a monotone path from $v$ to $u$ (resp.\ from $u$ to $v$). We denote by $S_G(v)$ (by $P_G(v)$) the set of successors (resp.\ predecessors) of $v$ in $G$. An edge $(u,v)$ of $G$ is \emph{transitive} if $G$ contains a monotone path from $u$ to $v$ with at least one internal vertex. When the direction of an edge in $G$ is not known or relevant, we denote it by $\{u,v\}$ instead. 

A drawing of a directed graph $G$ is \emph{upward} if each edge $(u,v)$ is represented by a curve monotonically increasing in the $y$-direction from $u$ to $v$. A drawing of $G$ is \emph{upward planar} if it is both upward and planar. A graph is \emph{upward planar} if it admits an upward planar drawing.





Consider an upward planar drawing $\Gamma$ of a directed graph $G$ and consider a vertex $v$. The list $\mathcal S(v)=[w_1,\dots,w_k]$ contains the adjacent successors of $v$ in ``left-to-right order''. That is, consider a half-line $\ell$ starting at $v$ and directed leftwards; rotate $\ell$ around $v$ in clockwise direction and append a vertex $w_i$ to $\mathcal S(v)$ when $\ell$ overlaps with the tangent to the edge $(v,w_i)$ at $v$. The list $\mathcal P(v)=[z_1,\dots,z_l]$ of the adjacent predecessors of $v$ is defined similarly. Then two upward planar drawings of a connected directed graph are \emph{equivalent} if they have the same lists $\mathcal S(v)$ and $\mathcal P(v)$ for each vertex $v$. An \emph{upward embedding} is an equivalence class of upward planar drawings. If a vertex $v$ in an upward planar graph $G$ is not a source or a sink, then a planar embedding of $G$ determines $\mathcal S(v)$ and $\mathcal P(v)$. However, if $v$ is a source or a sink, then different upward planar drawings might have different lists $\mathcal S(v)$ or $\mathcal P(v)$, respectively. The combinatorial properties of the upward embeddings have been characterized by Bertolazzi et al.~\cite{BertolazziBLM94}.




Given an upward planar graph $G$ with a fixed upward embedding, and given a subgraph $G'$ of $G$, we assume that $G'$ is associated with the upward embedding corresponding to the upward planar drawing of $G'$ obtained from an upward planar drawing of $G$ with the given upward embedding by removing the vertices and edges not in $G'$. 

In order to study the time complexity of the {\sc UPE} and {\sc UPE-FUE} problems, we assume that an instance $\langle G, H, \Gamma_H \rangle$ of any of such problems is such that $\Gamma_H$ is a \emph{polyline drawing}, that is, a drawing in which the edges are represented as polygonal lines. Then we define the \emph{size} of $\langle G, H, \Gamma_H \rangle$ as $|\langle G, H, \Gamma_H \rangle| = |V(G)| + |E(G)| + s$, where $s$ is the number of segments that compose the polygonal lines representing the edges in $\Gamma_H$.



Consider an upward planar $st$-graph $G$ with a fixed upward embedding. In any upward planar drawing $\Gamma$ of $G$ every face $f$ is delimited by two monotone paths $(u_1,\dots,u_k)$ and $(v_1,\dots,v_l)$ connecting the same two vertices $u_1=v_1$ and $u_k=v_l$. Assuming that $\mathcal S(u_1) = [\dots,u_2,v_2,\dots]$, we call $(u_1,\dots,u_k)$ the \emph{left boundary} of $f$ and $(v_1,\dots,v_l)$ the \emph{right boundary} of $f$. For a vertex $v\neq t$ in $G$, the \emph{leftmost outgoing path} $\mathcal L^+_G(v)=(w_1,\dots,w_m)$ of $v$ is the monotone path such that $w_1=v$, $w_m=t$, and $\mathcal S(w_i)=[w_{i+1},\dots]$, for each $i=1,\dots,m-1$. The \emph{rightmost outgoing path} $\mathcal R^+_G(v)$, the \emph{leftmost incoming path} $\mathcal L^-_G(v)$ and the \emph{rightmost incoming path} $\mathcal R^-_G(v)$ are defined similarly. The paths $\mathcal L^+_G(s)$ and $\mathcal R^+_G(s)$ are also called \emph{leftmost} and \emph{rightmost path} of $G$, respectively. Note that these paths delimit the outer face of $G$. Consider a monotone path $\cal Q$ from $s$ to $t$. Let $\mathcal Q^*$ be obtained by extending $\mathcal Q$ with a $y$-monotone curve directed upwards from $t$ to infinity and with a $y$-monotone curve directed downwards from $s$ to infinity. Then a vertex $u$ is \emph{to the left} (\emph{to the right}) of $\mathcal Q$ if it lies in the region to the left (resp.\ to the right) of $\mathcal Q^*$. In particular, $u$ is \emph{to the left} of a vertex $v$ if it lies to the left of the monotone path composed of $\mathcal L^+_G(v)$ and $\mathcal L^-_G(v)$. Analogously, $u$ is \emph{to the right} of $v$ if it lies to the right of the monotone path composed of $\mathcal R^+_G(v)$ and $\mathcal R^-_G(v)$. We denote by $L_G(v)$ (by $R_G(v)$) the set of vertices that are to the left (resp.\ to the right) of a vertex $v$ in $G$. Note that all the definitions introduced in this paragraph do not depend on the actual drawing of $G$, but only on its upward embedding.

\subsection{SPQR-Trees}\label{se:spqr-trees}

A \emph{cut-vertex} in a graph $G$ is a vertex whose removal disconnects $G$. A \emph{separation pair} in $G$ is a pair of vertices whose removal disconnects $G$. A graph is \emph{biconnected} (\emph{triconnected}) if it has no cut-vertex (resp.\ no separation pair). A \emph{biconnected component} of $G$ is a maximal biconnected subgraph of $G$.

Let $G$ be an $n$-vertex biconnected upward planar $st$-graph containing the edge $(s,t)$. A \emph{split pair} of $G$ is either a separation pair or a pair of adjacent vertices. A \emph{split component} for a split pair $\{u,v\}$ is either the edge $(u,v)$ or a maximal subgraph $G'$ of $G$ which is an upward planar $uv$-graph and such that $\{u, v\}$ is not a split pair of $G'$. A split pair $\{u,v\}$ is \emph{maximal} if there is no distinct split pair $\{w,z\}$ in $G$ such that $\{u, v\}$ is contained in a split component of $\{w,z\}$.

\begin{figure}[tb!]
	\centering
	\includegraphics[height=.25\textwidth]{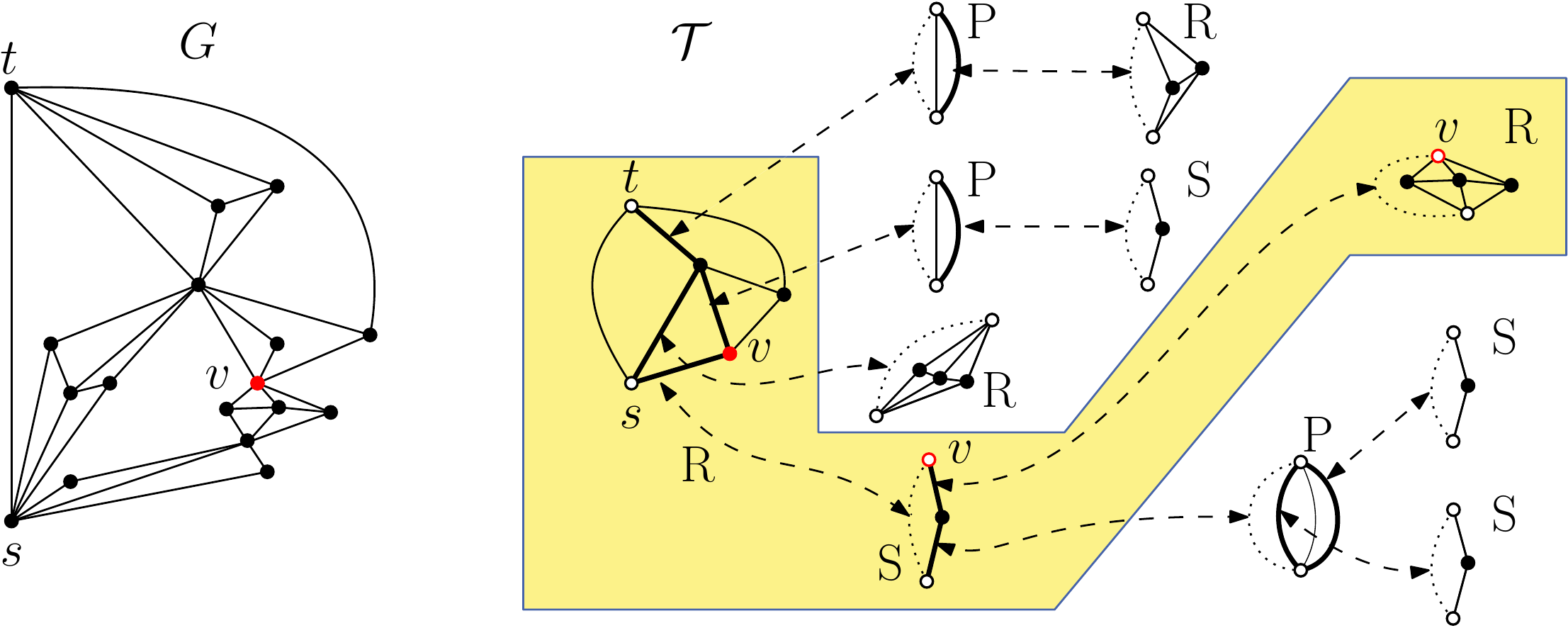}
	\caption{(left) A biconnected upward planar $st$-graph $G$ and (right) the SPQR-tree $\cal T$ of $G$ rooted at the R-node $\mu$ adjacent to the edge $(s,t)$. The skeletons of all the non-leaf nodes of $\cal T$ are depicted; virtual edges corresponding to edges of $G$ are thin, whereas virtual edges corresponding to S-, P-, and R-nodes are thick. The allocation nodes of the vertex $v$ are in the yellow-shaded region; the node $\mu$ is the proper allocation~node~of~$v$.}
	\label{fig:SPQR}
\end{figure}

The SPQR-tree $\mathcal T$ of $G$, defined as in~\cite{DBLP:journals/siamcomp/BattistaT96}, describes a recursive decomposition of $G$ with
respect to its split pairs and represents succinctly all the upward planar embeddings of $G$. The tree $\mathcal T$ is a rooted tree with four types of nodes: S, P, Q, and R (refer to~\cref{fig:SPQR}). Any node $\mu$ of $\mathcal T$ is associated with an upward planar $uv$-graph, called \emph{skeleton} of $\mu$, which might contain multiple edges and which we denote by $\sk(\mu)$. The edges of $\sk(\mu)$ are called \emph{virtual edges}. The tree $\mathcal T$ is recursively defined as follows.

\begin{itemize}
	\item {\em Trivial case}. If $G$ consists of a single edge $(s,t)$, then $\mathcal T$ is a Q-node $\mu$ and $\sk(\mu)$ also coincides with the edge $(s,t)$.
	\item {\em Series case}. If $G$ is not a single edge and is not biconnected, then let $c_1, \dots, c_{k-1}$ (for some $k \geq  2$) be the cut-vertices of $G$, where $c_i$ belongs to two biconnected components $G_i$ and $G_{i+1}$, for $i=1,\dots,k-1$, with $s\in V(G_1)$ and $t\in V(G_k)$. Further, set $c_0=s$ and $c_k=t$; then $G_i$ is an upward planar $c_{i-1}c_i$-graph, for $i=1,\dots,k$. The root of $\mathcal T$ is an S-node $\mu$. Finally, $\sk(\mu)$ is a monotone path $(c_0,c_1,\dots,c_k)$ plus the edge $(s,t)$. 
	\item {\em Parallel case}. If $G$ is not a single edge, if it is biconnected, and if $\{s,t\}$ is a split pair of $G$ defining split components $G_1,\dots,G_k$ (for some $k \geq  2$), then $G_i$ is an upward planar $st$-graph, for $i=1,\dots,k$. The root of $\mathcal T$ is a P-node $\mu$.  Finally, $\sk(\mu)$ consists of $k+1$ parallel edges $(s,t)$.
	\item {\em Rigid case}. If $G$ is not a single edge, if $G$ is biconnected, and if $\{s,t\}$ is not a split pair of $G$, then let $\{s_1,t_1\},\dots,\{s_k,t_k\}$ be the maximal split pairs of $G$ (for some $k \geq  1$). Further, let $G_i$ be the union of all the split components of $\{s_i,t_i\}$, for $i=1,\dots,k$. Then $G_i$ is an upward planar $s_i t_i$-graph, for $i=1,\dots,k$. The root of $\mathcal T$ is an R-node $\mu$. Finally, the graph $\sk(\mu)$ is obtained from $G$ by replacing each subgraph $G_i$ with an edge $(s_i,t_i)$ and by adding the edge $(s,t)$. We have that $\sk(\mu)$ is a triconnected upward planar $st$-graph.
\end{itemize}

In each of the last three cases, the subtrees of $\mu$ are the SPQR-trees of $G_1,\dots,G_k$, rooted at the children $\mu_1,\dots,\mu_k$ of $\mu$, respectively. Further, the virtual edge $(s,t)$ in $\sk(\mu)$ (just one of the edges $(s,t)$ in the case of a P-node) is associated with the parent of $\mu$ in $\mathcal T$, while every other virtual edge $e_i$ is associated with a child $\mu_i$ of $\mu$ and with the graph $G_i$. The graph $G_i$ corresponds to a virtual edge $e_i$ of $\sk(\mu)$ and it is called the \emph{pertinent graph} of $e_i$ and of $\mu_i$. 

The overall tree $\mathcal T$ of $G$ is rooted at the only neighbor of the Q-node $(s,t)$; the skeleton for the root of $\mathcal T$ is defined slightly differently from the other nodes, as it does not contain the virtual edge representing its parent. It is known that $\mathcal T$ has $O(n)$ nodes and that the total number of virtual edges in the skeletons of the nodes of $\mathcal T$ is in $O(n)$. All the upward embeddings of $G$ can be obtained by suitably permuting the virtual edges of the skeletons of the P-nodes and by flipping the skeletons of the R-nodes. For a specific choice of such permutations and flips, an upward embedding is recursively obtained by substituting the virtual edges in the skeleton of a node $\mu$ with the upward embeddings associated to the children of $\mu$. 

Let $v$ be a vertex of $G$. The \emph{allocation nodes} of $v$ are the nodes of $\mathcal T$ whose skeletons contain $v$. Note that $v$ has at least one allocation node. The lowest common ancestor of the allocation nodes of vertex $v$ is itself an allocation node of $v$ and it is called the \emph{proper allocation node} of $v$. Let $\mu$ be a node of $\mathcal T$. The \emph{representative} of $v$ in $\skel(\mu)$ is the vertex or edge $x$ of $\skel(\mu)$ defined as follows: If $\mu$ is an allocation node of $v$, then $x=v$; otherwise, $x$ is the edge of $\skel(\mu)$ whose pertinent graph contains $v$. 

From a computational complexity perspective, the SPQR-tree $\mathcal T$ of an $n$-vertex upward planar $st$-graph $G$ can be constructed in $O(n)$ time. Further, within the same time bound, it is possible to set up a data structure that allows us to query for the proper allocation node of a vertex of $G$ in $O(1)$ time and to query for the lowest common ancestor of two nodes of $\mathcal T$ in $O(1)$ time~\cite{DBLP:journals/siamcomp/BattistaT96}.

\subsection{Simplifications}\label{se:reductions}

In this section we prove that it is not a loss of generality to restrict our attention to instances of the {\sc UPE} and {\sc UPE-FUE} problems in which the partial graph contains no edges {\em or} no two vertices share the same $y$-coordinate in the partial drawing. We first deal with instances in which the partial graph \mbox{contains no edges.}



\begin{restatable}{lemma}{lemmaNoEdges}\label{le:no-edges}
	Let $\langle G, H, \Gamma_H \rangle$ be an instance of the {\sc UPE} or {\sc UPE-FUE} problem and let $n = |\langle G, H, \Gamma_H \rangle|$. There exists an equivalent instance $\langle G', H', \Gamma_{H'} \rangle$ of the {\sc UPE} or {\sc UPE-FUE} problem, respectively, \mbox{such that:} 
	\begin{enumerate}[\bf (i)]
		\item $E(H')=\emptyset$,
		\item if $V(H)=V(G)$, then $V(H')=V(G')$, and
		\item if $G$ is an $st$-graph, then $G'$ is an $st$-graph.
	\end{enumerate}
	Further, the instance $\langle G', H', \Gamma_{H'} \rangle$  has $O(n)$ size and can be constructed in $O(n \log {n})$ time. The drawing $\Gamma_{H'}$ may contain vertices with the same $y$-coordinate even if $\Gamma_H$ does not.
\end{restatable}

\begin{proof}
	Throughout this proof we will assume that $\Gamma_H$ is a straight-line drawing of $H$. This is not a loss of generality, as if an edge $(u,v)$ of $H$ is represented in $\Gamma_H$ by a polygonal line $(u,b_1,\dots,b_h,v)$, where $b_1,\dots,b_h$ are the bends of the polygonal line, then dummy vertices can inserted on $b_1,\dots,b_h$ in $\Gamma_H$; further, the edge $(u,v)$ becomes a monotone path $(u,b_1,\dots,b_h,v)$ both in $H$ and in $G$. Note that the size of the instance remains asymptotically the same. 
	
	We now define the instance $\langle G', H', \Gamma_{H'} \rangle$ and argue about its size. We will later describe how to compute it efficiently.  
	
	The graph $G'$ is obtained from $G$ by replacing certain edges of $G$ that are also in $H$ by monotone paths; how to precisely perform such a replacement will be described later. If $G$ has a prescribed upward embedding, then $G'$ derives its upward embedding from the one of $G$. That is, if an edge $(u,v)$ is replaced by a monotone path $(u=u_1,u_2,\dots,u_{k-1},u_k=v)$, then $u_2$ substitutes $v$ in $\mathcal S(u)$ and $u_{k-1}$ substitutes $u$ in $\mathcal P(v)$; further, for $i=2,\dots,k-1$, we have $\mathcal S(u_i)=[u_{i+1}]$ and  $\mathcal P(u_i)=[u_{i-1}]$. Property (iii) of the lemma's statement is trivially satisfied.
	
	The graph $H'$ is composed of all the vertices of $H$ plus all the vertices internal to the monotone paths that are inserted in $G'$ to replace edges of $G$ that are also in $H$. Property (ii) of the lemma's statement is then satisfied. Further, $H'$ contains no edge, hence Property (i) of the lemma's statement is satisfied. 
	
	The drawing $\Gamma_{H'}$ coincides with $\Gamma_{H}$ when restricted to the vertices also belonging to $H$. It remains to specify the lengths of the monotone paths that are inserted in $G'$ to replace edges of $G$ that are also in $H$ and to describe how to place their internal vertices in $\Gamma_{H'}$. This is done in the following.

	\begin{figure}[htb]
		\centering
		\begin{subfigure}{.4\textwidth}
			\centering
			\includegraphics[scale=0.7]{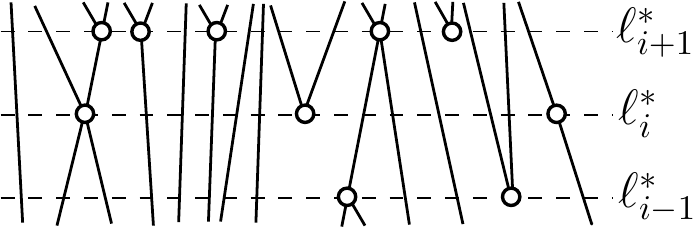}
			\label{fi:noedge1}
		\end{subfigure}
		\hfil
		\begin{subfigure}{.4\textwidth}
			\centering
			\includegraphics[scale=0.7]{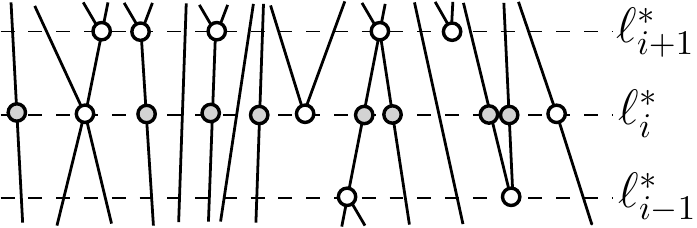}
			\label{fi:noedge2}
		\end{subfigure}
		\caption{The drawings $\Gamma_{H}$ (left) and $\Gamma_{H'}$ (right) in the proximity of $\ell^*_i$. The vertices that are inserted on $\ell^*_i$ are gray; those inserted on $\ell^*_{i-1}$ and $\ell^*_{i+1}$ are not shown.}
		\label{fi:noedges}
	\end{figure}

	Refer to~\cref{fi:noedges}. Among all the possible $y$-coordinates, we call {\em interesting} the ones of the vertices of $H$ in $\Gamma_H$. We examine the interesting $y$-coordinates in increasing order $y^*_1,\dots,y^*_m$. For any $y^*_i$, denote by $\ell^*_i$ the line with equation $y=y^*_i$. We look at the left-to-right order $X^*_i$ in which the vertices of $H$ lying on $\ell^*_i$ and the edges of $H$ crossing $\ell^*_i$ appear in $\Gamma_{H}$. We place a vertex $v$ in $\Gamma_{H'}$ at the point $p_v$ in which an edge $e$ of $H$ crosses $\ell^*_i$ in two cases: 
	
	\begin{enumerate}[(i)]
		\item If $e$ is preceded or followed by a vertex of $H$ in $X^*_i$; or 
		\item if $e$ has an end-vertex whose $y$-coordinate in $\Gamma_H$ is $y^*_{i-1}$ or $y^*_{i+1}$.
	\end{enumerate} 
	If we place a vertex $v$ at the crossing point $p_v$ of $e$ with $\ell^*_i$, then $v$ is also a vertex that is internal to the monotone path that is inserted in $G'$ to replace $e$. Clearly, the edges of the monotone paths are directed so to connect vertices in increasing order of their $y$-coordinates. This concludes the construction of $\langle G', H', \Gamma_{H'} \rangle$.
	
	We prove that the size of $\langle G', H', \Gamma_{H'} \rangle$ is linear in the size of $\langle G, H, \Gamma_H \rangle$. The only vertices which belong to $G'$ and not to $G$ are those internal to the monotone paths that are inserted in $G'$ to replace edges of $G$ that are also in $H$. The number of vertices that are inserted in $G'$ because of case (i) is at most $2|V(H)|$; further, the number of vertices that are inserted in $G'$ because of case (ii) is at most $2|E(H)|$. Since the only edges that belong to $G'$ and not to $G$ are those of the monotone paths above, it follows that the number of such edges is also in $O(|V(H)|+|E(H)|)$. The claim about \mbox{the size of $\langle G', H', \Gamma_{H'} \rangle$ follows.}  
	
	We next show how to construct the instance $\langle G', H', \Gamma_{H'} \rangle$ efficiently. 
	
	First, we construct a list $\cal V$ in which the vertices of $H$ are ordered by increasing $y$-coordinate and, secondarily, by increasing $x$-coordinate in $\Gamma_H$. This is done in~$O(n\log n)$~time.
	
	Handling the vertices that need to be inserted in $G'$ because of case (ii) is easy. Indeed, it suffices to look at every edge $e$ of $H$. Let $y^*_i$ and $y^*_j$ be  the $y$-coordinates of its end-vertices, for some $i<j$; then insert vertices at the crossing points of $e$ with the lines $y=y^*_{i+1}$ and $y=y^*_{j-1}$. This can be done in $O(n)$ overall time. 
	
	Handling the vertices that need to be inserted in $G'$ because of case (i) is more difficult. The rest of the analysis of the running time deals with this. 
	
	For each interesting $y$-coordinate $y^*_i$, a trivial algorithm would: (1) compute the crossing points between $\ell^*_i$ and the edges of $H$ crossing it; (2) order such crossing points together with the vertices of $H$ lying on $\ell^*_i$ by increasing $x$-coordinate; and (3) insert the new vertices at the crossing points which are preceded or followed by vertices of $H$. However, this might require $\Omega(n^2\log n)$ time in total, as $\Omega(n)$ edges might span the $y$-coordinates of $\Omega(n)$ vertices of $H$ in $\Gamma_H$. We now describe how to bring the time complexity down to $O(n\log n)$. 
	
	We examine the interesting $y$-coordinates in increasing order $y^*_1,\dots,y^*_m$; denote by $\mathcal V^*_i$ the restriction of $\cal V$ to the vertices on $\ell^*_i$.
	
	For $i=1,\dots,m$, our algorithm computes the list $X^*_i$ defined above; note that $X^*_1$ coincides with $\mathcal V^*_1$. Further, for $i=1,\dots,m-1$, our algorithm also computes a list $E_{i,i+1}$ which contains the edges of $H$ crossing the line $y=(y^*_i+y^*_{i+1})/2$ in $\Gamma_H$, in left-to-right order. 
	
	Assume that $X^*_i$ has been computed, for some $i\in\{1,\dots,m-1\}$. We compute $E_{i,i+1}$ as follows. We initialize $E_{i,i+1}$ to $X^*_i$. Then, for each vertex $v$ in $\mathcal V^*_i$, we replace $v$ in $E_{i,i+1}$ by its outgoing edges in left-to-right order as they appear in $\Gamma_H$. The vertex $v$ is accessed in $X^*_i$ and in $E_{i,i+1}$ in $O(1)$ time via a pointer associated to $v$ in $\mathcal V^*_i$. Thus, $E_{i,i+1}$ can be constructed in time proportional to the number of edges exiting vertices in $\mathcal V^*_i$, hence in total $O(n)$ time.
	
	Now assume that $E_{i,i+1}$ has been computed, for some $i\in\{1,\dots,m-1\}$. We compute $X^*_{i+1}$ as follows. We initialize $X^*_{i+1}$ to $E_{i,i+1}$. Then, for each vertex $v$ in $\mathcal V^*_{i+1}$, we determine whether $E_{i,i+1}$ contains edges entering $v$. In case it does, such edges are removed from $X^*_{i+1}$ (indeed, they ``end'' on $\ell^*_{i+1}$), and $v$ is inserted in their place; otherwise, $v$ is just inserted in a suitable position in $X^*_{i+1}$.  
	
	The above approach is realized by performing, for each vertex $v$ in $\mathcal V^*_i$, a binary search on $X^*_{i+1}$, which compares $x(v)$ to the $x$-coordinates of the intersection points of the edges in $X^*_{i+1}$ with $\ell^*_{i+1}$. More precisely, when the binary search considers an edge $e$ in $X^*_{i+1}$, the $x$-coordinate $x_i(e)$ of its intersection point with $\ell^*_{i+1}$ is computed in $O(1)$ time and then compared to $x(v)$; we stress the fact that not all the intersection points of the edges in $X^*_{i+1}$ with $\ell^*_{i+1}$ are computed, but only the $O(\log n)$ ones that are compared with $x(v)$ in the binary search. The binary search stops in two possible situations. The first is the one in which an edge $e$ is found that enters $v$; that is, $e$ is such that $x_i(e)=x(v)$. That edge and all the other edges entering $v$ are then removed from $X^*_{i+1}$ and $v$ is inserted in place of such edges (and a pointer to it is set in $V^*_{i+1}$). The second situation in which the binary search stops is the one in which the $x$-coordinates of two consecutive elements of $X^*_{i+1}$ have been compared to $x(v)$, and $x(v)$ has been found to be in between such coordinates. Then $v$ is inserted between such elements (and a pointer to it is set in $V^*_{i+1}$). Finding the proper place where to insert $v$ in $X^*_{i+1}$ takes $O(\log n)$ time, and hence $O(n \log n)$ time over all the vertices of $H$. Further, removing the edges entering $v$ can be done in $O(|P_H(v)|)$ time because all such edges appear consecutively in $E_{i,i+1}$, and hence in $X^*_{i+1}$, by the planarity of $\Gamma_H$; hence this takes $O(n)$ time in total.
	
	Once $X^*_{i+1}$ has been computed, we determine the placement in $\Gamma_{H'}$ of the vertices in $V(H')\setminus V(H)$ that are inserted on $\ell^*_{i+1}$. In order to do that, we again access each vertex $v\in \mathcal V^*_{i+1}$; then we look at the two elements adjacent to $v$ in $X^*_{i+1}$, and if any of them is an edge we insert a vertex at its intersection point with $\ell^*_{i+1}$. For each $v\in \mathcal V^*_{i+1}$ this takes $O(1)$ time, hence $O(n)$ time in total.
	
	It remains to prove the equivalence between $\langle G, H, \Gamma_{H} \rangle$ and $\langle G', H', \Gamma_{H'} \rangle$. 
	
	One direction is easy. Indeed, assume that an upward planar drawing $\Gamma_G$ of $G$ exists which extends $\Gamma_H$. We modify $\Gamma_G$ by placing each vertex $v\in V(H')\setminus V(H)$ at the point $p_v$ where $v$ lies in $\Gamma_{H'}$. Let $e$ be the edge of $G$ that is also in $H$ and that is replaced by a monotone path which includes $v$ as an internal vertex in $G'$. By construction $p_v$ lies on the curve representing $e$ in $\Gamma_G$; hence, the insertion of all the vertices in $V(H')\setminus V(H)$ in $\Gamma_G$ results in an upward planar drawing $\Gamma_{G'}$ of $G'$ extending $\Gamma_{H'}$. 
	
	The other direction is much more involved. Assume that an upward planar drawing $\Gamma_{G'}$ of $G'$ extending $\Gamma_{H'}$ exists. By construction, all the vertices of $V(H)$ are also in $V(H')$; further, these vertices are placed at the same points in $\Gamma_{H}$ and $\Gamma_{H'}$. Let $\Gamma_{G}$ be the drawing of $G$ obtained from $\Gamma_{G'}$ by interpreting the representation of each monotone path of $G'$ corresponding to an edge $e$ of $G$ as a representation of $e$. Then we have that $\Gamma_{G}$ is an upward planar drawing of $G$ such that the vertices in $V(H)$ are placed as in $\Gamma_H$. However, it might be the case that the edges of $G$ that are also in $H$ are not drawn in $\Gamma_{G}$ as they are required to be drawn in $\Gamma_{H}$. We hence need to modify the drawing of such edges in $\Gamma_{G}$.
	
	Our proof that such a modification is always feasible is a direct consequence of~\cref{cl:simplification-no-edges} below. We introduce some definitions. Every edge of $E(G) \cap E(H)$ corresponds to a monotone path in $G'$; we call {\em replacement edge} each edge of such monotone paths. Note that the end-vertices of a replacement edge are in $H'$, while the replacement edge itself is not in $H'$. Recall that, for $i=1,\dots,m$, the list $X^*_i=[x^1_i,x^2_i,\dots,x^t_i]$ represents the left-to-right order in which the vertices of $H$ lying on $\ell^*_i$ and the edges of $H$ crossing $\ell^*_i$ appear along $\ell^*_i$ in $\Gamma_{H}$. By construction, each vertex of $H$ lying on $\ell^*_i$ is also a vertex of $H'$ lying on $\ell^*_i$; further, each edge of $H$ crossing $\ell^*_i$ either determines a vertex of $H'$ lying on $\ell^*_i$ or corresponds to a replacement edge of $G'$ crossing $\ell^*_i$. Hence, there is a bijection $\gamma:X^*_i\rightarrow Z^*_i$ between $X^*_i$ and a set $Z^*_i$ which contains the vertices of $H'$ lying on $\ell^*_i$ and the replacement edges of $G'$ crossing $\ell^*_i$.
	
	While the left-to-right order $[x^1_i,x^2_i,\dots,x^t_i]$ of the elements of $X^*_i$ along $\ell^*_i$ in any upward planar drawing of $G$ extending $\Gamma_H$ is fixed, as all the elements of $X^*_i$ are drawn in $\Gamma_{H}$, some elements of $Z^*_i$ do not have any fixed intersection with $\ell^*_i$. Namely, the elements of $Z^*_i$ corresponding to vertices of $H'$ lying on $\ell^*_i$ do have a fixed intersection with $\ell^*_i$; however, in principle, the elements of $Z^*_i$ corresponding to replacement edges of $G'$ crossing $\ell^*_i$ do not have any restriction on where they should cross $\ell^*_i$ with respect to the other elements of $Z^*_i$. The following claim proves that such a restriction actually exists.
	
	%
	%
	\begin{myclaim} \label{cl:simplification-no-edges}
		In any upward planar drawing $\Gamma_{G'}$ of $G'$ extending $\Gamma_{H'}$ the left-to-right order of the elements of $Z^*_i$ along $\ell^*_i$ is $[\gamma(x^1_i),\gamma(x^2_i),\dots,\gamma(x^t_i)]$.
	\end{myclaim}
	
	\begin{proof}
		We prove the claim by induction on $i$. The base case, in which $i=1$, is trivial, since all the elements of $Z^*_1$ are vertices of $H'$, whose left-to-right order along $\ell^*_1$ is determined~by~$\Gamma_{H'}$. 
		
		Now suppose that the claim holds for $Z^*_i$, for some $1\leq i\leq m-1$; we prove the claim for $Z^*_{i+1}$. 
		
		We first argue about the elements $\gamma(x^j_{i+1})\in Z^*_{i+1}$ that are associated with a $y$-monotone curve $\lambda^j_{i+1}$ that extends between $\ell^*_{i}$ and $\ell^*_{i+1}$ in $\Gamma_{G'}$. More precisely, an element $\gamma(x^j_{i+1})$ of $Z^*_{i+1}$ is of interest here in one of the following two cases: 
		
		\begin{enumerate}[(1)]
			\item $\gamma(x^j_{i+1})$ is a replacement edge $e$ of $G'$ crossing $\ell^*_{i+1}$. 
			\item $\gamma(x^j_{i+1})$ is a vertex of $H'$ lying on $\ell^*_i$ that is incident to a replacement edge $e$ \mbox{which is incoming at $\gamma(x^j_{i+1})$.} 
		\end{enumerate}
		
		Note that, in both cases, the edge $e$ either crosses $\ell^*_{i}$ or has its source on $\ell^*_i$. Then, in both cases, we define $\lambda^j_{i+1}$ as the part of the replacement edge $e$ between $\ell^*_{i}$ and $\ell^*_{i+1}$. Let $M^*_{i+1}\subseteq Z^*_{i+1}$ be the set of such elements $\gamma(x^j_{i+1})$. Hence, $M^*_{i+1}$ contains all the elements of $Z^*_{i+1}$, except for the vertices of $H'$ lying on $\ell^*_{i+1}$ in $\Gamma_{H'}$ corresponding to sources of $H$.


		Consider any element $\gamma(x^j_{i+1})\in M^*_{i+1}$, its associated $y$-monotone curve $\lambda^j_{i+1}$, and the replacement edge $e$ the curve $\lambda^j_{i+1}$ is part of. If $e$ crosses $\ell^*_{i}$ in $\Gamma_{G'}$, then $e$ is represented by an element $\gamma(x^{l}_i)$ in $Z^*_{i}$. Otherwise, $e$ has its source on $\ell^*_{i}$; we denote such a source also by $\gamma(x^{l}_i)$. The elements $\gamma(x^j_{i+1})$ and $\gamma(x^{l}_i)$ correspond to elements $x^j_{i+1}$ and $x^{l}_i$ of $X^*_{i+1}$ and $X^*_{i}$, respectively. A $y$-monotone curve $\delta^j_{i+1}$ between a point of $\ell^*_{i+1}$ and a point of $\ell^*_{i}$ is associated with $x^j_{i+1}$; this curve represents part of an edge of $H$ which might cross $\ell^*_{i+1}$ and/or $\ell^*_{i}$ in $\Gamma_H$ (then $x^j_{i+1}$ and/or $x^{l}_i$ represent such an edge), or might have an end-vertex on $\ell^*_{i+1}$ and/or $\ell^*_{i}$ in $\Gamma_H$ (then $x^j_{i+1}$ and/or $x^{l}_i$ represent the end-vertices of such an edge).
		
		Now consider any two distinct elements $\gamma(x^j_{i+1}),\gamma(x^k_{i+1})\in M^*_{i+1}$. These are associated to elements $\gamma(x^{l}_i)$ and $\gamma(x^{h}_i)$ of $Z^*_{i}$, respectively, via the curves $\lambda^j_{i+1}$ and $\lambda^k_{i+1}$. Assume w.l.o.g. that $\gamma(x^j_{i+1})$ is to the left of $\gamma(x^k_{i+1})$ along $\ell^*_{i+1}$ in $\Gamma_{G'}$. We establish that $x^j_{i+1}$ precedes $x^k_{i+1}$~in~$X^*_{i+1}$.
		
		\begin{itemize}
			\item If $\gamma(x^{l}_i)\neq \gamma(x^{h}_i)$, then $\gamma(x^{l}_i)$ is to the left of $\gamma(x^{h}_i)$, as otherwise the $y$-monotone curves $\lambda^j_{i+1}$ and $\lambda^k_{i+1}$ would cross, while $\Gamma_{G'}$ is planar. By induction, $x^{l}_i$ is to the left of $x^{h}_i$ in $X^*_i$. Hence $x^j_{i+1}$ precedes $x^k_{i+1}$ in $X^*_{i+1}$, as otherwise the $y$-monotone curves $\delta^j_{i+1}$ and $\delta^k_{i+1}$ would cross, while $\Gamma_{G'}$ is assumed to be planar.
			\item If $\gamma(x^{l}_i)= \gamma(x^{h}_i)$, then $\lambda^j_{i+1}$ and $\lambda^k_{i+1}$ represent parts of replacement edges of $G'$ whose sources coincide with $\gamma(x^{l}_i)= \gamma(x^{h}_i)$. Since $\gamma$ defines a bijection between $X^*_i$ and $Z^*_i$, we have $x^{l}_i= x^{h}_i$, hence $x^j_{i+1}$ and $x^k_{i+1}$ are edges of $H$ both incident to $x^{l}_i= x^{h}_i$. By construction, a vertex is inserted in $\Gamma_{H'}$ at the crossing point between $x^j_{i+1}$ and $\ell^*_{i+1}$, and a vertex is inserted in $\Gamma_{H'}$ at the crossing point between $x^k_{i+1}$ and $\ell^*_{i+1}$ (these vertices are inserted in case (ii) of the construction of $\langle G', H', \Gamma_{H'}\rangle$). Such vertices are in fact $\gamma(x^j_{i+1})$ and $\gamma(x^k_{i+1})$. Hence, the edges $x^j_{i+1}$ and $x^k_{i+1}$ are in this left-to-right order along $\ell^*_{i+1}$, thus $x^j_{i+1}$ precedes $x^k_{i+1}$ in $X^*_{i+1}$. 
		\end{itemize}
		
	It follows that, in any upward planar drawing $\Gamma_{G'}$ of $G'$ extending $\Gamma_{H'}$, the left-to-right order in which the elements of $M^*_{i+1}$ appear along $\ell^*_{i+1}$ coincides with the order $[\gamma(x^1_{i+1}),\gamma(x^2_{i+1}),\dots,\gamma(x^r_{i+1})]$ restricted to the elements~of~$M^*_{i+1}$.
		
	Now consider any maximal sequence $x^j_{i+1},x^{j+1}_{i+1}, \dots,x^l_{i+1}$ of consecutive elements of $X^*_{i+1}$ that are sources of $H$. Since $\gamma(x^j_{i+1}),\gamma(x^{j+1}_{i+1}), \dots,\gamma(x^l_{i+1})$ are vertices of $H'$, they appear in the order $\gamma(x^j_{i+1}),\gamma(x^{j+1}_{i+1}), \dots,\gamma(x^l_{i+1})$ along $\ell^*_{i+1}$. Hence, we only have to ensure that $\gamma(x^j_{i+1}),\gamma(x^{j+1}_{i+1}), \dots,\gamma(x^l_{i+1})$ fit between the suitable elements of $M^*_{i+1}$ along $\ell^*_{i+1}$. Assume that $x^{j-1}_{i+1}$ and $x^{l+1}_{i+1}$ both exist (the case in which $x^j_{i+1},x^{j+1}_{i+1}, \dots,x^l_{i+1}$ is an initial and/or final subsequence of $X^*_{i+1}$ is easier to handle). By the maximality of the sequence $x^j_{i+1},x^{j+1}_{i+1}, \dots,x^l_{i+1}$, we have that $\gamma(x^{j-1}_{i+1})$ and $\gamma(x^{l+1}_{i+1})$ belong to $M^*_{i+1}$, with $\gamma(x^{j-1}_{i+1})$ immediately preceding $\gamma(x^{l+1}_{i+1})$ along $\ell^*_{i+1}$. If $x^{j-1}_{i+1}$ is a vertex of $H$, then $\gamma(x^{j-1}_{i+1})$ is a vertex of $H'$, hence $\gamma(x^{j-1}_{i+1})$ precedes $\gamma(x^j_{i+1})$ in left-to-right order along $\ell^*_{i+1}$ in $\Gamma_{H'}$ and hence in $\Gamma_{G'}$. Otherwise, $x^{j-1}_{i+1}$ is an edge of $H$ crossing $\ell^*_{i+1}$; however, since $x^j_{i+1}$ is a vertex of $H$, a vertex has been inserted in $\Gamma_{H'}$ at the crossing point between $x^{j-1}_{i+1}$ and $\ell^*_{i+1}$ (this vertex is inserted in case (i) of the construction of $\langle G', H', \Gamma_{H'}\rangle$). In fact $\gamma(x^{j-1}_{i+1})$ is such a vertex. Then $\gamma(x^{j-1}_{i+1})$ precedes $\gamma(x^j_{i+1})$ in left-to-right order along $\ell^*_{i+1}$ in $\Gamma_{H'}$ and hence in $\Gamma_{G'}$. An analogous argument proves that $\gamma(x^l_{i+1})$ precedes $\gamma(x^{l+1}_{i+1})$ in left-to-right order along $\ell^*_{i+1}$ in $\Gamma_{H'}$ and hence in $\Gamma_{G'}$. The claim follows.
	\end{proof}
	
	\cref{cl:simplification-no-edges} implies that the desired modification of $\Gamma_{G}$, which ensures that the edges of $H$ are drawn as in $\Gamma_H$, can be performed ``locally'' to each strip $\mathcal S^*_i$ delimited by two consecutive lines $\ell^*_i$ and $\ell^*_{i+1}$. Then one can start from the intersection between $\Gamma_H$ and $\mathcal S^*_i$, which consists of a set of $y$-monotone curves connecting points on $\ell^*_i$ with points on $\ell^*_{i+1}$, and independently draw the part of $\Gamma_{G}$ inside each region of $\mathcal S^*_i$ defined by such~$y$-monotone~curves.

	This concludes the proof of the lemma.
\end{proof}

We next deal with instances of the {\sc UPE} and {\sc UPE-FUE} problems in which no two vertices have the same $y$-coordinate in the given partial drawing. 

\begin{restatable}{lemma}{LemmaNoVerticesSameY}\label{le:no-vertices-sameY}
	Let $\langle G, H, \Gamma_H \rangle$ be an instance of the {\sc UPE} or {\sc UPE-FUE} problem and let $n = |\langle G, H, \Gamma_H \rangle|$. There exists an equivalent instance $\langle G', H', \Gamma_{H'} \rangle$ of the {\sc UPE} or {\sc UPE-FUE} problem, respectively, such that: 
	\begin{enumerate}[\bf (i)]
		\item no two vertices of $H'$ share the same $y$-coordinate in $\Gamma_{H'}$ and
		\item if $V(H)=V(G)$, then $V(H')=V(G')$.
	\end{enumerate}
	Further, the instance $\langle G', H', \Gamma_{H'} \rangle$  has $O(n)$ size and can be constructed in $O(n \log {n})$ time. The graph $H'$ may contain edges even if $H$ does not.
\end{restatable}

\begin{proof}
	First, we apply~\cref{le:no-edges} in order to transform $\langle G, H, \Gamma_H \rangle$ into an equivalent instance $\langle G^*, H^*, \Gamma_{H^*} \rangle$ such that $H^*$ contains no edges and such that $V(H^*)=V(G^*)$ if $V(H)=V(G)$. The instance $\langle G^*, H^*, \Gamma_{H^*} \rangle$ has $O(n)$ size and can be constructed in $O(n \log n)$ time. Hence, it suffices to prove the statement of the lemma with  $\langle G^*, H^*, \Gamma_{H^*} \rangle$ in place of $\langle G, H, \Gamma_H \rangle$. 
	
	
	We define the desired instance $\langle G', H', \Gamma_{H'} \rangle$ as follows. Refer to~\cref{fig:no-two-vertices-sameY}. We initialize $\langle G', H', \Gamma_{H'} \rangle = \langle G^*, H^*, \Gamma_{H^*} \rangle$. We construct a list $\cal V$ of the vertices of $H^*$ ordered by increasing $y$-coordinates and, secondarily, by increasing $x$-coordinates in $\Gamma_{H^*}$. This is done in $O(n\log n)$ time. Among all the possible $y$-coordinates, we call {\em interesting} the ones of the vertices of $H^*$ in $\Gamma_{H^*}$. We examine the interesting $y$-coordinates in increasing order $y^*_1,\dots,y^*_m$. For any $y^*_i$, denote by $\ell^*_i$ the line with equation $y=y^*_i$; further, denote by $\mathcal V^*_i$ the restriction of $\cal V$ to the vertices on $\ell^*_i$. 
	
	\begin{figure}[htb]
		\centering
		\hfil
		\begin{subfigure}{.4\textwidth}
			\includegraphics[scale=0.7]{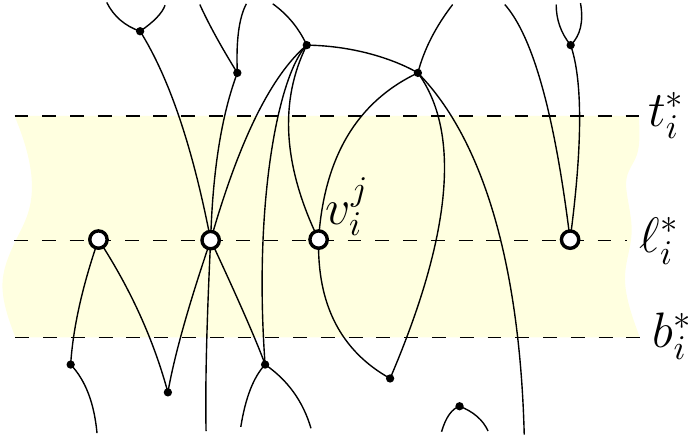}
			\label{fi:sameY1}
			\subcaption{}
		\end{subfigure}
		\begin{subfigure}{.4\textwidth}
			\includegraphics[scale=0.7]{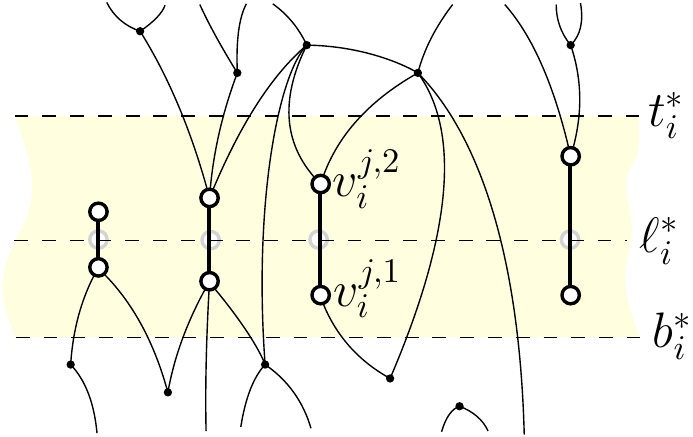}
			\label{fi:sameY2}
			\subcaption{}
		\end{subfigure}
		\caption{Illustration for the construction of $\langle G', H', \Gamma_{H'} \rangle$. The portions of $\Gamma_{H^*}$ (a) and of $\Gamma_{H'}$ (b) in the proximity of the strip $\mathcal S^*_i$ are shown; $\mathcal S^*_i$ is colored light yellow. The white disks represent vertices of $H^*$ and $H'$, while the black disks represent vertices in $V(G^*)\setminus V(H^*)$ and  $V(G')\setminus V(H')$.}
		\label{fig:no-two-vertices-sameY}
	\end{figure}

	We consider horizontal strips $\mathcal S^*_1,\dots,\mathcal S^*_m$ such that: (i) for $i=1,\dots,m$, the line $\ell^*_i$ is in the interior of $\mathcal S^*_i$, at equal distance from the horizontal lines $t^*_i$ and $b^*_i$ delimiting $\mathcal S^*_i$ from above and from below, respectively; and (ii) for $i=1,\dots,m-1$, the line $t^*_{i}$ lies below $b^*_{i+1}$. Observe that condition (ii) can be achieved by letting the strips $\mathcal S^*_1,\dots,\mathcal S^*_m$~be~sufficiently~thin. 
	
	We now replace each vertex in $\mathcal V^*_i$ by an edge. This is done as follows. Let $h^*_i$ be the height of $\mathcal S^*_i$. Recall that the vertices lying on $\ell^*_i$ are ordered in $\mathcal V^*_i$ by increasing $x$-coordinate; let $v_i^1,\dots,v_i^{|V^*_i|}$ be such an ordering. For $j=1,\dots,|V^*_i|$, we draw a vertical straight-line segment $s_i^j$ whose midpoint is $v_i^j$ and whose length is $j\cdot h^*_i/(3\cdot |V^*_i|)$. Note that $s_i^j$ lies entirely in $\mathcal S^*_i$ and crosses $\ell^*_i$. 
	
	We replace $v_i^j$ by two vertices $v_i^{j,1}$ and $v_i^{j,2}$ and by an edge $(v_i^{j,1},v_i^{j,2})$, both in $G'$ and in $H'$. The outgoing edges of $v_i^j$ in $G'$ now exit $v_i^{j,2}$ and the incoming edges of $v_i^j$ in $G'$ now enter $v_i^{j,1}$. If $G^*$ has a prescribed upward embedding, then the list of successors of the adjacent predecessors of $v_i^j$ and the list of predecessors of the adjacent successors of $v_i^j$ are updated by replacing $v_i^j$ with $v_i^{j,1}$ and $v_i^{j,2}$, respectively. Further, $\mathcal P(v_i^{j,1})$ coincides with $\mathcal P(v_i^j)$, while $\mathcal S(v_i^{j,1})=[v_i^{j,2}]$, and $\mathcal S(v_i^{j,2})$ coincides with $\mathcal S(v_i^j)$, while $\mathcal P(v_i^{j,2})=[v_i^{j,1}]$. The vertex $v_i^j$ is also removed from $\Gamma_{H'}$. Finally, the edge $(v_i^{j,1},v_i^{j,2})$ is represented in $\Gamma_{H'}$ by the segment $s_i^j$. This concludes the construction of $\langle G', H', \Gamma_{H'} \rangle$. 
	
	Vertices of $H'$ lying in different strips $\mathcal S^*_i$ and $\mathcal S^*_k$ do not share their $y$-coordinates in $\Gamma_{H'}$, as such strips are horizontal and disjoint. By construction vertices of $H'$ lying in the same strip $\mathcal S^*_i$ do not share their $y$-coordinates in $\Gamma_{H'}$. Thus, Property~(i) is satisfied. If $V(H^*)=V(G^*)$, then $V(H')=V(G')$, as every vertex in $H^*$ is replaced by two vertices belonging to $H'$. Hence Property~(ii) is also satisfied.
	
	The size of $\langle G', H', \Gamma_{H'} \rangle$ is linear in the size of $\langle G^*, H^*, \Gamma_{H^*} \rangle$ as at most two vertices and one edge are introduced in $G'$ for each vertex of $G^*$. Constructing the order $\mathcal V^*$ takes $O(n\log n)$ time. Every other step of the algorithm can be performed in $O(n)$ time, hence the total $O(n\log n)$ running time.
	
	We prove the equivalence between the instances $\langle G^*, H^*, \Gamma_{H^*} \rangle$ and $\langle G', H', \Gamma_{H'} \rangle$. 
	
	Suppose that an upward planar drawing $\Gamma_{G'}$ of $G'$ exists that extends $\Gamma_{H'}$. For each vertex $v_i^j$ of $G^*$ that has been replaced in $G'$ by two vertices $v_i^{j,1}$ and $v_i^{j,2}$ and by an edge $(v_i^{j,1},v_i^{j,2})$, we perform the following modification; see~\cref{fi:sameY3}. By construction and since $\Gamma_{G'}$ extends $\Gamma_{H'}$, the straight-line segment $(v_i^{j,1},v_i^{j,2})$ in $\Gamma_{G'}$ contains the point at which $v_i^j$ is placed in $\Gamma_{H^*}$. Then a sufficiently small value $\varepsilon>0$ can be defined so that the region $B_i^j$ obtained as the Minkowski sum of a ball with radius $\varepsilon$ with the straight-line segment $s_i^j$ intersects $\Gamma_{G'}$ only in a set of line segments incident to $v_i^{j,1}$ and $v_i^{j,2}$. For each vertex $v_i^j$ of $H^*$ we delete from $\Gamma_{G'}$ the interior of $B_i^j$, as in~\cref{fi:sameY4}, and we draw line segments connecting $v_i^j$ with the endpoints of the removed line segments on the boundary of $B_i^j$, as in~\cref{fi:sameY5}. This results in the desired drawing of $G^*$ extending $\Gamma_{H^*}$.
	
	\begin{figure}[htb]
		\centering
		\begin{subfigure}{.3\textwidth}
			\centering
			\includegraphics[scale=0.7]{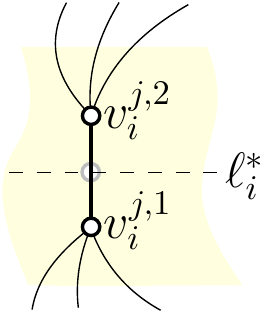}
			\subcaption{}\label{fi:sameY3}
		\end{subfigure}
		\begin{subfigure}{.3\textwidth}
			\centering
			\includegraphics[scale=0.7]{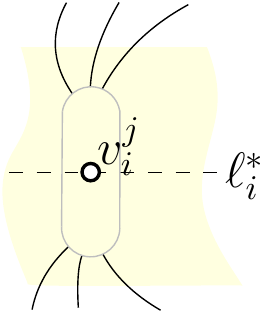}
			\subcaption{}\label{fi:sameY4}
		\end{subfigure}
		\begin{subfigure}{.3\textwidth}
			\centering
			\includegraphics[scale=0.7]{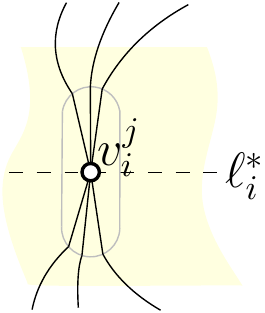}
			\subcaption{}\label{fi:sameY5}
		\end{subfigure}
		\caption{Construction of $\Gamma_{G^*}$ from $\Gamma_{G'}$. Modifications in the proximity of an edge $(v_i^{j,1},v_i^{j,2})$ of $G'$.}
		\label{fig:no-two-vertices-sameY-2}
	\end{figure}

	\begin{figure}[!b]
		\centering
		\begin{subfigure}{.3\textwidth}
			\centering
			\includegraphics[scale=0.7]{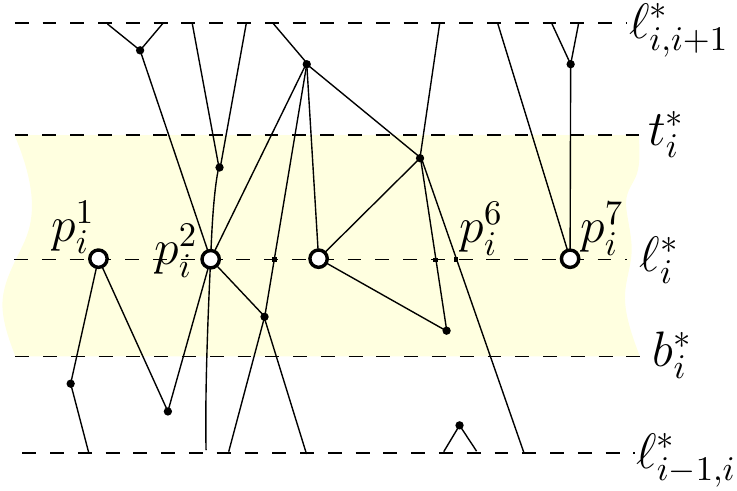}
			\subcaption{}
			\label{fi:sameY6}
		\end{subfigure}
		\hfil
		\begin{subfigure}{.3\textwidth}
			\centering
			\includegraphics[scale=0.7]{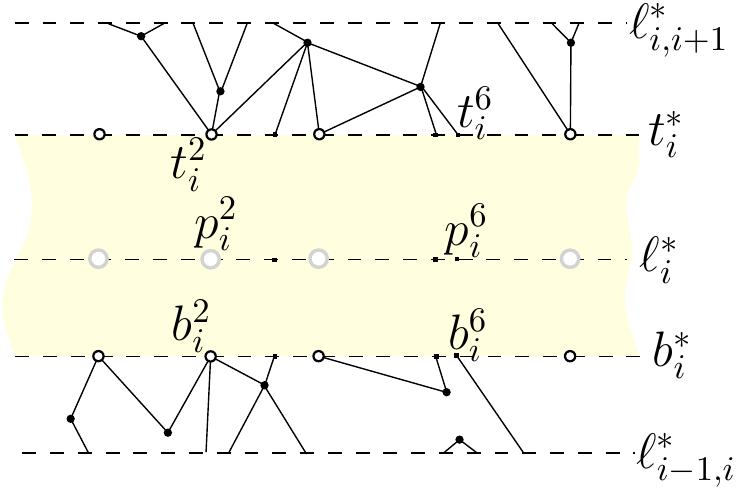}
			\subcaption{}
			\label{fi:sameY7}
		\end{subfigure}\\
		\begin{subfigure}{.3\textwidth}
			\centering
			\includegraphics[scale=0.7]{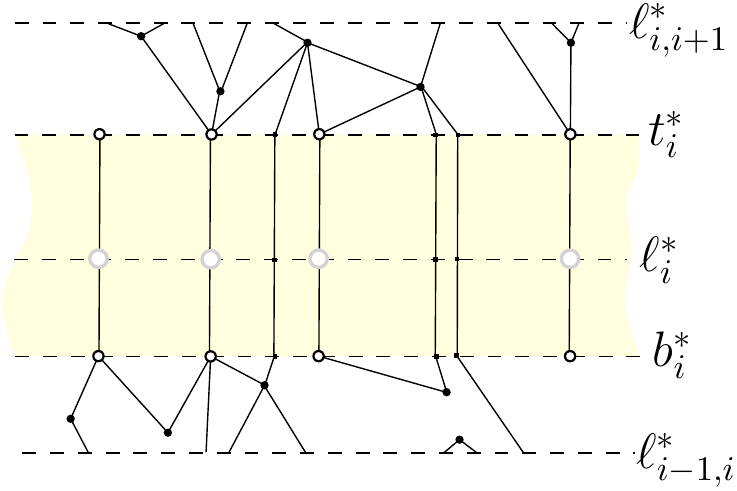}
			\subcaption{}
			\label{fi:sameY8}
		\end{subfigure}
		\hfil
		\begin{subfigure}{.3\textwidth}
			\centering
			\includegraphics[scale=0.7]{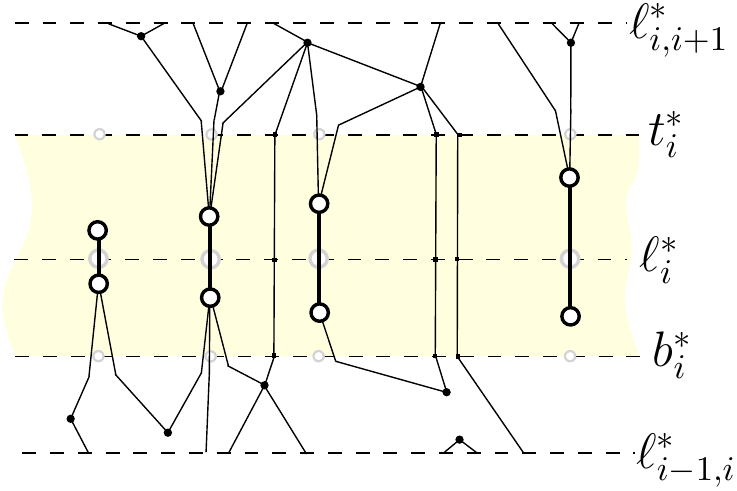}
			\subcaption{}
			\label{fi:sameY9}
		\end{subfigure}
		\caption{Construction of $\Gamma_{G'}$ from $\Gamma_{G^*}$. Modifications in the proximity of a strip $\mathcal S^*_i$, which is colored light yellow. (a) The drawing $\Gamma_{G^*}$ after the perturbation which moves vertices not in $H^*$ away from $\ell^*_{i}$. (b) Scaling down and translating the parts of $\Gamma_{G^*}$ between $\ell^*_{i}$ and $\ell^*_{i,i+1}$ and between $\ell^*_{i-1,i}$ and $\ell^*_{i}$. (c) Drawing the segments $|^h_i$. (d) Reconnecting the edges formerly incident to $v^j_i$ to the vertices $v^{j,1}_i$ and $v^{j,2}_i$.}
		\label{fig:no-two-vertices-sameY-3}
	\end{figure}
	Conversely, suppose that an upward planar drawing $\Gamma_{G^*}$ of $G^*$ exists that extends $\Gamma_{H^*}$.  For each $i=1,\dots,m$, we modify $\Gamma_{G^*}$ as follows. First, we vertically perturb the position of each vertex that lies on $\ell^*_{i}$ and that does not belong to $H^*$. If the perturbation is small enough, then the curves representing the edges incident to that vertex can also be suitably modified so that $\Gamma_{G^*}$ remains an upward planar drawing of $G^*$. Refer to~\cref{fi:sameY6} and denote by $p^1_i,\dots,p_i^{k(i)}$ the left-to-right order of the points of $\Gamma_{G^*}$ along $\ell^*_{i}$; note that each of such points might correspond to a vertex of $H^*$ or to a crossing point of an edge of $G^*$ with $\ell^*_{i}$. Let $\ell^*_{i,i+1}$ be a horizontal line above $t^*_i$ and below $b^*_{i+1}$. As in~\cref{fi:sameY7}, we vertically scale down the part of $\Gamma_{G^*}$ between $\ell^*_{i}$ and $\ell^*_{i,i+1}$ so that its vertical extension becomes equal to the distance between $t^*_{i}$ and $\ell^*_{i,i+1}$. Further, we translate the scaled part of the drawing so that it lies between $t^*_{i}$ and $\ell^*_{i,i+1}$; after doing so, the points $p^1_i,\dots,p_i^{k(i)}$ of $\Gamma_{G^*}$ on $\ell^*_{i}$ acquire corresponding points $t^1_i,\dots,t_i^{k(i)}$ on $t^*_{i}$. We perform an analogous modification to the part of $\Gamma_{G^*}$ comprised between $\ell^*_{i}$ and a horizontal line $\ell^*_{i-1,i}$ above $t^*_{i-1}$ and below $b^*_i$; denote by $b^1_i,\dots,b_i^{k(i)}$ the points on $b^*_{i}$ corresponding to $p^1_i,\dots,p_i^{k(i)}$, respectively. Now $\mathcal S^*_i$ does not contain any part of $\Gamma_{G^*}$ in its interior, except for the points $p^1_i,\dots,p^{k(i)}_i$. For every $h=1,\dots,k(i)$, we draw a vertical straight-line segment $|^h_i$ connecting $t^h_i$ with $b^h_i$, as in~\cref{fi:sameY8}; note that $|^h_i$ passes through $p^h_i$. If $p^h_i$ used to be a crossing point between an edge of $G^*$ and $\ell^*_{i}$ in $\Gamma_{G^*}$, then $|^h_i$ reconnects that edge; indeed, the latter became disconnected after scaling and translating parts of $\Gamma_{G^*}$. Otherwise, $p^h_i$ is the point in which a vertex $v_i^j$ was placed in $\Gamma_{G^*}$ and $\Gamma_{H^*}$. Then $s_i^j$ is a part of $|^h_i$; see~\cref{fi:sameY9}. We remove a small disk around $t^h_i$ from $\Gamma_{G^*}$ and extend the edges outgoing $v_i^j$ down along $|^h_i$ until they reach $v_i^{j,2}$. Analogously, we remove a small disk around $b^h_i$ from $\Gamma_{G^*}$ and extend the edges incoming $v_i^j$ up along $|^h_i$ until they reach $v_i^{j,1}$. This completes the construction of the desired drawing $\Gamma_{G'}$ of $G'$ extending $\Gamma_{H'}$.
\end{proof}





\section{Complexity of the {\sc UPE} and {\sc UPE-FUE} Problems} \label{se:complexity}

In this section we study the complexity of the {\sc UPE} and {\sc UPE-FUE} problems. We show the NP-hardness of such problems even for instances $\langle G, H, \Gamma_H \rangle$ in which $H$ contains no edges and $V(H)=V(G)$, and for instances $\langle G, H, \Gamma_H \rangle$ in which no two vertices share the same $y$-coordinate in $\Gamma_H$ and $V(H)=V(G)$. We will then show that the instances $\langle G, H, \Gamma_H \rangle$ in which $H$ contains no edges, no two vertices share the same $y$-coordinate in $\Gamma_H$, and $V(H)=V(G)$ can be solved in polynomial time. 

We start with the following.


\begin{lemma} \label{le:np}
	The {\sc UPE} and {\sc UPE-FUE} problems are in NP.
\end{lemma}

\begin{proof}
	Let $\langle G, H, \Gamma_H \rangle$ be an instance of the {\sc UPE} problem. Ideally, we would like to guess a solution for $\langle G, H, \Gamma_H \rangle$, that is, a drawing $\Gamma_G$ of $G$, and then to verify in polynomial time whether the guess actually is a solution, that is, whether $\Gamma_G$ is an upward planar drawing of $G$ extending $\Gamma_H$. However, since there is an infinite number of drawings of $G$, we cannot explicitly guess one of them. Hence, we proceed by associating to each drawing a combinatorial structure, in such a way that: (1) the set of distinct combinatorial structures the drawings of $G$ are associated to is finite; and (2) it is possible to test in polynomial time whether an upward planar drawing of $G$ extending $\Gamma_H$ exists by assuming that such a drawing is associated with a fixed combinatorial structure. This is done in the following.
	
	First, we guess the order by increasing $y$-coordinates of the vertices of $G$ in an upward planar drawing $\Gamma_G$ of $G$ extending $\Gamma_H$. This can be done by assigning a number in $\{1,2,\dots,n\}$ to each vertex of $G$. The assignment is, in general, not injective, i.e., the guess might result in several vertices having the same $y$-coordinate. Let $Y$ be the guessed order of the vertices of $G$. Denote by $V_1,\dots,V_k$ the subsets of $V(G)$ such that: (i) for any $u,v\in V_i$, the $y$-coordinates of $u$ and $v$ are the same in $Y$; and (ii) for any $u\in V_i$ and $v\in V_{i+1}$, the $y$-coordinate of $u$ precedes the one of $v$ in $Y$. Note that there are only finitely many such orders $Y$. 
	
	Having guessed $Y$, we now know whether the horizontal line $h_i$ through the vertices in $V_i$ crosses an edge $(u,v)$ in $\Gamma_G$. Indeed, this happens if and only if $u\in V_j$ and $v\in V_l$, with $j<i$ and $l>i$, or with $j>i$ and $l<i$. For each $i=1,\dots,k$, we independently guess a left-to order $\sigma_i$ of the vertices in $V_i$ together with the edges crossing $h_i$ in $\Gamma_G$. Again, there are only finitely many such orders.  
	
	We now show how to verify in polynomial time whether an upward planar drawing $\Gamma_G$ of $G$ exists that extends $\Gamma_H$ and that respects $Y$ and $\sigma_1,\dots,\sigma_k$. We perform the checks described below. If any check fails, then we conclude that there is no upward planar drawing of $G$ that extends $\Gamma_H$ and that respects $Y$ and $\sigma_1,\dots,\sigma_k$, otherwise we proceed to the next check. If all the checks succeed, then an upward planar drawing $\Gamma_G$ of $G$ exists that extends $\Gamma_H$ and that respects $Y$ and $\sigma_1,\dots,\sigma_k$. 
	
	\begin{enumerate} 
		\item For any edge $(u,v)$ of $G$ such that $u\in V_i$ and $v \in V_j$, we check whether $i<j$.
		\item We say that an edge $(u,v)$ of $G$, with $u\in V_j$ and $v\in V_l$, {\em spans} $V_{i}$ and $V_{i+1}$ if $j\leq i$ and $l\geq i+1$. For any two edges $e$ and $e'$ that span $V_{i}$ and $V_{i+1}$, we check whether the order of $e$ and $e'$ (or of their end-vertices) in $\sigma_i$ and $\sigma_{i+1}$ is the same.
		\item For any two vertices $u,v\in V(H)$ such that $u\in V_i$ and $v \in V_j$ with $i<j$, we check whether \mbox{$y(u)<y(v)$ in $\Gamma_H$.}
		\item For any two vertices $u,v \in V(H)\cap V_i$, we check whether $y(u)=y(v)$ in $\Gamma_H$.
		\item For any two vertices $u,v\in V(H)\cap V_i$ such that $y(u)=y(v)$ in $\Gamma_H$ and such that $u$ precedes $v$ in $\sigma_i$, we check whether $x(u)<x(v)$ in $\Gamma_H$.
		\item For any vertex $u\in V(H)\cap V_i$ and for any edge $e\in E(H)$ in $\sigma_i$, we check whether $e$ precedes $u$ in $\sigma_i$ if and and only if the crossing point between $e$ and the horizontal line through $u$ is \mbox{to the left of $u$ in $\Gamma_H$.}
		\item For any vertex $u$ in $V(H)\cap V_i$ and any two edges $e,e'\in E(H)$ in $\sigma_i$, we check whether $e$ precedes $e'$ in $\sigma_i$ if and only if the crossing point of $e$ with the horizontal line through $u$ is to the left of the crossing point of $e'$ with the horizontal line through $u$ in $\Gamma_H$. 
	\end{enumerate}
	
	If check 1.\ fails that any drawing of $G$ that respects $Y$ is not upward. If check 2.\ fails that any upward drawing of $G$ that respects $Y$, $\sigma_i$, and $\sigma_{i+1}$ is not planar. If checks 3.\ or 4.\ fail, then any upward drawing of $G$ that respects $Y$ does not extend $\Gamma_H$. If checks 5., 6., or 7. fail, then any drawing of $G$ that respects $Y$ and $\sigma_i$ does not extend $\Gamma_H$. 
	
	On the other hand, if all the checks succeed, then we can construct an upward planar drawing $\Gamma_G$ of $G$ that extends $\Gamma_H$ and that respects $Y$ and $\sigma_1,\dots,\sigma_k$ as follows. First, for $i=1,\dots,k$, we draw a horizontal line $h_i$ with equation $y=y_i$ on which the vertices in $V_i$ are going to be placed, where $y_1<y_2<\dots<y_k$. This is done so that $h_i$ passes through a vertex $u$ if $u\in V(H) \cap V_i$, which is possible since checks 3.\ and 4.\ succeed. For $i=1,\dots,k$, we place the vertices in $V_i$ that are not in $V(H)$ along $h_i$ and we fix the crossing points between $h_i$ and the edges of $G$ that are in $\sigma_i$ but not in $H$, so that the left-to-right order of the vertices of $V_i$ (including those in $V(H)$) and of the edges of $G$ (including those in $E(H)$) along $h_i$ is $\sigma_i$. This is possible because checks 5., 6., and 7.\ succeed. Each edge $(u,v)$ in $E(G)\setminus E(H)$ is now drawn as a sequence of line segments, each connecting points on two consecutive lines $h_i$ and $h_{i+1}$. Since check 2.\ succeeds, every such a line segment can be drawn in the horizontal strip delimited by $h_i$ and $h_{i+1}$ without introducing any crossings, thus the resulting drawing $\Gamma_G$ of $G$ is planar. Since check 1.\ succeeds, it follows that $\Gamma_G$ is also upward. The drawing $\Gamma_H$ has not been modified during the construction of $\Gamma_G$, hence $\Gamma_G$ extends $\Gamma_H$. By construction, for each $i=1,\dots,k-1$, the vertices in $V_i$ have a smaller $y$-coordinate than the vertices in $V_{i+1}$, hence $\Gamma_G$ respects $Y$. Finally, by construction, for each $i=1,\dots,k$, the left-to-right order of the vertices in $V_i$ and of the edges of $G$ crossing $h_i$ is $\sigma_i$, hence $\Gamma_G$ respects $\sigma_1,\dots,\sigma_k$.
	
	Since all the above checks can be easily performed in polynomial time, we have that {\sc UPE} is in NP. The membership in NP of the {\sc UPE-FUE} problem can be proved analogously. In particular, a check has to be introduced in the above algorithm in order to verify whether the orders $Y$ and $\sigma_1,\dots,\sigma_k$ define an upward embedding which is the one prescribed~in~the~instance.
\end{proof}

We now prove the NP-hardness results. The NP-hardness of the {\sc UPE} problem directly follows from the one of the {\sc Upward Planarity Testing} problem~\cite{GargT01}, which coincides with the special case of the {\sc UPE} problem in which the partial graph is the empty graph. However, we can establish the NP-hardness of the {\sc UPE} problem even in a much more constrained scenario. In order to do that, we relate the complexity of the {\sc UPE} problem to the one of a problem called {\sc Ordered Level Planarity}~\cite{KlemzR17} ({\sc OLP}, for short), which is defined as follows. 

A \emph{level graph} is a pair $(G,\ell)$, where $G$ is a directed graph and $\ell: V(G) \rightarrow \{1,\dots,k\}$ is a function that assigns each vertex to one of $k$ levels in such a way that $\ell(u) < \ell(v)$, for each $(u,v) \in E(G)$. We denote $V_i= \{v \in V(G): \ell(v)=i\}$.
A level graph is \emph{proper} if $\ell(v)-\ell(u)=1$, for each $(u,v) \in E(G)$.
A \emph{level drawing} of $(G,\ell)$ is a drawing of $G$ that maps each vertex to a point on the horizontal line with equation $y=\ell(v)$ and each edge $(u,v)$ to a $y$-monotone curve connecting the points corresponding to $u$ and $v$.
A level drawing is \emph{level planar} if it is crossing-free and a level graph is \emph{level planar} if it admits a level planar drawing.

An \emph{ordered level graph} is a triple $(G, \ell, \xi)$ where $(G,\ell)$ is a level graph and $\xi$ is a set of total orders  $\xi_i : V_i \rightarrow \{1, \dots, |V_i|\}$ of the vertices in $V_i$. Given an ordered level graph $(G, \ell, \xi)$, the {\sc OLP} problem asks whether there exists an \emph{ordered level planar drawing} of $(G, \ell, \xi)$, that is, a level planar drawing of $(G,\ell)$ in which, for every $v \in V_i$, the $x$-coordinate of $v$ is $\xi_i(v)$.
As observed by Klemz and Rote~\cite{KlemzR17}, for an ordered level graph $(G, \ell, \xi)$, the $y$- and $x$-coordinates assigned via $\ell$ and $\xi$, respectively, only serve the purpose of encoding a partial order among vertices of different levels and a total order among vertices of the same level, respectively. In particular, the {\sc OLP} problem does not change its complexity if $\xi$ and $\ell$ map to the reals. We exploit this fact in the proof of the following lemma.

\begin{lemma} \label{th:equivalence} 
The following statements hold true:	
	\begin{enumerate}[\bf (i)]
		\item Let $(G, \ell, \xi)$ be an instance of the {\sc OLP} problem, where $G$ is an $n$-vertex graph. It is possible to construct in $O(n)$ time an equivalent instance $\langle G,H,\Gamma_H \rangle$ of the {\sc UPE} problem with $H=(V(G),\emptyset)$ such that if $(G, \ell, \xi)$ contains at most $\lambda$ vertices belonging to the same level, then $\Gamma_H$ contains at most $\lambda$ vertices sharing the same $y$-coordinate.
		\item Let $\langle G,H,\Gamma_H \rangle$ be an instance of the {\sc UPE} problem, where $G$ is an $n$-vertex graph and $H=(V(G),\emptyset)$. It is possible to construct in $O(n \log n)$ time an equivalent instance $(G, \ell, \xi)$ of the {\sc OLP} problem such that if $\Gamma_H$ contains at most $\lambda$ vertices sharing the same $y$-coordinate, then $(G, \ell, \xi)$ contains at most $\lambda$ \mbox{vertices belonging to the same level.}
	\end{enumerate}
\end{lemma}

\begin{proof}
{\bf (i)} Given an ordered level graph $(G, \ell, \xi)$ with $k$ levels, where $G$ is an $n$-vertex graph, we construct an instance $\langle G,H,\Gamma_H \rangle$ of the {\sc UPE} problem as follows. First, $G$ coincides in the two instances. Second, we define $H=(V(G),\emptyset)$, as required. Third, for every $i \in \{1,\dots,k\}$ and for every $v\in V_i$, we define the position of $v$ in $\Gamma_H$ as $(x(v),y(v))=(\xi_i(v),\ell(v))$. This completes the construction of $\langle G,H,\Gamma_H \rangle$; this construction can be clearly performed in $O(n)$ time. By construction, if $(G, \ell, \xi)$ contains at most $\lambda$ vertices belonging to the same level, then $\Gamma_H$ contains at most $\lambda$ vertices sharing the same $y$-coordinate.

We claim that the instance $(G, \ell, \xi)$ of the {\sc OLP} problem and the instance $\langle G,H,\Gamma_H \rangle$ of the {\sc UPE} problem are equivalent. Specifically, we have that any ordered level planar drawing $\Gamma_o$ of $(G, \ell, \xi)$ is also an upward planar drawing of $G$ that extends $\Gamma_H$ and that any upward planar drawing $\Gamma_G$ of $G$ that extends $\Gamma_H$ is also an ordered level planar drawing of $(G, \ell, \xi)$. In fact, it holds that (i) each vertex of $G$ has the same coordinates both in $\Gamma_o$ and in $\Gamma_G$, and that (ii) the edges of $G$ are drawn as $y$-monotone curves both in $\Gamma_o$ and in $\Gamma_G$.

{\bf (ii)} Given an instance $\langle G,H,\Gamma_H \rangle$ of the {\sc UPE} problem, where $G$ is an $n$-vertex graph and $H=(V(G),\emptyset)$, we construct an ordered level graph $(G, \ell, \xi)$ as follows. First, the graph $G$ coincides in the two instances. Second, we define a partition $V_1,V_2,\dots,V_k$ of $V(G)$ such that, for each $i=1,\dots,k$, all the vertices in $V_i$ have the same $y$-coordinate in $\Gamma_H$, and such that, for each $i=1,\dots,k-1$, the vertices in $V_i$ have a $y$-coordinate smaller than the one of the vertices in $V_{i+1}$ in $\Gamma_H$. This partition can be constructed in $O(n)$ time after the vertices of $G$ have been ordered by their $y$-coordinates in $\Gamma_H$, which takes $O(n \log n)$ time. For each vertex $v$ of $G$, we now assign $\ell(v)=i$ if and only $v\in V_i$. Finally, for each $i=1,\dots,k$, we compute the order $\xi_i$ of the vertices in $V_i$ by increasing $x$-coordinates. Such an order can be constructed in $O(|V_i| \log |V_i|)$ time, hence in $O(n \log n)$ time over all sets $V_1,V_2,\dots,V_k$. We claim that the instance $\langle G,H,\Gamma_H \rangle$ of the {\sc UPE} problem and the instance $(G, \ell, \xi)$ of the {\sc OLP} problem are equivalent. 

Suppose that an upward planar drawing $\Gamma_G$ of $G$ exists that extends $\Gamma_H$. For each $i=1,\dots,k$, we draw a horizontal line $h_i$ through the vertices in $V_i$ in $\Gamma_G$; let $X^*_i=(x^i_1,x^i_2,\dots,x^i_{r_i})$ be the left-to-right order in which the vertices in $V_i$ and the edges of $G$ crossing $h_i$ appear along $h_i$ in $\Gamma_G$. We construct a level planar drawing $\Gamma_L$ of $(G, \ell, \xi)$ by placing, for each $i=1,\dots,k$, a sequence $P_i=(p^i_1,p^i_2,\dots,p^i_{r_i})$ of $r_i$ points along the line $y=i$, such that:

\begin{itemize}
	\item for each $j=1,\dots,r_i-1$, the $x$-coordinate of $p^i_j$ is smaller than the one of $p^i_{j+1}$; and
	\item for each $j=1,\dots,r_i$, if $x^i_j$ corresponds to a vertex $v$ of $G$, then the $x$-coordinate of $p^i_{j}$ \mbox{is equal to $\xi_i(v)$.}
\end{itemize} 

The level planar drawing $\Gamma_L$ of $(G, \ell, \xi)$ is completed by representing each edge $e=(u,v)$ of $G$ as a sequence of straight-line segments. Specifically, if $u\in V_i$ and $v\in V_l$, then $e$ starts at the point $p^i_{f(i,u)}$ of $P_i$ such that $x^i_{f(i,u)}$ corresponds to $u$, then proceeds with a straight-line segment to the point $p^{i+1}_{f(i+1,e)}$ of $P_{i+1}$ such that $x^{i+1}_{f(i+1,e)}$ corresponds to $e$, then proceeds with a straight-line segment to the point $p^{i+2}_{f(i+2,e)}$ of $P_{i+2}$ such that $x^{i+2}_{f(i+2,e)}$ corresponds to $e$, and so on until reaching the point $p^l_{f(l,v)}$ of $P_l$ such that $x^l_{f(l,v)}$ corresponds to $v$.  

The proof that an upward planar drawing $\Gamma_G$ of $G$ can be constructed from a level planar drawing $\Gamma_L$ of $(G, \ell, \xi)$ is analogous. In particular, $X^*_i$ now represents the left-to-right order in which the vertices of $V_i$ and the edges of $G$ crossing the line $y=i$ appear along such a line. This order is used so to define a sequence of points along the horizontal line through the vertices of $V_i$ in $\Gamma_H$. Such points are used to represent each edge as a sequence of straight-line segments.
\end{proof}
	
The first NP-hardness result is a direct consequence of reduction (i) from~\cref{th:equivalence} and of the NP-completeness of {\sc Ordered Level Planarity}~\cite{KlemzR17}.

\begin{theorem} \label{th:UPE-and-UPE-FE-hard}
	The {\sc UPE} problem is NP-complete even if 
	\begin{inparaenum}[(i)]
		\item the partial graph contains all the vertices and no edges and
		\item no three vertices share the same $y$-coordinate in the partial drawing.
	\end{inparaenum}
\end{theorem}

\begin{proof}
	The membership of {\sc UPE} in NP follows from~\cref{le:np}. Klemz and Rote proved the NP-hardness of the {\sc OLP} problem even for ordered level graphs $(G, \ell, \xi)$ such that $G$ consists of a set of disjoint paths and $\ell$ assigns at most $2$ vertices to any level~\cite{KlemzR17}. Thus, by applying reduction (i) of~\cref{th:equivalence} to an instance $(G, \ell, \xi)$ of the {\sc OLP} problem with the above properties, we obtain in linear time an equivalent instance $\langle G,H,\Gamma_H \rangle$ of the {\sc UPE} problem such that $H=(V(G),\emptyset)$ and such that no three vertices share the same $y$-coordinate in $\Gamma_H$. This completes the proof.
\end{proof}

\cref{le:no-vertices-sameY}, together with~\cref{th:UPE-and-UPE-FE-hard}, implies the following.

\begin{corollary} \label{cor:UPE-and-UPE-FE-hard}
	The {\sc UPE} problem is NP-complete even if 
	\begin{inparaenum}[(i)]
		\item the partial graph contains all the vertices and
		\item no two vertices share the same $y$-coordinate in the partial drawing.
	\end{inparaenum}
\end{corollary}

We now discuss the complexity of the {\sc UPE-FUE} problem, in which the input graph comes with a prescribed upward embedding which the required drawing has to respect. We establish the NP-hardness of the {\sc UPE-FUE} problem via a reduction from the {\sc Partial Level Planarity} (for short, {\sc PLP}) problem, recently introduced by Br\"uckner and Rutter~\cite{DBLP:conf/soda/BrucknerR17}. This reduction is quite involved and requires a new analysis of some results in~\cite{DBLP:conf/soda/BrucknerR17}.

Given a $4$-tuple $(G,\ell,H,\Gamma_H)$ where $(G,\ell)$ is a level planar graph, $H$ is a subgraph of $G$, and $\Gamma_H$ is a level planar drawing of $(H,\ell)$, the {\sc PLP} problem asks whether a level planar drawing $\Gamma_G$ of $(G,\ell)$ exists that coincides with $\Gamma_H$ when restricted to the vertices and edges of $H$. The instance $(G,\ell,H,\Gamma_H)$ is \emph{proper} if $(G,\ell)$ is a proper level graph.

A proper instance $(G=(V,E),\ell,H,\Gamma_H)$ of the {\sc PLP} problem can also be represented as a triple $(G,\ell,\prec)$, where $\prec = \{\prec_1,\dots,\prec_k\}$ is a set of total orders $\prec_i$ of $V_i \cap V(H)$ such that, for any two vertices $u,v \in V_i \cap V(H)$, we have that $u \prec_i v$ if and only if $u$ precedes $v$ along the horizontal line $\ell_i$ with equation $y=i$. Then, $(G,\ell)$ admits a level planar drawing that extends $\Gamma_H$ if and only if $(G,\ell,\prec)$ admits a level planar drawing in which the order of the vertices in $V_i$ along $\ell_i$ is a linear extension of $\prec_i$, for $i=1,\dots,k$.

Br\"uckner and Rutter~\cite{DBLP:conf/soda/BrucknerR17} proved that the {\sc PLP} problem is NP-complete even for instances that are proper and connected. Since the reduction they present does not exploit changes of the planar embedding in the produced instances, this allows them to further augment such instances to subdivisions of triconnected graphs, which have a unique planar embedding (up to a flip). Therefore, the {\sc PLP} problem is NP-complete even for proper instances in which the planar embedding of the level~graph~is~prescribed. 

\begin{figure}[htb]
	\centering
	{\includegraphics[scale=.85]{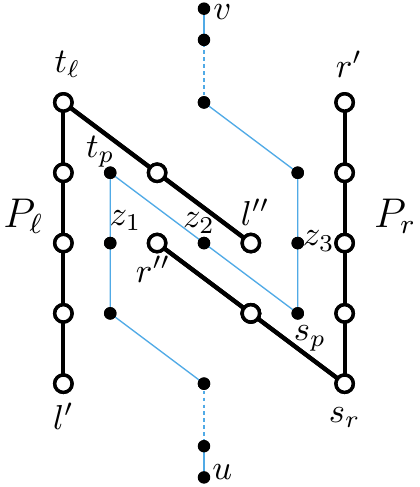}}
	\caption{A socket gadget (thick black edges) and a plug gadget (thin blue edges).}\label{fig:plug-socket}
\end{figure}

We claim that the hardness result of Br\"uckner and Rutter holds even when the upward embedding of the level graph is prescribed. 
The basic building blocks of their reduction are two gadgets: The \emph{plug} and the \emph{socket} gadgets, which are illustrated in~\cref{fig:plug-socket}.
The \emph{plug} gadget is simply a directed path $P=(u,\dots,z_1,t_p,z_2,s_p,z_3\dots,v)$ consisting of three monotone paths $(u,\dots,z_1,t_p)$, $(t_p,z_2,s_p)$, and $(s_p,z_3,\dots,v)$ that traverse several levels of the level graph. The vertices of the plugs are not subject to any ordering constraints, that is, there exists no ordering $\prec_i$ in $\prec$ such that $a \prec_i b$ or $b \prec_i a$ and $a,b \in V(P)$.
The \emph{socket} gadget consists of a \emph{left-boundary path} $P_\ell = (l',\dots, t_\ell,\dots, l'')$ consisting of two monotone paths $(l',\dots,t_\ell)$ and $(l'',\dots,t_\ell)$, and of a \emph{right-boundary path} $P_r = (r',\dots,s_r,\dots,r'')$, consisting of two monotone paths $(s_r,\dots,r')$ and $(s_r,\dots,r'')$, all traversing several levels of the level graph. The vertices in $V(P_\ell) \cup V(P_r)$ are totally ordered by the functions $\prec_i$, for each level $i$ traversed by such paths; in particular, the left-to-right order of the predecessors of $t_\ell$ and the left-to-right order of the successors of $s_r$ is prescribed. The crucial property exploited throughout the constructions presented in~\cite{DBLP:conf/soda/BrucknerR17} is that at most one plug can be placed between the left- and the right-boundary paths of a socket, in every level-planar drawing $\Gamma$ of the input level graph that realizes the prescribed orders of the vertices on each level. Our \cref{cl:PLP-hard} below then follows by observing that in $\Gamma$ 
the edges $(z_1,t_p)$ and $(z_2,t_p)$ appear in this left-to-right order and similarly for the edges $(s_p,z_2)$ and $(s_p,z_3)$. We thus have the following.

\begin{myclaim}\label{cl:PLP-hard}
	The {\sc PLP} problem is NP-complete even for proper instances $(G,\ell,\prec)$ such that $G$ is connected and has a prescribed upward embedding.
\end{myclaim}

We exploit \cref{cl:PLP-hard} in order to prove the following.
\begin{theorem}\label{th:UPE-FUE-hard}
	The {\sc UPE-FUE} problem is NP-complete. This result holds even if 
	\begin{inparaenum}[(i)]
		\item the instance is connected and
		\item the partial graph contains all the vertices and no edges.
	\end{inparaenum}
\end{theorem}

\begin{proof}
	The membership of the {\sc UPE-FUE} problem in NP follows from~\cref{le:np}.
	
	Let $(G,\ell,\prec)$ be a proper instance of the {\sc PLP} problem such that $(G,\ell)$ is a connected level graph with a prescribed upward embedding. We show how to construct in polynomial time an equivalent instance $\langle U,H,\Gamma_H\rangle$ of the {\sc UPE-FUE} problem satisfying properties (i) and (ii) of the statement. Refer to~\cref{fig:hardness}.

	\begin{figure}[htb]
		\centering
		{\includegraphics[page=4,width=\columnwidth]{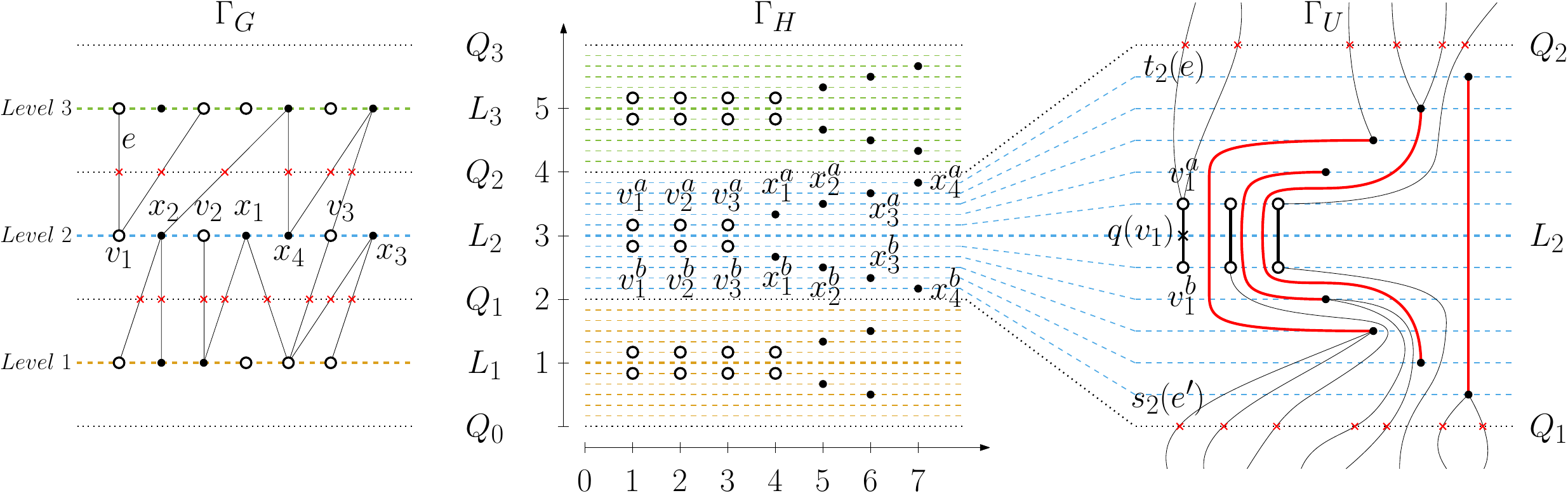}}
		\caption{(left) An instance $(G,\ell,\prec)$ of {\sc PLP} with a prescribed upward embedding. 
			(center) An instance $(U,H,\Gamma_H)$ of {\sc UPE-FUE} equivalent to $(G,\ell,\prec)$ constructed using the reduction of~\cref{th:UPE-FUE-hard}.
			(right) A drawing $\Gamma_U$ of $U$ that extends $\Gamma_H$, focused on the subgraph of $U$ stemming from the vertices of the third level of $(G,\ell)$ and from their incident edges.
		}\label{fig:hardness}
	\end{figure}
	
	The graph $U$ is defined as follows. For each vertex $v$ of $G$, the graph $U$ contains vertices $v^a$ and $v^b$ and an edge $(v^b,v^a)$. Further, for each edge $(u,v)$ in $G$, the edge set of $U$ contains an edge $(u^a,v^b)$. The graph $U$ is connected, since $G$ is connected, which implies~property~(i). 
	%
	We set the upward embedding of $U$ as follows. 
	For each vertex $v$ of $G$ such that $\mathcal P(v)=(u_1,u_2,\dots,u_k)$ and $\mathcal S(v)=(w_1,w_2,\dots,w_l)$, we set $\mathcal P(v^b)=[u^a_1,u^a_2,\dots,u^a_k]$, $\mathcal S(v^b)=[v^a]$, $\mathcal P(v^a)=[v^b]$, and $\mathcal S(v^a)=[w^b_1,w^b_2,\dots,w^b_l]$.

	The graph $H$ is defined as $H=(V(U),\emptyset)$, which implies property (ii).
	
	Finally, the drawing $\Gamma_H$ of $H$ is defined as follows. Since $H$ contains no edges, in order to construct $\Gamma_H$ we only need to assign a position to each vertex of $H$.
	
	For $i=1,\dots,k$, let $F_i \subseteq V_i$ be the subset of vertices of level $i$ of $(G,\ell)$ whose left-to-right order is fixed by $\prec_i$, and let $v_{i,1} \prec_i v_{i,2} \prec_i  \dots \prec_i v_{i,r_i}$ be the vertices of $F_i$. Also, let $x_{i,1},\dots,x_{i,d_i}$ be the vertices in $V_i \setminus F_i$ (whose left-to-right order is not prescribed by $\prec_i$).
	Let $F = \max^k_{i=1} |V_i \setminus F_i|$ and $\alpha = \frac{1}{F+2}$. 
	
	First, for $i=1,\dots,k$ and for $j=1,\dots,r_i$, we place the vertex $v^a_{i,j}$ at the point $(j, 2i-1 + \alpha)$ and the vertex $v^b_{i,j}$ at the point $(j, 2i-1 - \alpha)$.
	Second, for $q = 1,\dots,d_i$, we place the vertex $x^a_{i,q}$ at the point $(r_i + q, 2i-1 + (q+1)\alpha)$ and the vertex $x^b_{i,q}$ at the point $(r_i + q, 2i - 1 - (q+1)\alpha)$. 
	Observe that all the vertices of $H$ stemming from vertices in $V_i$ lie above the horizontal line $Q_{i-1} := y = 2i-2$ and below the horizontal line $Q_i:= y = 2i$.
	This concludes the construction of the instance $(U,H,\Gamma_H)$, which can be carried out in polynomial time in the size of $(G,\ell,\prec)$.

	We now prove that $(G,\ell,\prec)$ and $(U,H,\Gamma_H)$ are equivalent.

	For the first direction, suppose that the instance $(U,H,\Gamma_H)$ of the {\sc UPE-FUE} problem is positive, that is, there exists a drawing $\Gamma_U$ of $U$ that extends $\Gamma_H$ and respects its prescribed upward embedding. We show how to construct a level planar drawing $\Gamma_G$ of $(G,\ell)$ that respects its prescribed  upward embedding and such that the order of the vertices in $V_i$ along the horizontal line $L_i:= y = 2i-1$ is a linear extension of $\prec_i$, for $i=1,\dots,k$.
	
	We construct $\Gamma_G$ as follows; refer to~\cref{fi:UtoG}. For $i=1,\dots,k$ and for each vertex $v \in V_i$, let $q(v)$ be the crossing point between the edge $(v^b,v^a)$ and $L_i$ in $\Gamma_U$. Observe that such a crossing point exists since $v^b$ and $v^a$ lie below and above $L_i$, respectively, in $\Gamma_H$ (and thus in $\Gamma_U$), by construction. We place the vertex $v$ at $q(v)$ in~$\Gamma_G$. We modify the drawings of the edges of $U$ incident to $v^a$ and $v^b$ as follows. We delete the parts of $\Gamma_U$ inside disks with radius $\varepsilon>0$ centered at $v^a$ and $v^b$; if $\varepsilon$ is sufficiently small, only $v^a$, $v^b$, and parts of their incident edges are deleted from $\Gamma_U$. We extend the edges that used to be incident to $v^a$ and $v^b$  downwards and upwards, respectively, following the drawing of the edge $(v^b,v^a)$ in $\Gamma_U$ 	sufficiently close to it, until they reach $q(v)$. Finally, we delete the rest of the drawing of the edge $(v^b,v^a)$. Repeating this process for every vertex $v$ of $G$ results in a drawing $\Gamma_G$ of $G$.


	\begin{figure}[tb]
		\centering
		{\includegraphics[page=5,scale =0.4]{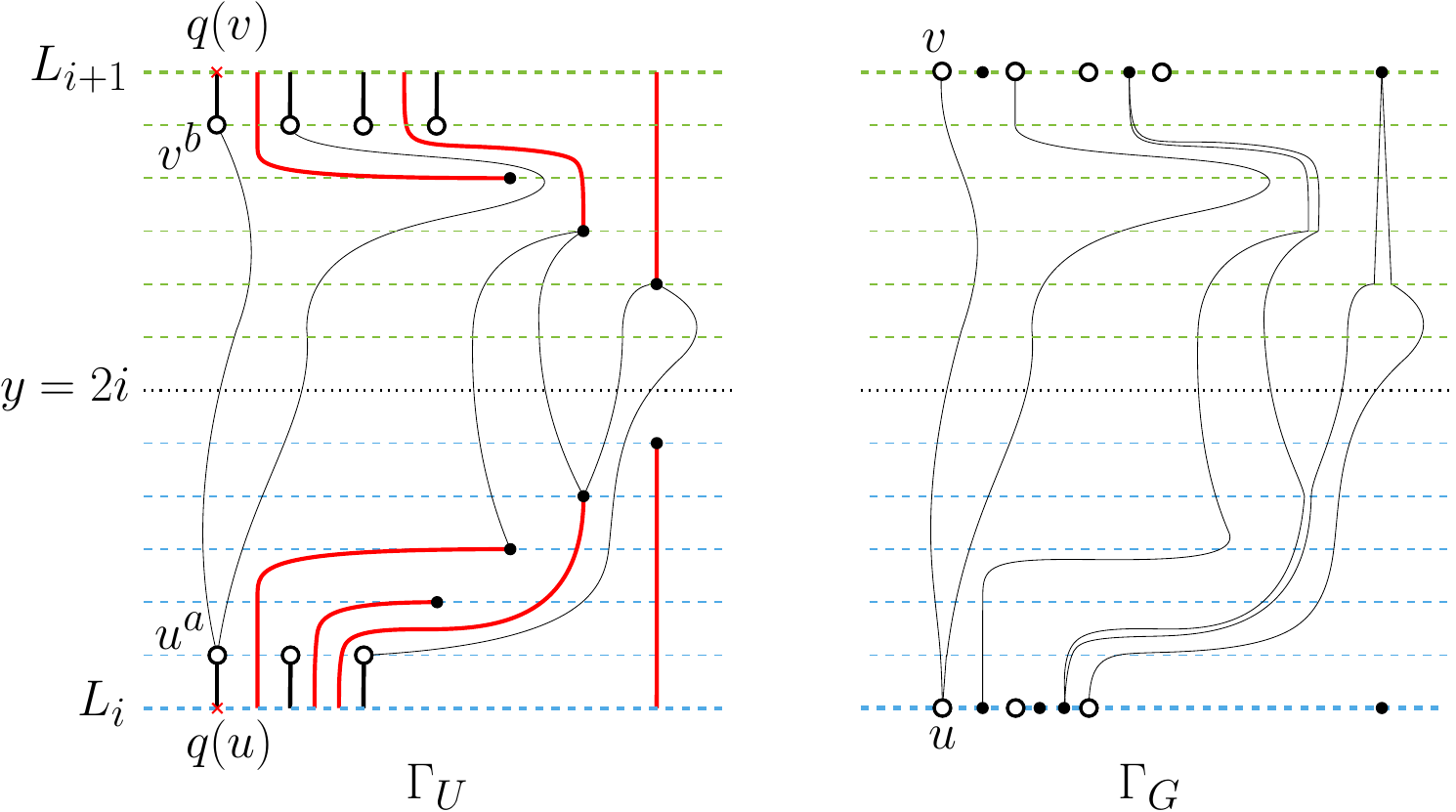}}
		\caption{Construction of the curves representing the edges of $G$ directed from level $i$ to level $i+1$.}\label{fi:UtoG}
	\end{figure}

%
	
	We have that $\Gamma_G$ is a level drawing. Namely, each edge $(u,v)$ of $G$ is represented in $\Gamma_G$ by a $y$-monotone curve, given that the edge $(u^a,v^b)$ of $U$ is represented in $\Gamma_U$ by a $y$-monotone curve and given that the parts of the representation of $(u,v)$ in $\Gamma_G$ that do not coincide with the representation of $(u^a,v^b)$ follow the $y$-monotone curves representing the edges $(u^b,u^a)$ and $(v^b,v^a)$. Further, by construction, for $i=1,\dots,k$, the vertices of the level $V_i$ lie on the horizontal line $L_i$ with equation $y=2i-1$. Hence, after a change of the coordinate system which maps any point $(x,y)\in \mathbb R^2$ to the point $(x,\frac{y+1}{2})$, we have that  $\Gamma_G$ is a level drawing. The drawing $\Gamma_G$ is also planar, given that $\Gamma_U$ is planar and given that, for each edge $(u,v)$ of $G$, the parts of the representation of $(u,v)$ in $\Gamma_G$ that do not coincide with the representation of $(u^a,v^b)$ in $\Gamma_U$ follow the crossing-free curves representing the edges $(u^b,u^a)$ and $(v^b,v^a)$.

	
	We now show that $\Gamma_G$ respects the prescribed upward embedding of $G$. Consider any vertex $v$ of $G$ and let $\mathcal S(v)=[w_1,w_2,\dots,w_l]$. By construction, the left-to-right order in $\Gamma_G$ of the edges outgoing at $v$ is $(v,w_1),(v,w_2),\dots,(v,w_l)$ if and only if the left-to-right order in $\Gamma_U$ of the edges outgoing at $v^a$ is $(v^a,w^b_1),(v^a,w^b_2),\dots,(v^a,w^b_l)$. Further, the left-to-right order in $\Gamma_U$ of the edges outgoing at $v^a$ is indeed $(v^a,w^b_1),(v^a,w^b_2),\dots,(v^a,w^b_l)$, given that $\mathcal S(v^a)=[w^b_1,w^b_2,\dots,w^b_l]$ and given that $\Gamma_U$ respects the prescribed upward embedding of $G$. Analogously, the left-to-right order of the edges incoming at each vertex $v$ of $G$ in $\Gamma_G$ corresponds to the prescribed order of the adjacent \mbox{predecessors of $v$.}

%
%
	Finally, we show that the ordering of the vertices in $V_i$ along $L_i$ is a linear extension of $\prec_i$, for $i=1,\dots,k$. Recall that each vertex $v \in V_i$ has been placed at the point $q(v)$ in $\Gamma_G$. Hence, it suffices to show that $q(v_{i,1}),q(v_{i,2}),\dots,q(v_{i,r_i})$ appear in this left-to-right order along $L_i$ in $\Gamma_U$, which follows from the fact that the vertices $v^a_{i,1},v^a_{i,2},\dots,v^a_{i,r_i}$ appear in this left-to-right order along the horizontal line $y= 2i-1 + \alpha$, that the vertices $v^b_{i,1},v^b_{i,2},\dots,v^b_{i,r_i}$ appear in this left-to-right order along the horizontal line $y= 2i-1 - \alpha$, and that the edges $(v^a_{i,j},v^b_{i,j})$ are drawn as non-crossing $y$-monotone curves in $\Gamma_U$. 
	
	
	For the second direction, suppose now that $(G,\ell,\prec)$ is a positive instance of the {\sc PLP} problem, that is, there exists a level planar drawing $\Gamma_G$ of $(G,\ell)$ that respects its prescribed upward embedding and such that the left-to-right order of the vertices in $V_i$ along the horizontal line $y=i$ is a linear extension of $\prec_i$, for $i=1,\dots,k$. We show how to construct an upward planar drawing $\Gamma_U$ of $U$ that extends $\Gamma_H$ and that respects its prescribed upward embedding.
	
	First, we change the coordinate system in $\Gamma_G$ so to map any point $(x,y)\in \mathbb R^2$ to the point $(x,2i-1)$. Now the vertices of $V_i$ lie on the horizontal line $L_i:= y = 2i-1$. For $i=1,\dots,k$ and for each vertex $v \in V_i$, let $q(v)$ be the point along $L_i$ the vertex $v$ lies upon in $\Gamma_G$. For $i=1,\dots,k$, let $T_i$ be the left-to-right order of the vertices in $V_i$ as they appear along $L_i$ in $\Gamma_G$. 
	The proof is based on the following claim.
	
	\begin{myclaim}\label{cl:order-vertices}
	There exists a drawing $\Gamma'$ of the edges $(v^b,v^a)$ of $U$ such that $\Gamma'$ extends $\Gamma_H$ and such that each edge $(v^b,v^a)$ intersects $L_i$ at~$q(v)$. 
	\end{myclaim}

\begin{figure}[htb]
	\centering
	{\includegraphics[page=6,height=.3\columnwidth]{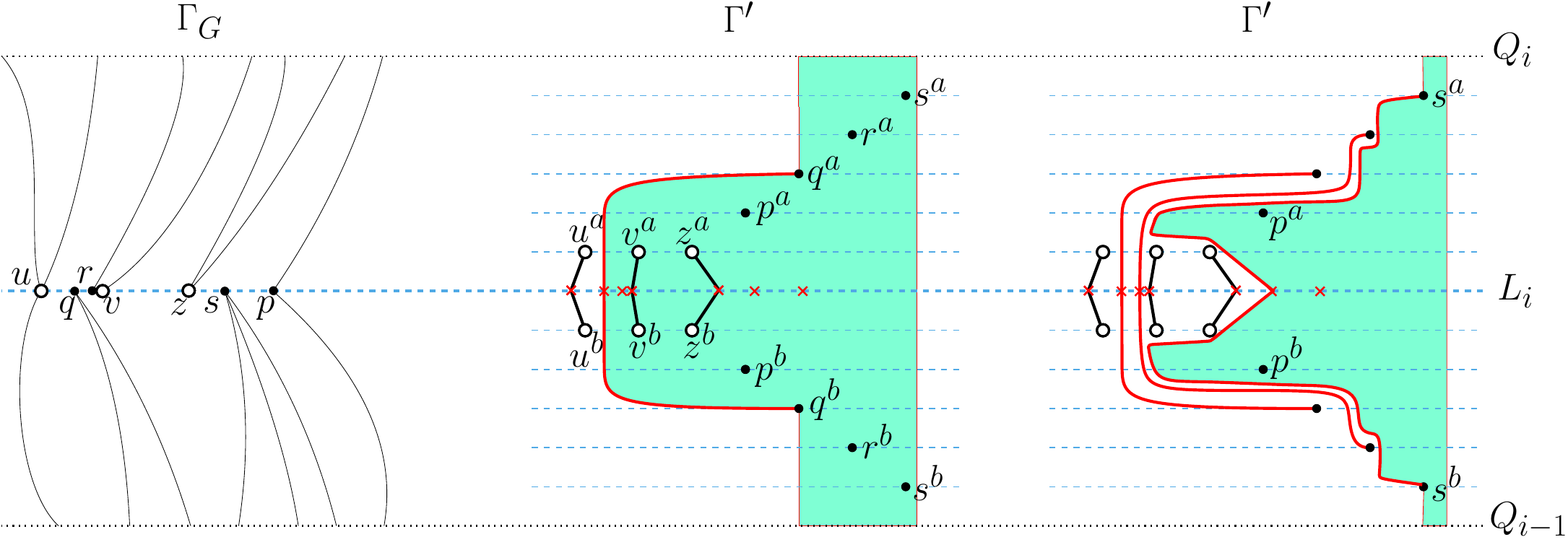}}
	\caption{Illustration for the proof of \cref{cl:order-vertices}. (left) A level planar drawing $\Gamma_G$ of $G$ focused on level $i$. (middle) and (right) $\Gamma'$ after drawing the edges $(q^b,q^a)$ and $(s^b,s^a)$, respectively. The corresponding regions $R$ are colored.}\label{fi:GtoU}
\end{figure}

	\begin{proof}
		First, we initialize $\Gamma'=\Gamma_H$. Then, for each $i=1,\dots,k$ and for each vertex $v$ in $F_i$, we draw the edge $(v^b,v^a)$ as a $y$-monotone curve composed of the straight-line segments $v^b q(v)$ and $q(v)v^a$. We prove that $\Gamma'$ is crossing-free. First, for any two vertices $u\in F_i$ and $v\in F_j$ with $i\neq j$, the edges $(u^b,u^a)$ and $(v^b,v^a)$ are vertically disjoint. Second, for any two vertices $u,v\in F_i$, the straight-line segments $u^b q(u)$ and $q(v)v^a$ are vertically disjoint, as well as the straight-line segments $v^b q(v)$ and $q(u)u^a$. Finally, the straight-line segments $u^b q(u)$ and $v^b q(v)$ do not cross since $q(u)$ and $q(v)$ have the same $y$-coordinate, since $u^b$ and $v^b$ have the same $y$-coordinate, and since the $x$-coordinate of $q(u)$ is smaller than the one of $q(v)$ if and only if the $x$-coordinate of $u^b$ is smaller than the one of $v^b$; similarly, the straight-line segments $u^a q(u)$ and $v^a q(v)$ do not cross.   
		
		In order to complete $\Gamma'$, it remains to draw each edge $(v^b,v^a)$, with $v \in V_i \setminus F_i$. To this aim, we process the vertices in $V_i \setminus F_i$ in the order they appear in $T_i$. When we process a vertex $v$, we draw the edge $(v^b,v^a)$ as a $y$-monotone curve passing trough $q(v)$, without introducing any crossings in $\Gamma'$; this is done so to maintain the invariant that in $\Gamma'$ there exists a region $R$ that satisfies the following properties (refer to~\cref{fi:GtoU}): 
		
		\begin{enumerate} [(i)]
			\item $R$ is delimited from below and from above by $Q_{i-1}$ and $Q_i$, respectively, and from the left and from the right by $y$-monotone curves $\lambda_R$ and $\rho_R$ that extend from $Q_{i-1}$ to $Q_i$, respectively;
			\item for every unprocessed vertex $u$, the region $R$ contains in its interior $u^b$, $u^a$, and $q(u)$; and
			\item the region $R$ does not contain in its interior any part of an edge $(u^b,u^a)$, for any processed vertex $u\in V_i\setminus F_i$.
		\end{enumerate}
		
		Let $v$ be the currently processed vertex. The edge $(v^b,v^a)$ is drawn as a $y$-monotone curve $\ell(v)$ lying inside $R$ and composed of eight parts. The first part $\ell_1(v)$ of $\ell(v)$ starts from $v^b$ and approaches $\lambda_R$ by moving left while slightly and monotonically increasing in the $y$-coordinate. The second part $\ell_2(v)$ of $\ell(v)$ follows $\lambda_R$, slightly to the right of it, until reaching a point slightly lower than the line $y=2i-1-\alpha$. The third part $\ell_3(v)$ of $\ell(v)$ approaches a point from which $q(v)$ is visible while slightly and monotonically increasing in the $y$-coordinate. The fourth part $\ell_4(v)$ of $\ell(v)$ is a straight-line segment reaching $q(v)$. The fifth part $\ell_5(v)$, the sixth part $\ell_6(v)$, the seventh part $\ell_7(v)$, and the eighth part $\ell_8(v)$ are symmetric to $\ell_4(v)$, $\ell_3(v)$, $\ell_2(v)$, and $\ell_1(v)$, respectively. Property~(i) of $R$ ensures that $\ell(v)$ can be drawn as a $y$-monotone curve. Properties~(ii) and~(iii) ensure that this can be done \mbox{without introducing crossings in $\Gamma'$.} 
	
		The new region $R$ is delimited from the right by the same $y$-monotone curve $\rho_R$ as the old region $R$; the new $y$-monotone curve $\lambda_R$ delimiting $R$ from the left is composed of $\ell(v)$ and of the straight-line vertical segments connecting $v^b$ and $v^a$ with $Q_{i-1}$ and $Q_i$, respectively. The new region $R$ satisfies Properties~(i)--(iii); in particular, for every unprocessed vertex $u$, the new region $R$ contains in its interior $u^b$ and $u^a$, provided that $\ell_2(v)$ and $\ell_7(v)$ are sufficiently close to the $y$-monotone curve $\lambda_R$ delimiting the old region $R$, and contains in its interior $q(u)$, given that $q(v)$ is to the left of $q(u)$, by the definition of $T_i$. This concludes the proof of the claim.	
\end{proof}

	For $i=1,\dots,k-1$, denote by $t'_i(e)$ the intersection point of an edge $e$ of $G$ directed from a vertex of the level $V_i$ to a vertex of the level $V_{i+1}$ with the line $Q_i:= y=2i$ in $\Gamma_G$.
	For each $i=1,\dots,k-1$, consider the left-to-right order $W_i = t'_i(e_1),t'_i(e_2),\dots$ of the intersection points of the edges directed from the vertices of the level $V_i$ to the vertices of the level $V_{i+1}$ with the line $Q_i$ in $\Gamma_G$. The
	The sequence $W_i$ has the following properties:
	\begin{enumerate}[(1)]
		\item For any vertex $v \in V_i$, the intersection points of the edges exiting $v$ with the line $Q_i$ are consecutive in $W_i$. This follows from the fact that $\Gamma_G$ is a level planar drawing.
		\item Let $e_{v,1},e_{v,2},\dots,e_{v,h}$ be the left-to-right order of the edges exiting $v$ as they appear in the upward embedding of $G$; then $t'_i(e_{v,1}), t'_i(e_{v,2}),\dots,t'_i(e_{v,h})$ appear in this left-to-right order in $W_i$. This is due to the fact that $\Gamma_G$ respects the upward embedding of $G$.
		\item For any two vertices $u,v \in V_i$ such that $u \prec_i v$, the intersection points of the edges exiting $u$ precede the intersection points of the edges exiting $v$ in $W_i$. This is due to the fact that $\Gamma_G$ is a level planar drawing and that the left-to-right order of the vertices in $V_i$ in $\Gamma_G$ along $L_i$ is a linear extension of $\prec_i$.
	\end{enumerate}
	
	We obtain a drawing $\Gamma_U$ of $U$ from $\Gamma'$ by drawing, for each edge $e=(u,v)$ of $G$ where $u \in V_i$ and $v \in V_{i+1}$, the edge $(u^a,v^b)$ of $U$ as a $y$-monotone curve between $u^a$ and $v^b$ passing trough $t'_i(e)$. As in the previous direction, such curves can be drawn so that no two of them intersect, except possibly at common endpoints, by processing the edges directed from the vertices in $V_i$ to the vertices in $V_{i+1}$ in the left-to-right order in which they cross $Q_i$ (that is, according to the sequence $W_i$ of their crossing points with such a line). We have that $\Gamma_U$ is a drawing of $U$ that extends $\Gamma_H$. Further, \cref{cl:order-vertices} and Properties (1) and (3) guarantee that $\Gamma_U$ is upward planar. Finally, Property (2) guarantees that $\Gamma_U$ respects the upward embedding of $U$. This concludes the proof of the theorem.
\end{proof}

\cref{le:no-vertices-sameY}, together with~\cref{th:UPE-FUE-hard}, implies the following.

\begin{corollary} \label{cor:UPE-FUE-hard}
	The {\sc UPE-FUE} problem is NP-complete even if 
	\begin{inparaenum}[(i)]
		\item the partial graph contains all the vertices and
		\item no two vertices share the same $y$-coordinate in the partial drawing.
	\end{inparaenum}
\end{corollary}

We conclude this section by proving that the {\sc UPE} problem is solvable in almost-linear time for instances in which the partial graph contains all the vertices and no edges, and no two vertices share the same $y$-coordinate in the partial drawing. 


\begin{theorem}\label{th:UPE-algorithm}	
	The {\sc UPE} problem can be decided in $O(n \log n)$ time for instances $\langle G, H, \Gamma_H \rangle$ such that $G$ has $n$ vertices, $H=(V(G),\emptyset)$, and no two vertices share the same $y$-coordinate in $\Gamma_H$.
\end{theorem}

\begin{proof}
	By applying reduction (ii) of~\cref{th:equivalence} to the instance $\langle G, H, \Gamma_H \rangle$ of the {\sc UPE} 
	problem, we obtain in $O(n \log n)$ time an equivalent instance $(G, \ell, \xi)$ of {\sc Ordered Level Planarity} in which each level contains exactly one vertex.
	As observed by Klemz and Rote~\cite{KlemzR17} such instances of {\sc Ordered Level Planarity} are  solvable in $O(n)$ time, as a consequence of the fact that the {\sc Level Planarity} problem can be solved in $O(n)$ time~\cite{DBLP:conf/gd/JungerLM98}. 
\end{proof}

\section{Upward Planar $st$-Graphs}\label{se:st-graphs}

In this section we study the {\sc UPE} and {\sc UPE-FUE} problems for upward planar $st$-graphs.  
The following lemma will be useful for our algorithms.

\begin{restatable}{lemma}{LemmaConstantTimeQueries}\label{le:constant-time-queries}
	Let $G$ be an $n$-vertex upward planar $st$-graph with a given upward embedding. There exists a data structure to test in $O(1)$ time, for any two vertices $u$ and $v$ of $G$, whether $v\in S_G(u)$, $v\in P_G(u)$, $v\in L_G(u)$, or $v\in R_G(u)$. Further, such a data structure can be constructed in $O(n)$ time.
\end{restatable}

\begin{proof}
	First, we construct the \emph{transitive reduction} $G^*$ of $G$, that is, the upward planar $st$-graph obtained from $G$ by removing all its transitive edges. This can be done in $O(n)$ time by examining each face of $G$; indeed, an edge $(u,v)$ of an upward planar $st$-graph is transitive if and only if at least one of the two paths that connect $u$ and $v$, that delimit the faces the edge $(u,v)$ is incident to, and that are different from the edge $(u,v)$ is monotone.
	
	Next, we show that $G^*$ can be used in place of $G$ in order to answer the desired queries.
	\begin{figure}[t!]
		\centering
		\includegraphics[height=.25\textwidth]{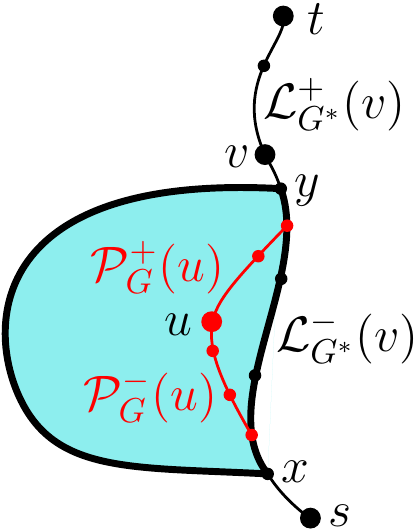}
		\caption{Illustration for the proof that $u\in L_{G^*}(v)$ implies that $u\in L_{G}(v)$. The cycle $\mathcal C_u$ is fat.}
		\label{fi:transitivereduction}
	\end{figure}
	
	\begin{myclaim}
		For each vertex $v$ of $G$ (and of $G^*$), we have $S_{G^*}(v)=S_{G}(v)$, $P_{G^*}(v)=P_{G}(v)$, $L_{G^*}(v)=L_{G}(v)$, and $R_{G^*}(v)=R_{G}(v)$.
	\end{myclaim}
	
	\begin{proof}
		If $u\in S_{G^*}(v)$, then $u\in S_{G}(v)$, as $G^*$ is a subgraph of $G$. Conversely, if $u\in S_{G}(v)$, then consider the longest monotone path from $v$ to $u$ in $G$; such a path also belongs to $G^*$, hence $u\in S_{G^*}(v)$. This proves that $S_{G^*}(v)=S_{G}(v)$; the proof that $P_{G^*}(v)=P_{G}(v)$ is analogous.

		Assume that $u\in L_{G}(v)$; hence, $u$ lies to the left of the monotone path composed of $\mathcal L^-_G(v)$ and $\mathcal L^+_G(v)$. Since $\mathcal L^-_G(v)$ is the leftmost incoming path of $v$ in $G$, all the monotone paths that replace transitive edges of $\mathcal L^-_G(v)$ in order to define $\mathcal L^-_{G^*}(v)$ lie to the right of the monotone path composed of $\mathcal L^-_G(v)$ and $\mathcal L^+_G(v)$ in $G$. Analogous arguments apply to $\mathcal L^+_{G^*}(v)$. Thus, $u$ lies to the left of the monotone path composed of $\mathcal L^-_{G^*}(v)$ and $\mathcal L^+_{G^*}(v)$, hence $u\in L_{G^*}(v)$.

		Conversely, assume that $u\in L_{G^*}(v)$ and suppose, for a contradiction, that $u\notin L_{G}(v)$, as in~\cref{fi:transitivereduction}. It follows that there is a transitive edge $(x,y)$ of $G$ such that: (i) $x$ and $y$ both belong to $\mathcal L^-_{G^*}(v)$ or both belong to $\mathcal L^+_{G^*}(v)$ -- assume the former, as the latter case can be treated analogously; (ii) the edge $(x,y)$ lies to the left of the monotone path composed of $\mathcal L^-_{G^*}(v)$ and $\mathcal L^+_{G^*}(v)$ in $G$; and (iii) $u$ lies inside the (undirected) cycle $\mathcal C_u$ of $G$ delimited by the edge $(x,y)$ and by the monotone path between $x$ and $y$ in $\mathcal L^-_{G^*}(v)$. Consider any longest incoming path $\mathcal P^-_{G}(u)$ of $u$ in $G$ and any longest outgoing path $\mathcal P^+_{G}(u)$ of $u$ in $G$. Since these are longest paths, they also belong to $G^*$. Further, since they are incident to $s$ and $t$, respectively, they share vertices with $\mathcal C_u$. In particular, the subpaths of $\mathcal P^-_{G}(u)$ and $\mathcal P^+_{G}(u)$ that are incident to $u$ and whose edges are inside $\mathcal C_u$ determine a monotone path between two vertices of $\mathcal L^-_{G^*}(v)$ that is to the left of the monotone path composed of $\mathcal L^-_{G^*}(v)$ and $\mathcal L^+_{G^*}(v)$. This contradicts the fact that $\mathcal L^-_{G^*}(v)$  is the leftmost incoming path of $v$ in $G^*$ and hence proves that $u\in L_{G}(v)$. It follows that $L_{G^*}(v)=L_{G}(v)$. 
		
		The proof that $R_{G^*}(v)=R_{G}(v)$ is symmetric.
\end{proof}
	
We now compute a dominance drawing $\Gamma^*$ of $G^*$. A \emph{dominance drawing} of a directed graph is a straight-line drawing such that there is a monotone path from a vertex $u$ to a vertex $v$ if and only if $x(u) \leq x(v)$ and $y(u) \leq y(v)$. Di Battista et al.~\cite{dtt-arsdpud-92} presented an $O(n)$-time algorithm to construct a planar dominance drawing of an $n$-vertex upward planar $st$-graph without transitive edges. The dominance drawings constructed by the cited algorithm satisfy the following properties: (i) $x(v)<x(u)$ if and only if $v \in P_G(u) \cup L_G(u)$; and (ii) $y(v)<y(u)$ if and only if $v \in P_G(u) \cup R_G(u)$. 

We use the algorithm by Di Battista et al.~\cite{dtt-arsdpud-92} in order to construct $\Gamma^*$. Then, in order to query whether $v\in P_G(u)$, it suffices to check whether $x(v)<x(u)$ and $y(v)<y(u)$ in $\Gamma^*$.
The other queries can be similarly answered in $O(1)$ time.
\end{proof}

We now characterize the positive instances of the {\sc UPE-FUE} problem.

\begin{restatable}{lemma}{LemmaCharacterizationUpwardPlane}\label{le:characterization-upward-plane}
	An instance $\langle G, H, \Gamma_H \rangle$ of the {\sc UPE-FUE} problem such that $G$ is an upward planar $st$-graph with a given upward embedding and such that $H$ contains no edges is a positive instance if and only if:
	\begin{description}
		\item[\em Condition~1:] For each vertex $v$ of $H$, all its successors (predecessors) in $G$ that belong to $H$ have a $y$-coordinate in $\Gamma_H$ that is larger (smaller) than $y(v)$; and
		\item[\em Condition~2:] For each vertex $v$ of $H$, all the vertices of $H$ whose $y$-coordinate is the same as $y(v)$ and whose $x$-coordinate is larger (smaller) than $x(v)$ in $\Gamma_H$ are to the right (to the left) of $v$ in $G$.
	\end{description}
\end{restatable}

\begin{proof}
	Concerning the necessity of Condition~1, assume that two vertices $u$ and $v$ exist in $H$ such that $G$ contains a monotone path $P$ from $u$ to $v$ and such that the $y(u)\geq y(v)$ in $\Gamma_H$. Then $P$ cannot be upward in any drawing of $G$ extending $\Gamma_H$. Concerning the necessity of Condition~2, assume, for a contradiction, that two vertices $u$ and $v$ of $H$ exist such that: (i) $y(u)=y(v)$ and $x(u)<x(v)$ in $\Gamma_H$; and (ii) $u\notin L_G(v)$. If $u\in S_G(v)$, then a monotone path from $v$ to $u$ cannot be upward in any drawing of $G$ extending $\Gamma_H$. Analogously, $u\notin P_G(v)$. Finally, if $u\in R_G(v)$, then consider any vertex $w\in S_G(u) \cap S_G(v)$ such that there are two edge-disjoint monotone paths $Q_{uw}$ and $Q_{vw}$ from $u$ to $w$ and from $v$ to $w$, respectively. Such paths can be found by considering any two monotone paths from $u$ to $t$ and from $v$ to $t$, and by truncating these paths at their first common vertex; since $u\notin S_G(v)$ and $u\notin P_G(v)$, we have $w\neq u,v$. Since $u\in R_G(v)$, the left-to-right order of the edges entering $w$ in $G$ is: The edge of $Q_{vw}$ first and the edge of $Q_{uw}$ second. However, since $y(u)=y(v)$ and $x(u)<x(v)$, the left-to-right order of the edges entering $w$ in any upward planar drawing of $G$ extending $\Gamma_H$ is: The edge of $Q_{uw}$ first and the edge of $Q_{vw}$ second. This contradiction proves the necessity of Condition~2.
	
	To prove the sufficiency we construct an upward planar drawing $\Gamma_G$ of $G$ that extends~$\Gamma_H$. 
	
	We first augment $\Gamma_H$ by drawing every vertex $v\in V(G)\setminus V(H)$, so that Conditions~1 and~2 are still satisfied after the augmentation. This is done by assigning to $v$ any $x$-coordinate $x(v)$ and a $y$-coordinate $y(v)$ such that: (i) $y(v)\neq y(u)$, for any vertex $u\in V(G)$ with $v \neq u$; (ii) $y(v)>y(u)$, for any $u\in P_G(v)$; and (iii) $y(v)<y(w)$, for any $w\in S_G(v)$. We show how to construct such an assignment.
	
	We consider the vertices in $V(G)\setminus V(H)$ one at a time. When a vertex $v$ is considered, some of its successors in $G$ might be already drawn in $\Gamma_H$; denote by $S'_G(v)$ the set of such vertices. Analogously, let $P'_G(v)$ be the set of the predecessors of $v$ that are already drawn in $\Gamma_H$. Each vertex $w$ in $S'_G(v)$ has a $y$-coordinate in $\Gamma_H$ that is larger than the $y$-coordinate of any vertex $u$ in $P'_G(v)$; namely, since $w$ is a successor of $v$ and since $u$ is a predecessor of $v$, it follows that $w$ is a successor of $u$ and then Condition~1 ensures that $y(u)<y(w)$ in $\Gamma_H$. We place $v$ at any point that is higher than all the vertices in $P'_G(v)$, that is lower than all the vertices in $S'_G(v)$, and whose $y$-coordinate is different from all the other vertices of $H$. Clearly, Conditions~1 and~2 are still satisfied by the new instance.  After repeating this augmentation for all the vertices in $V(G)\setminus V(H)$, we eventually get that $V(H)=V(G)$ and the instance $\langle G, H, \Gamma_H \rangle$ still satisfies Conditions~1 and~2.
	
	Now set $\Gamma_G=\Gamma_H$. We are going to draw the edges of $G$ in $\Gamma_G$; we start with the edges of the leftmost path $\mathcal L^+_G(s)$ of $G$ (see~\cref{fig:leftmost}). We draw each edge $(u,v)$ of $\mathcal L^+_G(s)$ as follows. If there is no vertex $w$ such that $y(u)<y(w)<y(v)$, then we draw $(u,v)$ as a straight-line segment. Otherwise, we draw $(u,v)$ as a polygonal line composed of three straight-line segments: The first one connects $u$ with a point $p_u$ whose $y$-coordinate is slightly larger than $y(u)$ and whose $x$-coordinate is smaller than the one of every vertex of $G$ in $\Gamma_G$; the second one is a vertical straight-line segment connecting $p_u$ with a point $p_v$ whose $y$-coordinate is slightly smaller than $y(v)$; the third one connects $p_v$ with $v$. By construction all the vertices of $G$ not in $\mathcal L^+_G(s)$ are to the right of $\mathcal L^+_G(s)$ in $\Gamma_G$. 
	
\begin{figure}[htb]
	\centering
	\begin{subfigure}{.3\textwidth}\centering
		\includegraphics[height=.8\textwidth,page=1]{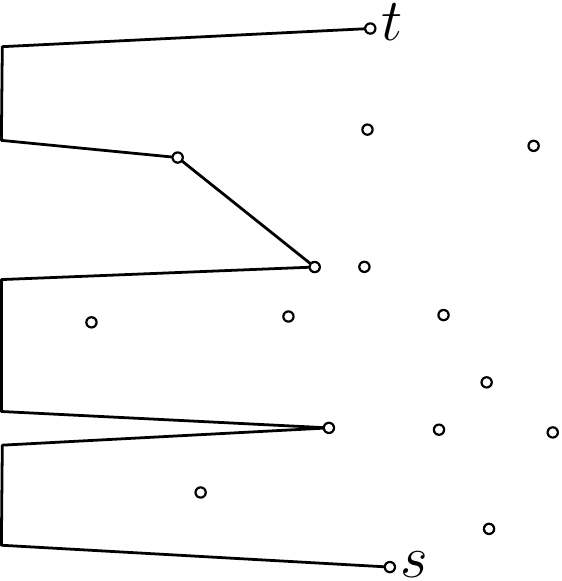}
		\subcaption{}\label{fig:leftmost}
	\end{subfigure}\hfil
	\begin{subfigure}{.3\textwidth}\centering
		\includegraphics[height=.8\textwidth,page=2]{ST-characterization-2}
		\subcaption{}\label{fi:rightboundary}
	\end{subfigure}
	\caption{
		(a) Drawing the leftmost path $\mathcal L^+_G(s)$ of $G$. (b) Drawing the right boundary $(u_1,\dots,u_l)$ of a face $f$.}
\end{figure}

	Now the algorithm proceeds in steps. After each step we maintain the invariants that: (i) the subgraph of $G$ currently drawn consists of an upward planar $st$-graph $G'$ plus a set of isolated vertices; (ii) the current drawing of $G'$ in $\Gamma_G$ is upward planar; and (iii) the rightmost path of $G'$ is represented by a $y$-monotone curve $\gamma_r$ such that all the isolated vertices lie to the right of $\gamma_r$ in $\Gamma_{G}$, when traversing $\gamma_r$ from $s$ to $t$. 
	
	The invariants are initially satisfied with $G'=\mathcal L^+_G(s)$ and with the drawing $\Gamma_G$ constructed as above. In each step we consider a face $f$ of $G$ whose left boundary belongs to $G'$ and whose right boundary consists of edges not in $G'$. In order to draw the right boundary $(u_1,\dots,u_l)$ of $f$, we draw each edge $(u_i,u_{i+1})$ independently (see~\cref{fi:rightboundary}). We draw $(u_i,u_{i+1})$ as a polygonal line composed of three parts. The first part is a straight-line segment connecting $u_i$ with a point $p_{u_i}$ whose $y$-coordinate is slightly larger than $y(u_i)$ and whose $x$-coordinate is slightly larger than the $x$-coordinate of the point of the left boundary of $f$ with the same $y$-coordinate. The second part is a polygonal line arbitrarily close to the left boundary of $f$, connecting $p_{u_i}$ with a point $p_{u_{i+1}}$ whose $y$-coordinate is slightly smaller than $y(u_{i+1})$. The third part is a straight-line segment connecting $p_{u_{i+1}}$ with $u_{i+1}$. Invariant (i) is maintained by the choice of the face $f$. Invariant (ii) is maintained since, for each $i=1,\dots,l-1$, we have $y(u_i)<y(u_{i+1})$ in $\Gamma_G$, by Condition~1, and since the drawing of $(u_i,u_{i+1})$ is upward and does not cross any previously drawn edge, by construction. Invariant (iii) is maintained since each edge $(u_i,u_{i+1})$ is arbitrarily close to the drawing of the rightmost path of $G'$, except in the interior of arbitrarily narrow horizontal strips enclosing $u_i$ and $u_{i+1}$; further, $(u_i,u_{i+1})$ does not keep any isolated vertex to its left inside such strips by Condition~2. 
	
	By invariant (ii), $\Gamma_G$ respects the upward embedding of $G$. Further, by construction, the vertices of $H$ have the same coordinates in $\Gamma_G$ as in $\Gamma_H$, hence $\Gamma_G$ extends $\Gamma_H$. 
\end{proof}

We can now prove the following algorithmic theorem.

\begin{restatable}{theorem}{TheoremStUPEFUE} \label{th:st-UPE-FUE}
	The {\sc UPE-FUE} problem can be solved in $O(n \log n)$ time for instances $\langle G, H, \Gamma_H \rangle$ with size $n = |\langle G, H, \Gamma_H \rangle|$ such that $G$ is an upward planar $st$-graph with a given upward embedding.
\end{restatable}

\begin{proof}
	If $H$ contains edges, then~\cref{le:no-edges} can be applied in $O(n \log n)$ time in order to transform $\langle G, H, \Gamma_H \rangle$ into an equivalent instance, which we again denote by $\langle G, H, \Gamma_H \rangle$, such that $|\langle G, H, \Gamma_H \rangle|\in O(n)$, such that $H$ contains no edges, and such that $G$ is still an upward planar $st$-graph. We show an $O(n \log n)$-time algorithm to test whether Conditions~1 and~2 of~\cref{le:characterization-upward-plane} are satisfied by $\langle G, H, \Gamma_H \rangle$.
	
	In order to test Condition~1, we proceed as follows. We construct an auxiliary graph $\mathcal A$, which we initialize to $G$. 
	We order in $O(h\log h)\in O(n\log n)$ time the vertices in $H$ according to their $y$-coordinates in $\Gamma_H$, where $h=|V(H)|$. Let $v_1,v_2,\dots,v_h$ be such an ordering.
	For every two maximal sets $\{v_i,v_{i+1},\dots,v_j\}$ and $\{v_{j+1},v_{j+2},\dots,v_{k}\}$ of vertices of $H$ such that $y(v_i)=\dots=y(v_j)$, such that $y(v_{j+1})=\dots=y(v_k)$, and such that $y(v_j)<y(v_{j+1})$,  we add to $\mathcal A$ a vertex $x_{i,k}$, directed edges $(v_l,x_{i,k})$, for $l = i,\dots,j$, and directed edges $(x_{i,k},v_l)$, for $l = j+1,\dots,k$. The graph $\mathcal A$ can be constructed in $O(n+h\log h)\in O(n \log n)$ time and has $O(n)$ vertices and edges.

	\begin{myclaim}\label{cl:acyclicity}
		$\langle G, H, \Gamma_H \rangle$ satisfies Condition~1 if and only if $\mathcal A$ is acyclic. 
	\end{myclaim}
	
	\begin{proof}
		For the necessity, suppose that Condition~1 holds true; we show that $\mathcal A$ is acyclic. For a contradiction, suppose that a simple directed cycle $\mathcal C$ exists in $\mathcal A$; then $\mathcal C$ can be partitioned into monotone paths $P_1,\dots,P_{2k}$ where all the edges of $P_i$ belong to $G$ if $i$ is even and all the edges of $P_i$ do not belong to $G$ if $i$ is odd. Note that, for $i=1,\dots,2k$, the vertex $u_i$ shared by two consecutive paths $P_i$ and $P_{i+1}$ belongs to $H$, where $P_{2k+1}=P_1$. Condition~1 implies that $y(u_i)<y(u_{i+1})$ for $i=2,4,\dots,2k$, where $u_{2k+1}=u_1$. The construction of $\mathcal A$ implies $y(u_i)<y(u_{i+1})$ for $i=1,3,\dots,2k-1$. However, this implies that $y(u_1)<y(u_1)$, a contradiction. 
		
		For the sufficiency, suppose that $\mathcal A$ is acyclic; we show that Condition~1 holds true. For a contradiction, suppose that two vertices $u,v\in V(H)$ exist such that $v\in S_G(u)$ and such that $y(u)> y(v)$ in $\Gamma_H$. We have that $\mathcal A$ contains a monotone path directed from $u$ to $v$ as $\mathcal A$ contains all the edges of $G$, and a monotone path directed from $v$ to $u$ composed of edges not in $G$, hence it contains a directed cycle, a contradiction.   
	\end{proof}
	
	Since $\mathcal A$ has $O(n)$ vertices and edges, it can be tested in $O(n)$ time whether it is acyclic. By \cref{cl:acyclicity} it follows that we can test Condition~1 in $O (n \log n)$ time.
	

	In order to test Condition~2, we proceed as follows. We sort each maximal set of vertices $\{v_i,\dots,v_j\}$ in $H$ with the same $y$-coordinate according to their $x$-coordinates. For all the vertices of $H$, this can be done in $O(h \log h)\in O(n\log n)$ time.
	Hence, we assume that $x(v_l)<x(v_{l+1})$, for $l=i,\dots,j-1$; then we test whether $v_l\in L_G(v_{l+1})$, for $l=i,\dots,j-1$. By~\cref{le:constant-time-queries}, this can be done in $O(1)$-time per query, after an $O(n)$-time preprocessing.  
	Therefore, the total time to test Condition~2 is also $O(n \log n)$. 
\end{proof}


Next, we deal with the {\sc UPE} problem. An instance $\langle G, H, \Gamma_H \rangle$ of the {\sc UPE} problem such that $G$ is an upward planar $st$-graph can be transformed into an equivalent instance of the {\sc PLP} problem.  
This is due to the fact that Condition~1 of~\cref{le:characterization-upward-plane} does not depend on the upward embedding of $G$ and that we can assume: 
\begin{enumerate}
	\item the edges set of $H$ to be empty, by~\cref{le:no-edges}; and 
	\item the partial drawing to contain all the vertices of $G$, by drawing each vertex in $V(G) \setminus V(H)$ as in the proof of~\cref{le:characterization-upward-plane}, without violating neither Condition~1 nor Condition~2 of the lemma. 
\end{enumerate}
Hence, the {\sc UPE} problem for upward planar $st$-graphs can be solved in quadratic time, due to the results of Br\"uckner and Rutter about the {\sc PLP} problem for single-source graphs~\cite{DBLP:conf/soda/BrucknerR17}.
However, in the following theorem we show how to reduce the time bound to almost linear. 

\begin{restatable}{theorem}{TheoremStUPE} \label{th:st-UPE}
	The {\sc UPE} problem can be solved in $O(n \log n)$ time for instances $\langle G, H, \Gamma_H \rangle$ with size $n = |\langle G, H, \Gamma_H \rangle|$ such that $G$ is an upward planar $st$-graph.
\end{restatable}

\begin{proof}
	We are going to test whether an upward embedding of $G$ exists that satisfies the conditions in~\cref{le:characterization-upward-plane}. Actually, Condition~1 does not depend on the upward embedding of $G$, hence it can be tested in $O(n\log n)$ time as described in the proof of~\cref{th:st-UPE-FUE} before any upward embedding of $G$ is considered. If the test succeeds, we apply~\cref{le:no-edges} in $O(n \log n)$ time to modify $\langle G, H, \Gamma_H \rangle$ so that $H$ contains no edges while $G$ remains an upward planar $st$-graph, and proceed as described in the following, otherwise we conclude that the instance is negative.
	
	 In order to test whether $G$ admits an upward embedding satisfying Condition~2 of~\cref{le:characterization-upward-plane} we proceed as follows. First, we add the edge $(s,t)$ to $G$, if $G$ does not contain such an edge. Second, we compute in $O(n)$ time the SPQR-tree $\mathcal T$ of $G$. Third, we compute in $O(n \log n)$ time the order $v_1,v_2,\dots,v_h$ of the vertices in $H$ by increasing $y$-coordinates and, secondarily, by increasing $x$-coordinates in $\Gamma_H$. 
	 	 
	 We now aim to decide a left-to-right order of the virtual edges of the skeleton of each P-node of $\mathcal T$ and a flip for the triconnected skeleton of each R-node of $\mathcal T$ so that Condition~2 is satisfied. We outline the approach for such decisions. Consider two vertices $u=v_i$ and $v=v_{i+1}$ sharing their $y$-coordinate. Note that $x(u)<x(v)$ in $\Gamma_H$. Then $u$ has to belong to $L_G(v)$ in the upward embedding of $G$ we look for. This imposes a constraint on $\sk(\nu)$ for a node $\nu$ of $\mathcal T$ such that $u$ and $v$ are in the pertinent graphs of two different virtual edges $e_u$ and $e_v$ of $\skel(\nu)$. Namely, if $\nu$ is a P-node, then $e_u$ has to precede $e_v$ in the left-to-right order of the virtual edges of $\sk(\nu)$. Further, if $\nu$ is an R-node, then $e_u$ has to be to the left of $e_v$ \mbox{in the chosen embedding of $\sk(\nu)$.} To impose these constraints, we employ several algorithmic tools; e.g., we compute in $O(1)$ time the proper allocation nodes $\mu_u$ and $\mu_v$ of $u$ and $v$ in $\mathcal T$, and the lowest common ancestor $\nu$ of $\mu_u$ and $\mu_v$ in $\mathcal T$. This approach is detailed as follows.
	
	We compute the following data structures. 
	
	\begin{enumerate}
		\item We equip each P-node $\nu$ of $\mathcal T$ with an auxiliary directed graph $LR_\nu$ containing a vertex for each virtual edge of $\sk(\nu)$, except for the one corresponding to the parent of $\nu$ in $\mathcal T$. The edge set of $LR_\nu$ is initially empty. 
		\item The skeleton of each R-node $\nu$ of $\mathcal T$ is triconnected, hence it admits two upward embeddings, which can be obtained from each other via a flip. We arbitrarily choose one of these upward embeddings. We equip $\skel(\nu)$ with a data structure that, given a pair $(x,y)$ where $x$ and $y$ are vertices or edges of $\skel(\nu)$, determines in $O(1)$ time whether $x\in L_{\skel(\nu)}(y)$, $x\in R_{\skel(\nu)}(y)$, or none of the previous, in the chosen upward embedding of $\skel(\nu)$. Such a data structure can be constructed in $O(|\sk(\nu)|)$ time, as in~\cref{le:constant-time-queries}, by inserting a dummy vertex on each virtual edge in order to handle the fact that queries might involve vertices but also virtual edges of $\skel(\nu)$. Further, we equip $\skel(\nu)$ with two boolean variables $preserve(\nu)$ and $flip(\nu)$ that we both initially set to \texttt{false}.
	\end{enumerate}
	
	We execute the following algorithm. 
	
	\begin{enumerate}
		\item We consider each pair $u=v_i,v=v_{i+1}$ of vertices with the same $y$-coordinate that are consecutive in the computed order. Note that $x(u)<x(v)$ in $\Gamma_H$. We perform the following operations, which ensure that $u\in L_G(v)$ in the upward embedding of $G$ constructed by the algorithm, if any.
		\begin{enumerate}
			\item We compute in $O(1)$ time the proper allocation nodes $\mu_u$ and $\mu_v$ of $u$ and $v$ in $\mathcal T$, respectively.
			\item We compute in $O(1)$ time the lowest common ancestor $\nu$ of $\mu_u$ and $\mu_v$ in $\mathcal T$. Let $x_u$ and $x_v$ be the representatives of $u$ and $v$ in $\skel(\nu)$, respectively.
			\item We perform different operations depending on the type of $\nu$. 
			\begin{enumerate}
				\item If $u=s$, $u=t$, $v=s$, or $v=t$, then we reject the instance.
				\item If $\nu$ is an S-node, then we reject the instance.
				\item If $\nu$ is a P-node, then we add in $O(1)$ time a directed edge in $LR_\nu$ from the vertex corresponding to $x_u$ to the vertex corresponding to $x_v$.
				\item If $\nu$ is an R-node, then we query in $O(1)$ time the data structure $\skel(\nu)$ has been equipped with.
				If $x_u\in L_{\skel(\nu)}(x_v)$, then we set $preserve(\nu)=$\texttt{true}.
				If $x_u\in R_{\skel(\nu)}(x_v)$, then we set $flip(\nu)=$\texttt{true}.
				If $x_u\notin L_{\skel(\nu)}(x_v)$ and $x_u\notin R_{\skel(\nu)}(x_v)$, then we reject the instance.
			\end{enumerate}
		\end{enumerate}
		\item For each P-node $\nu$ of $\mathcal T$, we test whether $LR_\nu$ contains a directed cycle. In case of a positive answer, we reject the instance. 
		\item For each R-node $\nu$ of $\mathcal T$, we test in $O(1)$ time whether $preserve(\nu)=$\texttt{true} and $flip(\nu)=$\texttt{true}. In case of a positive answer, we reject the instance.
		\item We accept the instance. 
	\end{enumerate}
	
	We analyze the running time of the algorithm. Step 1 takes $O(1)$ time for each pair $[v_i,v_{i+1}]$, hence $O(h)\in O(n)$ time in total. Step~2 takes, for each node $\nu$, time proportional to the number of the edges that are inserted in $LR_\nu$. Over all the P-nodes $\nu$ of $\mathcal T$, at most $h-1$ edges are inserted in the auxiliary graphs $LR_\nu$, namely at most one for each pair $[v_i,v_{i+1}]$. Thus the overall time complexity of Step~2 is $O(h)\in O(n)$. Step 3 takes $O(1)$ time for each R-node, hence $O(n)$ time in total. Finally, Step~4 takes $O(n)$ time. Hence, the total running time is dominated by the sorting of the vertices in $H$ and the algorithm runs in $O(n\log n)$ time.

	
	Finally, we argue about the correctness of the algorithm. 
	
	Consider step 1.(c)i. If $u$ or $v$ coincides with $s$ or $t$, then Condition~2 is not satisfied by any upward embedding of $G$, as $s$ and $t$ do not have any vertices to their left or to their right in any upward embedding. 
	
	Consider step 1.(c)ii. If $\nu$ is an S-node, then four cases are possible. If $\nu \neq \mu_u$ and $\nu \neq \mu_v$, as in~\cref{fig:snode1}, then $x_u$ and $x_v$ are two distinct virtual edges of $\sk(\nu)$. Further, if $\nu = \mu_u$ and $\nu \neq \mu_v$, as in~\cref{fig:snode2} (the case in which $\nu \neq \mu_u$ and $\nu = \mu_v$ is symmetric), then $x_u$ is a vertex of $\skel(\nu)$ and $x_v$ is a virtual edge of $\sk(\nu)$. Finally, if $\nu = \mu_u$ and $\nu = \mu_v$, as in~\cref{fig:snode3}, then $x_u$ and $x_v$ are distinct vertices of $\skel(\nu)$. In all cases, we have that either $u\in S_G(v)$ or $u\in P_G(v)$, hence Condition~2 is not satisfied in any upward embedding of $G$. 
	
	\begin{figure}[htb]
		\begin{subfigure}{.1\textwidth}
			\centering
			\includegraphics[scale = 0.8]{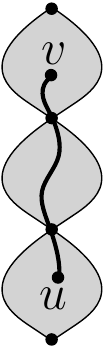}
			\subcaption{}\label{fig:snode1}
		\end{subfigure}
		\hfil
		\begin{subfigure}{.1\textwidth}
			\centering
			\includegraphics[scale = 0.8]{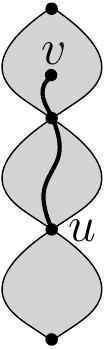}
			\subcaption{}\label{fig:snode2}
		\end{subfigure}
		\hfil
		\begin{subfigure}{.1\textwidth}
			\centering
			\includegraphics[scale = 0.8]{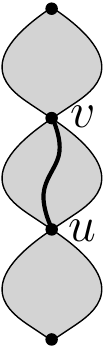}
			\subcaption{}\label{fig:snode3}
		\end{subfigure}
		\hfil
		\begin{subfigure}{.1\textwidth}
			\centering
			\includegraphics[scale = 0.8]{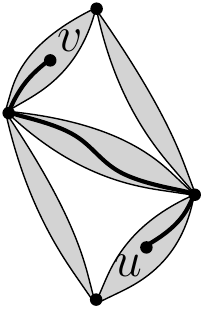}
			\subcaption{}\label{fig:rnode1}
		\end{subfigure}
		\hfil
		\begin{subfigure}{.1\textwidth}
			\centering
			\includegraphics[scale = 0.8]{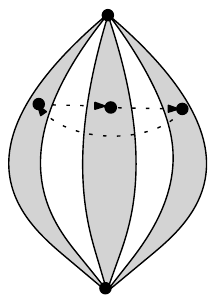}
			\subcaption{}\label{fig:pnode1}
		\end{subfigure}
		\hfil
		\begin{subfigure}{.1\textwidth}
			\centering
			\includegraphics[scale = 0.8]{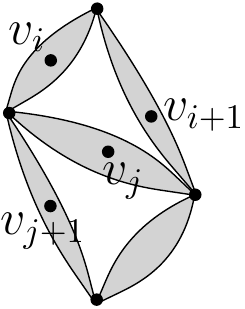}
			\subcaption{}	\label{fig:rnode2}
		\end{subfigure}
		\caption{The cases in which the algorithm rejects an instance. Gray-filled regions represent pertinent graphs of the virtual edges of the skeleton of the node under investigation. The virtual edge representing the parent of the considered node is not shown. (a)--(c) The node $\nu$ is an S-node. (d) The node $\nu$ is an R-node, $u\notin L_G(v)$, and $u\notin R_G(v)$. (e) For some P-node $\nu$ of $\mathcal T$ there is a directed cycle in $LR_\nu$. (f) For some R-node $\nu$ of $\mathcal T$, we have $preserve(\nu)=$\texttt{true} and $flip(\nu)=$\texttt{true}.}
		\label{fig:st-algorithm}
	\end{figure}
	
	Consider step 1.(c)iii. If $\nu$ is a P-node, then $x_u$ and $x_v$ are two virtual edges of $\sk(\nu)$. In fact if, say, $u$ were a vertex of $\sk(\nu)$, then it would also be a vertex of the skeleton of the parent of $\nu$ in $\mathcal T$; hence, its proper allocation node would be a proper ancestor of $\nu$, thus contradicting the fact that $\nu$ is the lowest common ancestor of the proper allocation nodes of $u$ and $v$ in $\mathcal T$. The insertion of the directed edge in $LR_\nu$ from the vertex corresponding to $x_u$ to the vertex corresponding to $x_v$ enforces $u\in L_G(v)$ in the upward embedding of $G$ that is constructed by the algorithm, if any.
	
	Consider step 1.(c)iv. If $\nu$ is an R-node, then four cases are possible, as for an S-node. In all of them we check whether the arbitrarily chosen upward embedding of $\skel(\nu)$ ensures that $u\in L_G(v)$ (then setting $preserve(\nu)=$\texttt{true} ensures that such an embedding  of $\skel(\nu)$ is the one in the upward embedding of $G$ that is constructed by the algorithm, if any), whether flipping the arbitrarily chosen upward embedding of $\skel(\nu)$ ensures that $u\in L_G(v)$ (then setting $flip(\nu)=$\texttt{true} ensures that such an embedding of $\skel(\nu)$ is flipped in the upward embedding of $G$ that is constructed by the algorithm, if any), or whether $u\in L_G(v)$ is not obtained by any choice of the embedding of $\skel(\nu)$, since $u\in S_G(v)$ or $u\in P_G(v)$; see~\cref{fig:rnode1}.
	
	Consider step 2. If a directed cycle is detected in $LR_\nu$, for some P-node $\nu$, then the left-to-right order of the pertinent graphs of the virtual edges of $\sk(\nu)$ in any upward embedding of $G$ does not satisfy all the constraints stemming from the left-to-right order of the vertices of $H$ sharing the same $y$-coordinate in $\Gamma_H$; see~\cref{fig:pnode1}.
	
	Consider step 3. If there is an R-node $\nu$ with $preserve(\nu)=$\texttt{true} and $flip(\nu)=$\texttt{true}, then no flip of $\sk(\nu)$ allows us to construct an upward embedding of $G$ satisfying Condition~2. For example, in~\cref{fig:rnode2} the pair $[v_i,v_{i+1}]$ forces the upward embedding of $\sk(\nu)$ to be the one shown in the illustration, given that $v_{i+1} \in R_G(v_i)$ in such an embedding; however the pair $[v_j,v_{j+1}]$ forces the upward embedding of $\sk(\nu)$ not to be the one shown in the illustration, given that $v_{j+1} \in L_G(v_j)$.
	
	Finally, consider step 4. In order to prove its correctness, we construct an upward embedding of $G$ satisfying Condition~2; this is done as follows. For each R-node $\nu$ of $\mathcal T$, if $flip(\nu)=$\texttt{true}, then we flip the arbitrarily chosen upward embedding of $\skel(\nu)$, otherwise we keep it as it is. 
	For each P-node $\nu$ in $\mathcal T$, we select a total order for the virtual edges of $\skel(\nu)$ corresponding to children of $\nu$ in $\mathcal T$ that extends the partial order given by $LR_\nu$. Finally, if $(s,t)$ was not originally in $G$, then we remove it. This leads to an upward embedding satisfying Condition~2 of~\cref{le:characterization-upward-plane}.
\end{proof}

\section{Directed Paths and Cycles}\label{se:pathsANDcycles}

In this section we study the upward planarity extension problem for instances $\langle G, H, \Gamma_H \rangle$ such that $G$ is a directed path or cycle. Determining the time complexity of the {\sc UPE} and {\sc UPE-FUE} problems for such instances, despite the simplicity of their structure, has proved to be very challenging. However, in the following we exhibit polynomial-time decision algorithms for the cases in which $H$ does not contain edges and no two vertices share the same $y$-coordinate in $\Gamma_H$. We start with the {\sc UPE-FUE} problem for directed paths.

\begin{restatable}{theorem}{TheoremAlgoUPEFUEpath}\label{th:algo-UPE-FUE-path}
	The {\sc UPE-FUE} problem can be solved in $O(n^4)$ time for instances $\langle G, H, \Gamma_H \rangle$ such that $G$ is an $n$-vertex directed path with a given upward embedding, $H$ contains no edges, and no two vertices share the same $y$-coordinate in $\Gamma_H$.
\end{restatable}

\begin{proof}
	Let $G=(u_1,\dots,u_n)$. We show a decision algorithm for the {\sc UPE-FUE} problem employing dynamic programming. The idea is to decide whether $\langle G, H, \Gamma_H \rangle$ is a positive instance of the {\sc UPE-FUE} problem based on whether the subpaths of $G$ admit upward planar extensions with given upward embedding. 
	
	
	In particular, we fill a table with entries $t(u_i,u_j,u_m,u_M)$, for all the indices $i,j,m,M\in \{1,\dots,n\}$ such that $i \leq m \leq j$ and $i \leq M \leq j$, with $i\neq j$ and $m\neq M$. Consider the subpath $G_{i,j}=(u_i,\dots,u_j)$ of $G$. 
	Let $\Gamma_{H,i,j}$ be the restriction of $\Gamma_H$ to the vertices that belong to $G_{i,j}$. The entry $t(u_i,u_j,u_m,u_M)$ has value {\sc true} if there is an upward planar drawing $\Gamma_{G,i,j}$ of $G_{i,j}$ that extends $\Gamma_{H,i,j}$ and such that $u_m$ and $u_M$ are the vertices with the smallest and largest $y$-coordinate in $\Gamma_{G,i,j}$, respectively; the entry $t(u_i,u_j,u_m,u_M)$ has value {\sc false} otherwise.  
	
	We start by computing the entries $t(u_i,u_j,u_m,u_M)$ such that $G_{i,j}$ is a monotone path; these include the entries $t(u_i,u_{i+1},u_m,u_M)$. Assume that the edge $(u_i,u_{i+1})$ of $G$ is outgoing $u_i$, the other case is symmetric. Then $t(u_i,u_j,u_m,u_M)=$ {\sc true} if and only if the following conditions are satisfied: (1) $m=i$; (2) $M=j$; and (3) for any two indices $i'$ and $j'$ such that $i\leq i'<j'\leq j$ and such that $u_{i'},u_{j'}\in V(H_{i,j})$, we have $y(u_{i'})<y(u_{j'})$ in $\Gamma_{H,i,j}$. 
	
	Assume now that $G_{i,j}$ is not a monotone path and that the values of all the entries $t(u_i,u_j,u_m,u_M)$ such that $1\leq j-i\leq x$ have been computed, for some $x\in \{1,2,\dots\}$. After the computation of the entries $t(u_i,u_j,u_m,u_M)$ such that $G_{i,j}$ is a monotone path, this is indeed the case with $x=1$. We compute the values of the entries $t(u_i,u_j,u_m,u_M)$ such that $j-i=x+1$. We distinguish three cases, based on how many of the equalities $i=m$, $i=M$, $j=m$, and $j=M$ are satisfied, that is, based on how many vertices among $u_m$ and $u_M$ are end-vertices of $G_{i,j}$. Refer to~\cref{fig:paths-3cases}.
	
	\begin{figure}[htb]
		\centering
		\begin{subfigure}{.3\textwidth}
			\includegraphics[scale=0.6]{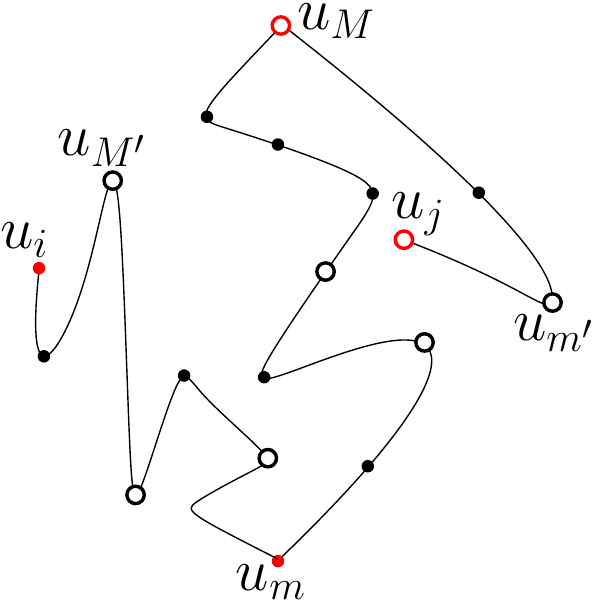}
			\subcaption{}
		\end{subfigure}
		\begin{subfigure}{.3\textwidth}
			\includegraphics[scale=0.6]{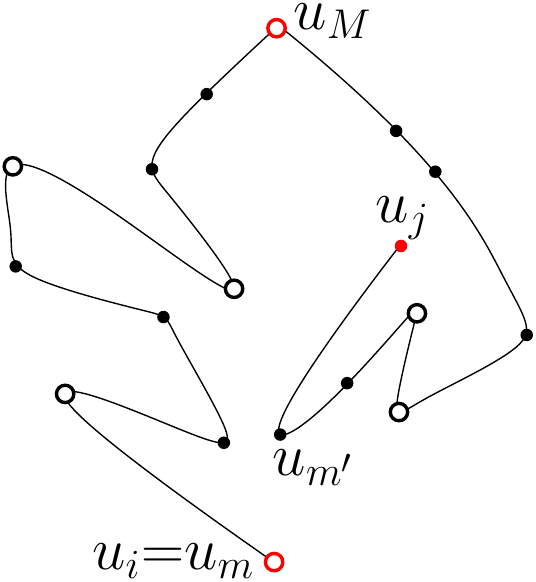}
			\subcaption{}
		\end{subfigure}
		\begin{subfigure}{.3\textwidth}
			\includegraphics[scale=0.6]{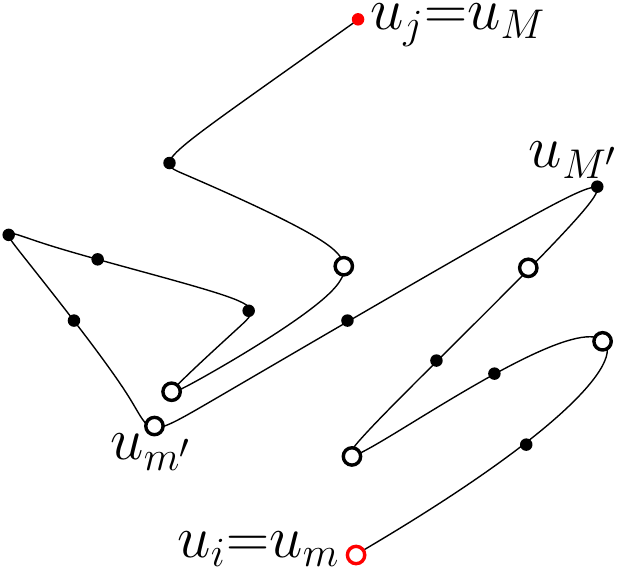}
			\subcaption{}
		\end{subfigure}
		\caption{Three cases that might occur in an upward planar drawing of a path $G_{i,j}$. (a) The vertices $u_m$ and $u_M$ with the smallest and largest $y$-coordinate are not end-vertices of $G_{i,j}$; (b) one of $u_m$ and $u_M$ is an end-vertex of $G_{i,j}$, while the other one is not; and (c) $u_m$ and $u_M$ are both end-vertices of $G_{i,j}$.}
		\label{fig:paths-3cases}
	\end{figure}
	
	{\em Suppose first that neither $u_m$ nor $u_M$ is an end-vertex of $G_{i,j}$.} Recall that $m\neq M$; in fact, assume that $m<M$, the case in which $m>M$ can be treated symmetrically. We have the following.
	
	\begin{myclaim} \label{cl:paths-no-endv}
		$t(u_i,u_j,u_m,u_M)=$ {\sc true} if and only if there exist indices $M'\in \{i,\dots,m-1\}$ and $m'\in \{M+1,\dots,j\}$ such that the following conditions hold true:
		\begin{enumerate} [(1)]
			\item $t(u_i,u_m,u_m,u_{M'})=$ {\sc true};
			\item $t(u_m,u_M,u_m,u_M)=$ {\sc true};
			\item $t(u_M,u_j,u_{m'},u_M)=$ {\sc true};
			\item either $u_m$ does not belong to $H$ or $u_m$ has the smallest $y$-coordinate among the vertices in $\Gamma_{H,i,j}$; and 
			\item either $u_M$ does not belong to $H$ or $u_M$ has the largest $y$-coordinate among the vertices in $\Gamma_{H,i,j}$.
		\end{enumerate}  
	\end{myclaim}  
	
	\begin{proof}
		We first prove the necessity. Suppose that $t(u_i,u_j,u_m,u_M)=$ {\sc true}, hence an upward planar drawing $\Gamma_{G,i,j}$ of $G_{i,j}$ exists that extends $\Gamma_{H,i,j}$ and in which $u_m$ and $u_M$ are the vertices with the smallest and largest $y$-coordinate, respectively. Restricting $\Gamma_{G,i,j}$ to the vertices and edges of $G_{m,M}$ yields an upward planar drawing $\Gamma_{G,m,M}$ of $G_{m,M}$ that extends $\Gamma_{H,m,M}$ and in which $u_m$ and $u_M$ are the vertices with the smallest and largest $y$-coordinate, respectively, which proves Condition~(2). Let $u_{M'}$ be the vertex of $G_{i,m}$ with the largest $y$-coordinate in $\Gamma_{G,i,j}$. Then restricting $\Gamma_{G,i,j}$ to the vertices and edges of $G_{i,m}$ yields an upward planar drawing $\Gamma_{G,i,m}$ of $G_{i,m}$ that extends $\Gamma_{H,i,m}$ and in which $u_m$ and $u_{M'}$ are the vertices with the smallest and largest $y$-coordinate, respectively, which proves Condition~(1). The proof of Condition~(3) is analogous. By assumption $u_m$ is the vertex with the smallest $y$-coordinate in $\Gamma_{G,i,j}$, hence Condition~(4) follows, given that $\Gamma_{G,i,j}$ extends $\Gamma_{H,i,j}$. Condition~(5) is proved analogously.
		
		The proof of the sufficiency is more involved. By Conditions~(1)--(3) there exist upward planar drawings $\Gamma_{G,i,m}$, $\Gamma_{G,m,M}$, and $\Gamma_{G,M,j}$ of $G_{i,m}$, $G_{m,M}$, and $G_{M,j}$ extending $\Gamma_{H,i,m}$, $\Gamma_{H,m,M}$, and $\Gamma_{H,M,j}$ in which the vertices with the smallest and largest $y$-coordinate are $u_{m}$ and $u_{M'}$, $u_{m}$ and $u_{M}$, and $u_{m'}$ and $u_{M}$, respectively, for some $M'\in \{i,\dots,m-1\}$ and $m'\in \{M+1,\dots,j\}$. We are going to glue together these drawings in order to construct an upward planar drawing $\Gamma_{G,i,j}$ of $G_{i,j}$ that extends $\Gamma_{H,i,j}$ and in which $u_m$ and $u_M$ are the vertices with the smallest and largest $y$-coordinate, respectively. However, before doing so, we need to perform some modifications on $\Gamma_{G,i,m}$, $\Gamma_{G,m,M}$, and $\Gamma_{G,M,j}$, while not altering the positions of the vertices of $H_{i,m}$, $H_{m,M}$, and $H_{M,j}$, respectively. 
		
		The first set of modifications aim to establish that: 
		\begin{enumerate}[(i)]
			\item $u_m$ is at the same point in $\Gamma_{G,i,m}$ and in $\Gamma_{G,m,M}$, and it has the smallest $y$-coordinate among all the vertices in $\Gamma_{G,i,m}$, $\Gamma_{G,m,M}$, and $\Gamma_{G,M,j}$; and
			\item $u_M$ is at the same point in $\Gamma_{G,m,M}$ and in $\Gamma_{G,M,j}$, and it has the largest $y$-coordinate among all the vertices in $\Gamma_{G,i,m}$, $\Gamma_{G,m,M}$, and $\Gamma_{G,M,j}$. 
		\end{enumerate}
		Note that $u_m$ is the vertex with the smallest $y$-coordinate in both $\Gamma_{G,i,m}$ and $\Gamma_{G,m,M}$, however there might be a vertex in $\Gamma_{G,M,j}$ whose $y$-coordinate is smaller than or equal to $y(u_m)$. A similar problem might occur for $u_M$. 
		
		\begin{itemize}
			\item If $u_m\notin V(H)$, then let $y^*$ be the smallest $y$-coordinate of any vertex in $\Gamma_{G,i,m}$, $\Gamma_{G,m,M}$, or $\Gamma_{G,M,j}$. Let $p\equiv (x,y)$, where $x$ and $y$ are real numbers such that $y<y^*$. Let $\epsilon>0$ be smaller than the vertical distance between any two vertices in $\Gamma_{G,i,m}$, $\Gamma_{G,m,M}$, and $\Gamma_{G,M,j}$. We modify $\Gamma_{G,i,m}$ as follows. We delete the part of $\Gamma_{G,i,m}$ inside a disk with radius $\epsilon$ centered at $u_m$. Since $\Gamma_{G,i,m}$ is an upward drawing, since $u_m$ is the vertex with the smallest $y$-coordinate in $\Gamma_{G,i,m}$, and by the choice of $\epsilon$, only $u_m$ and part of the edge $(u_{m},u_{m-1})$ is deleted from $\Gamma_{G,i,m}$. We place $u_m$ at $p$. We extend the part of the edge $(u_{m},u_{m-1})$ still in $\Gamma_{G,i,m}$ downwards, until it reaches a $y$-coordinate smaller than $y^*$ and larger than $y$, and then we connect it to $u_m$. We perform an analogous modification of $\Gamma_{G,m,M}$. 
			
			\begin{figure}[htb]
				\centering
				\begin{subfigure}{.45\textwidth}
					\centering
					\includegraphics[scale=.6]{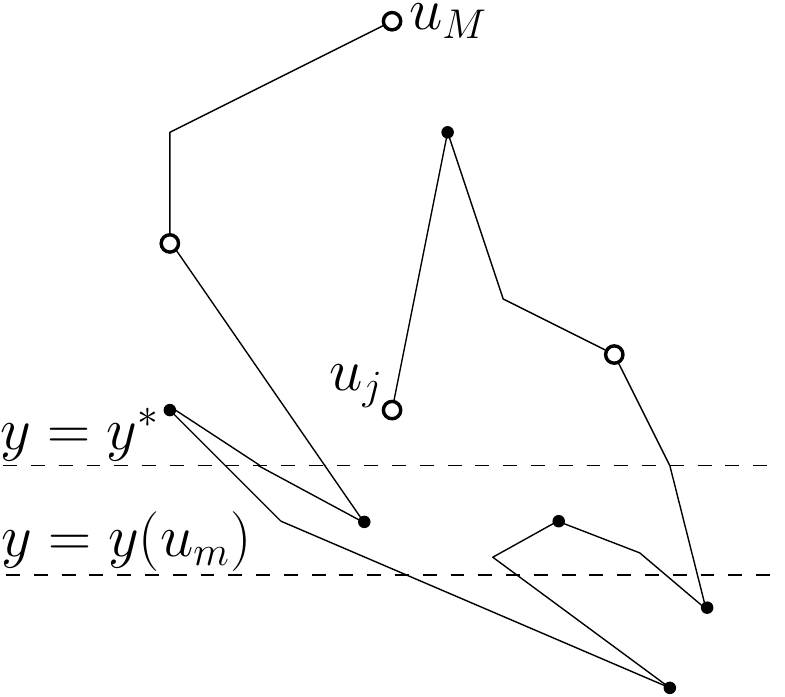}
				\end{subfigure}
				\begin{subfigure}{.45\textwidth}
					\centering
					\includegraphics[scale=.6]{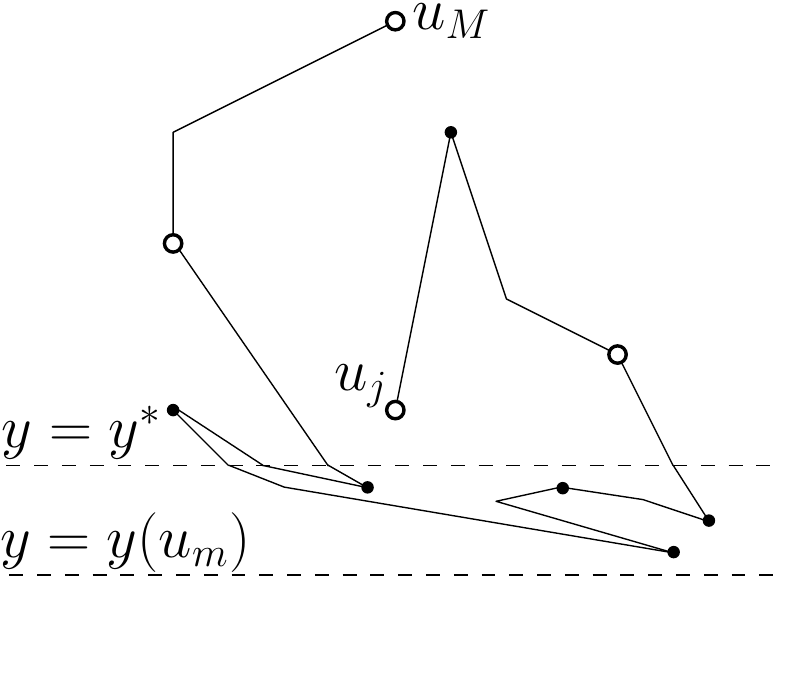}
				\end{subfigure}
				\caption{Modifying $\Gamma_{G,M,j}$ if $u_m\in V(H)$. Illustrations (a) and (b) show $\Gamma_{G,M,j}$ before and after the modification, respectively.}
				\label{fig:paths-claim1-vertical}
			\end{figure}
			
			\item If $u_m\in V(H)$, as in~\cref{fig:paths-claim1-vertical}, then $u_m$ is at the same point in $\Gamma_{G,i,m}$ and $\Gamma_{G,m,M}$, as such a point belongs to $\Gamma_{H,i,m}$ and $\Gamma_{H,m,M}$. Let $y^*$ be a real number larger than $y(u_m)$ and smaller than the $y$-coordinate of any vertex in $V(H)\cap \Gamma_{H,M,j}$. Such a real number exists by Condition~(4). Vertically scale down the part of $\Gamma_{G,M,j}$ in the half-plane $y\leq y^*$, while keeping fixed the points on the line $y=y^*$. As long as the scale factor is sufficiently small, the entire drawing $\Gamma_{G,M,j}$ lies above the line $y=y(u_m)$ after the scaling. 
		\end{itemize}


		The drawings $\Gamma_{G,i,m}$, $\Gamma_{G,m,M}$, and $\Gamma_{G,M,j}$ can be modified analogously in order to accomplish property~(ii) above. 
		
		
		Gluing together $\Gamma_{G,i,m}$, $\Gamma_{G,m,M}$, and $\Gamma_{G,M,j}$ results in an upward drawing $\Gamma'_{G,i,j}$ of $G_{i,j}$ that extends $\Gamma_{H,i,j}$ and in which $u_m$ and $u_M$ are the vertices with the smallest and largest $y$-coordinate, respectively. However, two problems arise. First, $\Gamma'_{G,i,j}$ might contain crossings; indeed, while each of $\Gamma_{G,i,m}$, $\Gamma_{G,m,M}$, and $\Gamma_{G,M,j}$ is planar, two edges from different graphs among $G_{i,m}$, $G_{m,M}$, and $G_{M,j}$ might cross each other. Second, the left-to-right order of the edges outgoing at $u_m$ in $\Gamma'_{G,i,j}$ might not correspond to $\mathcal S(u_m)$ and the left-to-right order of the edges incoming at $u_M$ in $\Gamma'_{G,i,j}$ might not correspond to $\mathcal P(u_M)$. 
		
		We show that both these problems can be overcome by redrawing the curves representing the edges of $G_{i,m}$, $G_{m,M}$, and $G_{M,j}$, while leaving the position of every vertex unaltered. Refer to~\cref{fig:paths-claim1-horizontal}. Assume first that  $\mathcal S(u_m)=[u_{m-1},u_{m+1}]$ and $\mathcal P(u_M)=[u_{M-1},u_{M+1}]$.
		
		\begin{figure}[htb]
			\centering
			\begin{subfigure}{.3\textwidth}
				\centering
				\includegraphics[height=\textwidth]{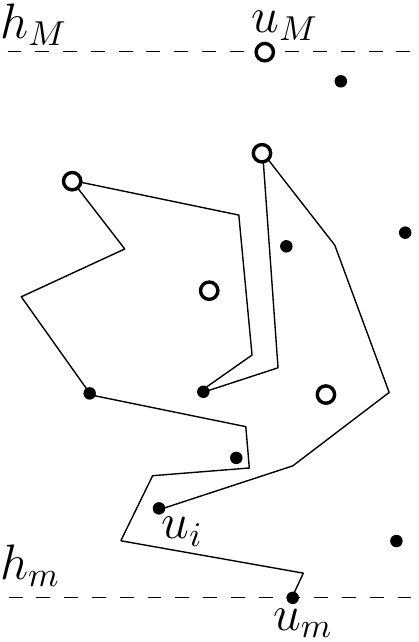}
			\end{subfigure}	
			\begin{subfigure}{.3\textwidth}
				\centering
				\includegraphics[height=\textwidth]{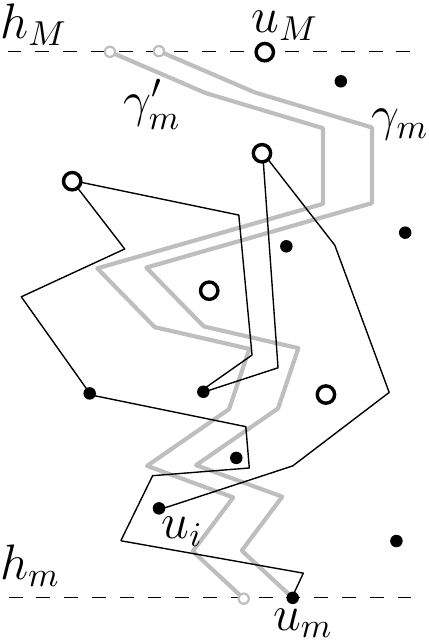}
			\end{subfigure}
			\begin{subfigure}{.3\textwidth}
				\centering
				\includegraphics[height=\textwidth]{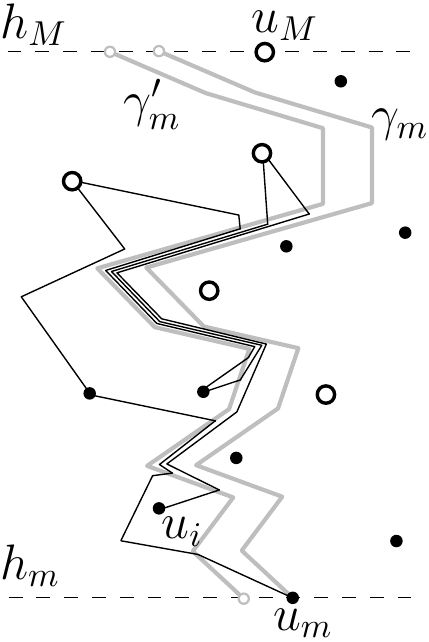}
			\end{subfigure}
			\caption{Modifications to $\Gamma'_{G,i,j}$ that yield to an upward planar drawing of $G_{i,m}$ which lies in the interior of $R_{i,m}$, except at $u_m$. (a) The drawing $\Gamma_{G,i,m}$ together with the vertices of $G_{m+1,j}$, as placed in $\Gamma'_{G,i,j}$; the edges incident to these vertices are not shown. (b) The curves $\gamma_m$ and $\gamma'_m$. (c) Redrawing parts of the edges in $G_{i,m}$ so that they lie in $R_{i,m}$.}
			\label{fig:paths-claim1-horizontal}
		\end{figure}
		
		Let the vertices of $G_{i,j}$ be placed as in $\Gamma'_{G,i,j}$. Let $h_m$ and $h_M$ be the horizontal lines through $u_m$ and $u_M$. Let $S$ be the horizontal strip delimited by $h_m$ and $h_M$. Let $\gamma_m$ be a $y$-monotone curve that connects $u_m$ with a point of $h_M$, that does not pass through any vertex of $G_{i,j}$ other than $u_m$, and such that the part of $S$ to the left of $\gamma_m$ contains in its interior all the vertices of $G_{i,m-1}$ and no vertex of $G_{m+1,j}$. Analogously, let $\gamma_M$ be a $y$-monotone line that connects $u_M$ with a point of $h_m$, that does not pass through any vertex of $G_{i,j}$ other than $u_M$, that does not intersect $\gamma_m$, and such that the part of $S$ to the right of $\gamma_M$ contains in its interior all the vertices of $G_{M+1,j}$ and no vertex of $G_{i,M-1}$. It is easy to see that such lines $\gamma_m$ and $\gamma_M$ exist; in particular, notice that there is no constraint on how these lines might intersect the edges of $G_{i,j}$ in $\Gamma'_{G,i,j}$.  
		
		The lines $\gamma_m$ and $\gamma_M$ partition $S$ into three regions: one unbounded region $R_{i,m}$ to the left of  $\gamma_m$, one bounded region $R_{m,M}$ between $\gamma_m$ and $\gamma_M$, and one unbounded region $R_{M,j}$ to the right of $\gamma_M$. Notice that $R_{i,m}$, $R_{m,M}$, and $R_{M,j}$ contain in their interiors all the points at which the vertices of $G_{i,m}$, $G_{m,M}$, and $G_{M,j}$ are placed in $\Gamma'_{G,i,j}$, respectively, except for the vertex $u_m$ which is on the common boundary of $R_{i,m}$ and $R_{m,M}$, and for the vertex $u_M$ which is on the common boundary of $R_{m,M}$ and $R_{M,j}$. We now redraw the edges of $G_{i,m}$, $G_{m,M}$, and $G_{M,j}$ in the interiors of $R_{i,m}$, $R_{m,M}$, and $R_{M,j}$, respectively. 
		
		We construct an upward planar drawing of $G_{i,m}$ inside $R_{i,m}$ by modifying $\Gamma_{G,i,m}$ via a homeomorphism of the plane which does not alter the positions of the vertices of $G_{i,m}$. More precisely, this can be done as follows. Let $\epsilon>0$ be the minimum horizontal distance between a vertex of $G_{i,m}$ different from $u_m$ and $\gamma_m$. Let $\gamma'_m$ be the translation of $\gamma'_m$ by $\epsilon/2$ to the left. Delete the part of $\Gamma_{G,i,m}$ to the right of $\gamma'_m$, except for the vertex $u_m$. Now any $y$-monotone curve that has been deleted from $\Gamma_{G,i,m}$ and that used to connect two points $p$ and $q$ on $\gamma'_m$ can be replaced by a $y$-monotone curve that connects $p$ and $q$ and that lies to the right of $\gamma'_m$, except at $p$ and $q$, and to the left of $\gamma_m$. The end-points of any two such curves do not alternate along $\gamma'_m$, given that $\Gamma_{G,i,m}$ is planar, hence all such curves can be drawn without intersections. The curves that used to connect a point $p$ with $u_m$ can be similarly redrawn to the right of $\gamma'_m$, except at $p$, and to the left of $\gamma_m$, except at $u_m$. 
		
		A similar modification allows us to construct an upward planar drawing of $G_{m,M}$ inside $R_{m,M}$ and an upward planar drawing of $G_{M,j}$ inside $R_{M,j}$. Since the interiors of these regions are disjoint, the edges of two different graphs among $G_{i,m}$, $G_{m,M}$, and $G_{M,j}$ do not intersect, except for the edges $(u_{m},u_{m-1})$ and $(u_{m},u_{m+1})$, which intersect at $u_m$, and for the edges $(u_{M},u_{M-1})$ and $(u_{M},u_{M+1})$, which intersect at $u_M$. Further, the left-to-right order of the edges outgoing at $u_m$ is $(u_{m},u_{m-1}),(u_{m},u_{m+1})$, which corresponds to $\mathcal S(u_m)=[u_{m-1},u_{m+1}]$, and the left-to-right order of the edges incoming at $u_M$ is $(u_{M-1},u_{M}),(u_{M+1},u_{M})$, which corresponds to $\mathcal P(u_M)=[u_{M-1},u_{M+1}]$. The drawing $\Gamma_{G,i,j}$ obtained as the union of the constructed drawings of $G_{i,m}$, $G_{m,M}$, and $G_{M,j}$ is hence the desired drawing of $G_{i,j}$.
		
		If $\mathcal S(u_m)=[u_{m+1},u_{m-1}]$ and $\mathcal P(u_M)=[u_{M+1},u_{M-1}]$, the construction is symmetric. In particular, $\gamma_m$ is defined as before, except that the part of $S$ to the {\em right} (and not to the left) of $\gamma_m$ contains in its interior all the vertices of $G_{i,m-1}$ and no vertex of $G_{m+1,j}$, and the part of $S$ to the {\em left} (and not to the right) of $\gamma_M$ contains in its interior all the vertices of $G_{M+1,j}$ and no vertex of $G_{i,M-1}$.
		
		If $\mathcal S(u_m)=[u_{m-1},u_{m+1}]$ and $\mathcal P(u_M)=[u_{M+1},u_{M-1}]$, then $\gamma_m$ is defined as in the case in which $\mathcal S(u_m)=[u_{m-1},u_{m+1}]$ and $\mathcal P(u_M)=[u_{M-1},u_{M+1}]$. However, the definition of $\gamma_M$ is now slightly different. Indeed, $\gamma_M$ is now a $y$-monotone line that connects $u_M$ with $u_m$ (and not just with any point of $h_m$), that does not pass through any vertex of $G_{i,j}$ other than $u_M$ and $u_m$, that does not intersect $\gamma_m$, except at $u_m$, and such that the part of $S$ to the right of $\gamma_M$ contains in its interior all the vertices of $G_{m+1,M-1}$ and no vertex of $G_{i,m-1}$ or $G_{M+1,j}$. The rest of the construction is analogous to the previous cases. 
		
		Finally, if $\mathcal S(u_m)=[u_{m+1},u_{m-1}]$ and $\mathcal P(u_M)=[u_{M-1},u_{M+1}]$, then the construction is symmetric to the case in which $\mathcal S(u_m)=[u_{m-1},u_{m+1}]$ and $\mathcal P(u_M)=[u_{M+1},u_{M-1}]$.

		This completes the proof of the claim. 
	\end{proof}
	
	{\em Suppose next that one of $u_m$ and $u_M$ is an end-vertex of $G_{i,j}$, while the other one is not.} Assume that $u_m=u_i$, the cases in which $u_m=u_j$, $u_M=u_i$, or $u_M=u_j$ can be treated analogously. Recall that $m\neq M$, hence $i<M<j$. We have the following.
	
	\begin{myclaim} \label{cl:paths-one-endv}
		$t(u_i,u_j,u_i,u_M)=$ {\sc true} if and only if there exists an index $m'\in \{M+1,\dots,j\}$ such that the following conditions hold true:
		\begin{enumerate}[(1)]
			\item $t(u_i,u_M,u_i,u_M)=$ {\sc true};
			\item $t(u_M,u_j,u_{m'},u_M)=$ {\sc true}; and
			\item either $u_i$ does not belong to $H$ or $u_i$ has the smallest $y$-coordinate among the vertices in~$\Gamma_H$.
		\end{enumerate}  
	\end{myclaim}  
	
	\begin{proof}
		The proof is very similar to (and in fact simpler than) the proof of~\cref{cl:paths-no-endv}. 
		
		Consider an upward planar drawing $\Gamma_{G,i,j}$ of $G_{i,j}$ in which $u_m$ and $u_M$ are the vertices with the smallest and largest $y$-coordinate, respectively. Restricting $\Gamma_{G,i,j}$ to the vertices and edges of $G_{i,M}$ (of $G_{M,j}$) provides an upward planar drawing $\Gamma_{G,i,M}$ of $G_{i,M}$ (resp.\ $\Gamma_{G,M,j}$ of $G_{M,j}$) extending $\Gamma_{H,i,M}$ (resp.\ $\Gamma_{H,M,j}$) in which the vertices with the smallest and largest $y$-coordinate are $u_i$ and $u_M$ (resp.\ $u_{m'}$ and $u_M$, for some $m'\in \{M+1,\dots,j\}$). This proves the necessity of Conditions~(1) and~(2). The property that $u_m$ is the vertex with the smallest $y$-coordinate in $\Gamma_{G,i,j}$ implies the necessity of Condition~(3).
		
		In order to prove the sufficiency we start from upward planar drawings $\Gamma_{G,i,M}$ and $\Gamma_{G,M,j}$ of $G_{i,M}$ and $G_{M,j}$ extending $\Gamma_{H,i,M}$ and $\Gamma_{H,M,j}$ in which the vertices with the smallest and largest $y$-coordinate are $u_{i}$ and $u_{M}$, and $u_{m'}$ and $u_{M}$, respectively, for some $m'\in \{M+1,\dots,j\}$. These drawings exist by Conditions~(1) and~(2). We then modify $\Gamma_{G,i,M}$ and/or $\Gamma_{G,M,j}$ so that the following properties hold true: 
		\begin{enumerate}[(i)]
			\item $u_M$ is at the same point in $\Gamma_{G,i,M}$ and $\Gamma_{G,M,j}$; and
			\item $u_i$ has the smallest $y$-coordinate among all the vertices in $\Gamma_{G,i,M}$ and $\Gamma_{G,M,j}$.
		\end{enumerate}
		In order to accomplish property~(ii) we act as follows. Let $y^*$ be a real number larger than $y(u_i)$ and smaller than the $y$-coordinate of any vertex in $V(H)\cap \Gamma_{H,M,j}$. Such a real number exists by Condition~(3). Vertically scale down the part of $\Gamma_{G,M,j}$ in the half-plane $y\leq y^*$, while keeping fixed the points on the line $y=y^*$. As long as the scale factor is sufficiently small, the entire drawing $\Gamma_{G,M,j}$ has a $y$-coordinate larger than $y(u_i)$ after the scaling. The satisfaction of property~(i) does not require any modifications to $\Gamma_{G,i,M}$ and $\Gamma_{G,M,j}$ if $u_M\in V(H)$; on the other hand, if $u_M\notin V(H)$, then $u_M$ and parts of its incident edges in $\Gamma_{G,i,M}$ and $\Gamma_{G,M,j}$ are redrawn so to let $u_M$ be at the same point in $\Gamma_{G,i,M}$ and $\Gamma_{G,M,j}$. 
		
		Now gluing together $\Gamma_{G,i,M}$ and $\Gamma_{G,M,j}$ results in an upward drawing $\Gamma'_{G,i,j}$ of $G_{i,j}$ that extends $\Gamma_{H,i,j}$ in which $u_i$ and $u_M$ are the vertices with the smallest and largest $y$-coordinate, respectively. However, $\Gamma'_{G,i,j}$ might contain crossings and the left-to-right order of the edges incoming at $u_M$ in $\Gamma'_{G,i,j}$ might not correspond to $\mathcal P(u_M)$. Hence, we redraw the curves representing the edges of $G_{i,M}$ and $G_{M,j}$, while leaving the position of every vertex unaltered. This is done by defining two internally-disjoint regions $R_{i,M}$ and $R_{M,j}$ that are separated by a $y$-monotone curve $\gamma_M$ through $u_M$ and that contain in their interiors the points at which the vertices of $G_{i,M}$ and $G_{M,j}$ are placed in $\Gamma'_{G,i,j}$, except for $u_M$ which is on the boundary of both such regions; whether $R_{i,M}$ is to the left or to the right of $R_{M,j}$ depends on whether $\mathcal P(u_M)=[u_{M-1},u_{M+1}]$ or $\mathcal P(u_M)=[u_{M+1},u_{M-1}]$, respectively. We redraw the edges of $G_{i,M}$ and $G_{M,j}$ inside $R_{i,M}$ and $R_{M,j}$, respectively. The drawing $\Gamma_{G,i,j}$ obtained as the union of the constructed drawings of $G_{i,M}$ and $G_{M,j}$ is the desired drawing of $G_{i,j}$. This completes the proof of the claim. 
	\end{proof}
	
	{\em Suppose finally that both $u_m$ and $u_M$ are end-vertices of $G_{i,j}$.} Assume that $u_m=u_i$ and $u_M=u_j$, the other case is symmetric. We have the following. 
	
	\begin{myclaim} \label{cl:paths-two-endv}
		$t(u_i,u_j,u_i,u_j)=$ {\sc true} if and only if there exist indices $M'\in \{i+1,\dots,j-2\}$ and $m'\in \{M'+1,\dots,j-1\}$ such that the following conditions hold true:
		\begin{enumerate}[(1)]
			\item $t(u_i,u_{M'},u_i,u_{M'})=$ {\sc true};
			\item $t(u_{M'},u_{m'},u_{m'},u_{M'})=$ {\sc true};
			\item $t(u_{m'},u_j,u_{m'},u_j)=$ {\sc true};
			\item either $u_i$ does not belong to $H$ or $u_i$ has the smallest $y$-coordinate among the vertices in $\Gamma_H$;  
			\item either $u_j$ does not belong to $H$ or $u_j$ has the largest $y$-coordinate among the vertices in $\Gamma_H$; and
			\item either $\mathcal P(u_{M'})=[u_{M'-1},u_{M'+1}]$ and $\mathcal S(u_{m'})=[u_{m'-1},u_{m'+1}]$, or $\mathcal P(u_{M'})=[u_{M'+1},u_{M'-1}]$ and $\mathcal S(u_{m'})=[u_{m'+1},u_{m'-1}]$.
		\end{enumerate}  
	\end{myclaim}  
	
	\begin{proof}
		We first prove the necessity. Suppose that $t(u_i,u_j,u_i,u_j)=$ {\sc true}, hence an upward planar drawing $\Gamma_{G,i,j}$ of $G_{i,j}$ exists that extends $\Gamma_{H,i,j}$ and in which $u_i$ and $u_j$ are the vertices with the smallest and largest $y$-coordinate, respectively. Since $G_{i,j}$ is not a monotone path, there exist internal vertices of $G_{i,j}$ which are sinks. Among all these vertices, let $u_{M'}$ be the one with the largest $y$-coordinate in $\Gamma_{G,i,j}$. Further, since $u_{M'}$ and $u_j$ are both sinks in $G_{i,j}$ (note that if $u_j$ were not a sink, then an upward drawing of $G_{i,j}$ in which $u_j$ is the vertex with the largest $y$-coordinate would not exist), there exist internal vertices of $G_{M',j}$ which are sources. Among all these vertices, let $u_{m'}$ be the one with the smallest $y$-coordinate in $\Gamma_{G,i,j}$. Note that $M'\in \{i+1,\dots,m-2\}$ and $m'\in \{M'+1,\dots,j-1\}$.
		
		Restricting $\Gamma_{G,i,j}$ to the vertices and edges of $G_{i,M'}$ yields an upward planar drawing $\Gamma_{G,i,M'}$ of $G_{i,M'}$ that extends $\Gamma_{H,i,M'}$ and in which $u_i$ and $u_{M'}$ are the vertices with the smallest and largest $y$-coordinate, respectively. In particular, no vertex $u$ of $\Gamma_{G,i,M'}$ has a $y$-coordinate larger than $u_{M'}$, as otherwise the sink $v$ such that there is a monotone path from $u$ to $v$ in $G_{i,M'}$ (possibly such a path is a single vertex if $u$ is a sink itself) would have a $y$-coordinate larger than $u_{M'}$, contradicting the choice of $u_{M'}$. This proves Condition~(1); the proofs of Conditions~(2) and~(3) are analogous. By assumption $u_i$ is the vertex with the smallest $y$-coordinate in $\Gamma_{G,i,j}$, hence Condition~(4) follows, given that $\Gamma_{G,i,j}$ extends $\Gamma_{H,i,j}$. Condition~(5) is proved analogously. 
		
		\begin{figure}[htb]
			\centering
			\includegraphics[scale=.7]{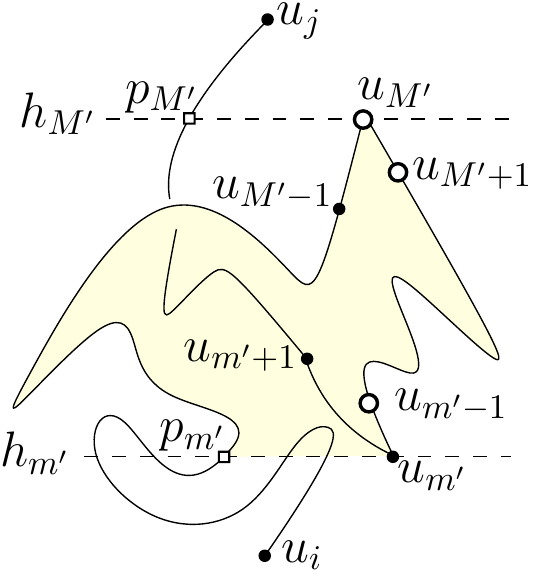}
			\label{fi:bridge1}
			\caption{Illustration for the proof of the necessity of Condition (6). The interior of $\mathcal C_m$ is colored light yellow.}
			\label{fig:paths-claim3-necessity}
		\end{figure}

		In order to prove the necessity of Condition (6) suppose, for a contradiction, that $\mathcal P(u_{M'})=[u_{M'-1},u_{M'+1}]$ and $\mathcal S(u_{m'})=[u_{m'+1},u_{m'-1}]$ (the case in which $\mathcal P(u_{M'})=[u_{M'+1},u_{M'-1}]$ and $\mathcal S(u_{m'})=[u_{m'-1},u_{m'+1}]$ can be treated analogously). Refer to~\cref{fig:paths-claim3-necessity}. Let $h_{M'}$ and $h_{m'}$ be the horizontal lines passing through $u_{M'}$ and $u_{m'}$ in $\Gamma_{G,i,j}$, respectively. Since $y(u_i)<y(u_{m'})<y(u_{M'})$, it follows that $G_{i,M'}$ crosses $h_{m'}$. Order the crossing points between $G_{i,M'}$ and $h_{m'}$ as they are encountered when walking along $G_{i,M'}$ from $u_i$ to $u_{M'}$ and let $p_{m'}$ be the last point in this order. Analogously, let $p_{M'}$ be the crossing point between $G_{m',j}$ and $h_{M'}$ that is encountered last when walking along $G_{m',j}$ from $u_j$ to $u_{m'}$. Thus the part of $\Gamma_{G,i,j}$ connecting $p_{m'}$ with $p_{M'}$ entirely lies in the strip $\mathcal H$ delimited by $h_{m'}$ and $h_{M'}$ (note that $\Gamma_{G,M',m'}$ does not cross $h_{m'}$ or $h_{M'}$, since $u_{m'}$ is the source with the smallest $y$-coordinate and $u_{M'}$ is the sink with the largest $y$-coordinate~in~$\Gamma_{G,M',m'}$). 
		Let $\mathcal C_m$ be the closed curve composed of the part of $\Gamma_{G,i,j}$ between $p_{m'}$ and $u_{m'}$ (this includes $u_{M'}$) and of the part of $h_{m'}$ between $p_{m'}$ and $u_{m'}$. Since $\mathcal P(u_{M'})=[u_{M'-1},u_{M'+1}]$, since the part of $\Gamma_{G,i,j}$ between $p_{m'}$ and $u_{m'}$ lies in $\mathcal H$, and by the planarity of $\Gamma_{G,i,j}$, the points $p_{m'}$, $u_{M'}$, and $u_{m'}$ appear in this clockwise order along $\mathcal C_m$. Further, since $\mathcal S(u_{m'})=[u_{m'+1},u_{m'-1}]$ and by the upward planarity of $\Gamma_{G,i,j}$, we have that $u_{m'+1}$ lies inside $\mathcal C_m$ in $\Gamma_{G,i,j}$; on the other hand, $u_j$ is outside $\mathcal C_m$ in $\Gamma_{G,i,j}$. By the Jordan curve's theorem, the path $G_{m'+1,j}$ crosses $\mathcal C_m$. Since the path $G_{m'+1,j}$ lies in the half-plane $y\geq y(u_{m'})$, it does not cross $h_{m'}$, hence it crosses the part of $\Gamma_{G,i,j}$ between $p_{m'}$ and $u_{m'}$, a contradiction to the planarity of $\Gamma_{G,i,j}$.

		
		The proof of the sufficiency is similar to the ones of~\cref{cl:paths-no-endv,cl:paths-one-endv}. Namely, we start from upward planar drawings $\Gamma_{G,i,M'}$, $\Gamma_{G,M',m'}$, and $\Gamma_{G,m',j}$ of $G_{i,M'}$, $G_{M',m'}$, and $G_{m',j}$ extending $\Gamma_{H,i,M'}$, $\Gamma_{H,M',m'}$, and $\Gamma_{H,m',j}$ in which the vertices with the smallest and largest $y$-coordinate are $u_{i}$ and $u_{M'}$, $u_{m'}$ and $u_{M'}$, and $u_{m'}$ and $u_j$, respectively, for some $M'\in \{i+1,\dots,j-2\}$ and $m'\in \{M'+1,\dots,j-1\}$. These drawings exist by Conditions~(1)--(3). 
		
		We then modify $\Gamma_{G,i,M'}$, $\Gamma_{G,M',m'}$, and $\Gamma_{G,m',j}$ so that the following properties hold:
		\begin{enumerate}[(i)]
			\item $u_{M'}$ is at the same point in $\Gamma_{G,i,M'}$ and $\Gamma_{G,M',m'}$;
			\item $u_{m'}$ is at the same point in $\Gamma_{G,M',m'}$ and $\Gamma_{G,m',j}$;
			\item $u_i$ has the smallest $y$-coordinate among all the vertices in $\Gamma_{G,i,M'}$, $\Gamma_{G,M',m'}$, and $\Gamma_{G,m',j}$; and
			\item $u_j$ has the largest $y$-coordinate among all the vertices in $\Gamma_{G,i,M'}$, $\Gamma_{G,M',m'}$, and $\Gamma_{G,m',j}$.
		\end{enumerate}  
		The satisfaction of property~(i) does not require any modifications to $\Gamma_{G,i,M'}$ and $\Gamma_{G,M',m'}$ if $u_{M'}\in V(H)$; on the other hand, if $u_{M'}\notin V(H)$, then $u_{M'}$ and parts of its incident edges in $\Gamma_{G,i,M'}$ and $\Gamma_{G,M',m'}$ are redrawn (in particular $u_{M'}$ is placed higher than it was in $\Gamma_{G,i,M'}$ and $\Gamma_{G,M',m'}$) so to let $u_{M'}$ be at the same point in $\Gamma_{G,i,M'}$ and $\Gamma_{G,M',m'}$. Property~(ii) is satisfied similarly. If $u_i\notin V(H)$, then Property~(iii) can be satisfied by redrawing $u_i$ (lower than every other vertex in $\Gamma_{G,i,M'}$, $\Gamma_{G,M',m'}$, and $\Gamma_{G,m',j}$) and part of its incident edge in $\Gamma_{G,i,M'}$. If $u_i\in V(H)$, then a real number $y^*$ larger than $y(u_i)$ and smaller than the $y$-coordinate of any vertex in $V(H)\cap \Gamma_{H,m',j}$ is chosen, which is possible by Condition~(4), and the part of $\Gamma_{G,m',j}$ in the half-plane $y\leq y^*$ is vertically scaled down, while keeping fixed the points on the line $y=y^*$. As long as the scale factor is sufficiently small, the entire drawing $\Gamma_{G,m',j}$ lies above the line $y=y(u_i)$ after the scaling. Property~(iv) is ensured analogously exploiting Condition~(5).

		
		Gluing together $\Gamma_{G,i,M'}$, $\Gamma_{G,M',m'}$, and $\Gamma_{G,m',j}$ results in an upward drawing $\Gamma'_{G,i,j}$ of $G_{i,j}$ that extends $\Gamma_{H,i,j}$ and in which $u_i$ and $u_j$ are the vertices with the smallest and largest $y$-coordinate, respectively. However, $\Gamma'_{G,i,j}$ might contain crossings, the left-to-right order of the edges incoming at $u_{M'}$ in $\Gamma'_{G,i,j}$ might not correspond to $\mathcal P(u_{M'})$, and the left-to-right order of the edges outgoing from $u_{m'}$ in $\Gamma'_{G,i,j}$ might not correspond to $\mathcal S(u_{m'})$. Hence, we redraw the curves representing the edges of $G_{i,M'}$, $G_{M',m'}$, and $G_{m',j}$, while leaving the position of every vertex unaltered. This is done by defining three internally-disjoint regions $R_{i,M'}$, $R_{M',m'}$, and $R_{m',j}$ containing in their interiors the points at which the vertices of $G_{i,M'}$, $G_{M',m'}$, and $G_{m',j}$ are placed in $\Gamma'_{G,i,j}$, respectively, except for $u_{M'}$ which is on the boundary of both $R_{i,M'}$ and $R_{M',m'}$, and for $u_{m'}$ which is on the boundary of both $R_{M',m'}$ and $R_{m',j}$. These regions are delimited by the horizontal lines $h_i$ and $h_j$ through $u_i$ and $u_j$, and by two $y$-monotone curves $\gamma_{M'}$ and $\gamma_{m'}$ passing through $u_{M'}$ and $u_{m'}$, respectively, and connecting points on $h_i$ and $h_j$. The curves $\gamma_{M'}$ and $\gamma_{m'}$ are defined so that the regions $R_{i,M'}$, $R_{M',m'}$, and $R_{m',j}$ appear in this left-to-right order inside the horizontal strip delimited by $h_i$ and $h_j$ in the case in which $\mathcal P(u_{M'})=[u_{M'-1},u_{M'+1}]$ and $\mathcal S(u_{m'})=[u_{m'-1},u_{m'+1}]$, or so that they appear in the opposite left-to-right order in the case in which $\mathcal P(u_{M'})=[u_{M'+1},u_{M'-1}]$ and $\mathcal S(u_{m'})=[u_{m'+1},u_{m'-1}]$. One of the two cases happens, because of Condition~(6). We now redraw the edges of $G_{i,M'}$, $G_{M',m'}$, and $G_{m',j}$ inside $R_{i,M'}$, $R_{M',m'}$, and $R_{m',j}$, respectively. The drawing $\Gamma_{G,i,j}$ obtained as the union of the constructed drawings of $G_{i,M'}$, $G_{M',m'}$, and $G_{m',j}$ is the desired drawing of $G_{i,j}$. This completes the proof of the claim. 
	\end{proof}
	
	\cref{cl:paths-no-endv,cl:paths-one-endv,cl:paths-two-endv} show how to compute the value of the entry $t(u_i,u_j,u_m,u_M)$, if $G_{i,j}$ is not a monotone path, based on the structure of the instance $\langle G_{i,j}, H_{i,j}, \Gamma_{H,i,j} \rangle$ and on the value of the entries $t(u_{i'},u_{j'},u_{m'},u_{M'})$ with $j'-i'<j-i$. Eventually the dynamic programming will compute the values of the entries $t(u_1,u_n,u_m,u_M)$, for all $1 \leq m \leq n$ and $1 \leq M \leq n$ with $m\neq M$; then the instance $\langle G, H, \Gamma_H \rangle$ of the {\sc UPE-FUE} problem is positive if and only if at least one of these entries \mbox{has value {\sc true}.} 
	
	We now show how to implement the described algorithm so that it runs in $O(n^4)$ time. Observe that there are $O(n^4)$ entries $t(u_i,u_j,u_m,u_M)$ whose value has to be computed. 
	
	We start by computing, for each entry $t(u_i,u_j,u_m,u_M)$, the smallest and largest $y$-coordinate of a vertex in $V(H)\cap \Gamma_{H,i,j}$, if any such a vertex exists. This can be done in $O(1)$ time per entry by induction on $j-i$. Indeed, if $j-i=1$, then at most two coordinates have to be compared in order to determine such two values. If $j-i>1$, then the two values either coincide with those computed for $t(u_{i+1},u_j,u_m,u_M)$, if $u_i\notin V(H)$, or are computed by comparing the two values computed for $t(u_{i+1},u_j,u_m,u_M)$ with $y(u_i)$, if $u_i\in V(H)$.  
	
	The computation of the {\sc true}-{\sc false} values for the entries $t(u_i,u_j,u_m,u_M)$ such that $G_{i,j}$ is a monotone path is done in $O(n^4)$ time as follows. First, for any pair $(i,j)$ of integers such that $1 \leq i<j \leq n$ we check whether $G_{i,j}$ is a monotone path. There are $O(n^2)$ such pairs of integers; further, for each pair $(i,j)$, the monotonicity of $G_{i,j}$ can be tested in $O(n)$ time by simply checking whether all the edges of $G_{i,j}$ are directed from $u_i$ to $u_j$, or whether they are all directed from $u_j$ to $u_i$. For each pair $(i,j)$ such that $G_{i,j}$ is monotone and its edges are directed from $u_i$ to $u_j$ (from $u_j$ to $u_i$), we give value {\sc true} to the entry $t(u_i,u_j,u_i,u_j)$ (resp.\ $t(u_i,u_j,u_j,u_i)$) and {\sc false} to all the other entries $t(u_i,u_j,u_m,u_M)$; there are $O(n^2)$ such entries, hence the quartic running time. 
	
	We now turn to the computation of the values of the entries $t(u_i,u_j,u_m,u_M)$ such that $G_{i,j}$ is not a monotone path. Consider a pair $(i,j)$ of integers such that $1 \leq i<j \leq n$, such that the values of the entries $t(u_i,u_j,u_m,u_M)$ have not been determined yet, and such that the values of the entries $t(u_{i'},u_{j'},u_{m'},u_{M'})$ have been determined for all the pairs $(i',j')$ of integers such that $j'-i'<j-i$. There are $O(n^2)$ such pairs $(i,j)$. 
	
	We perform the following preliminary check. For each index $m\in \{i+1,\dots,j-1\}$, we check whether any of the entries $t(u_{i},u_{m},u_{m},u_{M'})$ has value {\sc true}, over all the indices $M'\in \{i,\dots,m-1\}$. There are $O(n)$ indices $m$ to be considered and for each of them we check the value of $O(n)$ entries, hence this takes $O(n^2)$ time for the pair $(i,j)$, and hence $O(n^4)$ time over all the pairs $(i,j)$. Analogously, we check for each index $M\in \{i+1,\dots,j-1\}$ whether any of the entries $t(u_{i},u_{M},u_{m'},u_{M})$ has value {\sc true}, for each index $m\in \{i+1,\dots,j-1\}$ whether any of the entries $t(u_{m},u_{j},u_{m},u_{M'})$ has value {\sc true}, and for each index $M\in \{i+1,\dots,j-1\}$ whether any of the entries $t(u_{M},u_{j},u_{m'},u_{M})$ has value {\sc true}. This information can be stored in a separate table, whose entries are of the form $t(u_i,u_m,u_m,\cdot)$, $t(u_i,u_M,\cdot,u_M)$, $t(u_{m},u_j,u_m,\cdot)$, and $t(u_M,u_j,\cdot,u_M)$.   
	
	By means of \cref{cl:paths-no-endv}, we now determine the values of the entries $t(u_i,u_j,u_m,u_M)$ such that neither $u_m$ nor $u_M$ is an end-vertex of $G_{i,j}$. For each entry $t(u_i,u_j,u_m,u_M)$ such that $i<m<M<j$ (the entries such that $i<M<m<j$ are dealt with analogously) we test in $O(1)$ time whether Conditions~1--3 are satisfied by checking whether $t(u_i,u_m,u_m,\cdot)=$ {\sc true}, whether $t(u_m,u_M,u_m,u_M)=$ {\sc true}, and whether $t(u_M,u_j,\cdot,u_M)=$ {\sc true}. Condition~4 can be tested in $O(1)$ time by checking whether $u_m \in V(H)$ and, in case it does, whether $y(u_m)$ is the smallest $y$-coordinate among the vertices in $V(H)\cap \Gamma_{H,i,j}$ -- this information was computed at the beginning of the algorithm. Condition~5 can be tested in $O(1)$ time analogously.
	
	Next, by means of~\cref{cl:paths-one-endv}, we determine the values of the entries $t(u_i,u_j,u_m,u_M)$ such that one of $u_m$ and $u_M$ is an end-vertex of $G_{i,j}$. For each entry $t(u_i,u_j,u_i,u_M)$ (the entries $t(u_i,u_j,u_j,u_M)$, $t(u_i,u_j,u_m,u_i)$, and $t(u_i,u_j,u_m,u_j)$ are dealt with analogously) we test in $O(1)$ time whether Conditions~1--2 are satisfied by checking whether $t(u_i,u_M,u_i,u_M)=$ {\sc true} and whether $t(u_M,u_j,\cdot,u_M)=$ {\sc true}. Condition~3 can be tested in $O(1)$ time by checking whether $u_m \in V(H)$ and, in case it does, whether $y(u_m)$ is the smallest $y$-coordinate among the vertices in $V(H)\cap \Gamma_{H,i,j}$. 
	
	Finally, by means of~\cref{cl:paths-two-endv}, we determine the values of the entries $t(u_i,u_j,u_m,u_M)$ such that both $u_m$ and $u_M$ are end-vertices of $G_{i,j}$. For each entry $t(u_i,u_j,u_i,u_j)$ (the entries $t(u_i,u_j,u_j,u_i)$ are dealt with analogously) we test in $O(n^2)$ time whether Conditions~1--6 are satisfied as follows. First, we test in $O(1)$ time whether Condition~4 is satisfied by checking whether $u_m \in V(H)$ and, in case it does, whether $y(u_m)$ is the smallest $y$-coordinate among the vertices in $V(H)\cap \Gamma_{H,i,j}$. Condition~5 can be tested in $O(1)$ time analogously. In order to test Conditions~1--3 and Condition~6, we consider all the pairs of indices $(M',m')$ with $M'\in \{i+1,\dots,j-2\}$ and $m'\in \{M'+1,\dots,j-1\}$. For each such pair $(M',m')$ we check in $O(1)$ time whether $t(u_i,u_{M'},u_i,u_{M'})=$ {\sc true}, whether $t(u_{M'},u_{m'},u_{m'},u_{M'})=$ {\sc true}, whether $t(u_{m'},u_j,u_{m'},u_j)=$ {\sc true}, and whether $\mathcal P(u_{M'})=[u_{M'-1},u_{M'+1}]$ and $\mathcal S(u_{m'})=[u_{m'-1},u_{m'+1}]$, or $\mathcal P(u_{M'})=[u_{M'+1},u_{M'-1}]$ and $\mathcal S(u_{m'})=[u_{m'+1},u_{m'-1}]$. Note that there are $O(n^2)$ entries $t(u_i,u_j,u_i,u_j)$; for each of them we consider $O(n^2)$ pairs of indices $(M',m')$, and then we check the above conditions in $O(1)$ time. Thus the total running time is $O(n^4)$. 
\end{proof}

By exploiting arguments analogous to those in the proof of~\cref{th:algo-UPE-FUE-path} we can extend our quartic-time algorithm to cycles.

\begin{restatable}{theorem}{TheoremAlgoUPEFUEcycle}\label{th:algo-UPE-FUE-cycle}
	The {\sc UPE-FUE} problem can be solved in $O(n^4)$ time for instances $\langle G, H, \Gamma_H \rangle$ such that $G$ is an $n$-vertex cycle with given upward embedding, $H$ contains no edges, and no two vertices share the same $y$-coordinate in $\Gamma_H$.
\end{restatable}

\begin{proof}
	Let $G=(u_1,\dots,u_n)$. Suppose that an upward planar drawing $\Gamma_G$ of $G$ extending $\Gamma_H$ exists. Let $y_M$ be the largest $y$-coordinate of a vertex in $\Gamma_G$. Then it can be assumed without loss of generality that there is a unique vertex $u_M$ in $\Gamma_G$ such that $y(u_M)=y_M$. Indeed, if more than one vertex has $y$-coordinate equal to $y_M$ in $\Gamma_G$, then at least one vertex $v$ exists such that $y(v)=y_M$ and $v\notin V(H)$. Hence $v$ can be moved upwards and its incident edges can be extended upwards as well, so that $v$ becomes the unique vertex with the largest $y$-coordinate in $\Gamma_G$. It can be analogously assumed that there is a vertex $u_m$ in $\Gamma_G$ whose $y$-coordinate is smaller than the one of every other vertex.
	
	Our strategy is to test, for every possible pair of vertices $(u_m,u_M)$ with $m,M\in \{1,\dots,n\}$ and with $m\neq M$, whether there is an upward planar drawing $\Gamma_G$ of $G$ extending $\Gamma_H$ in which the vertices with the smallest and largest $y$-coordinate are $u_m$ and $u_M$, respectively. For any pair $(u_m,u_M)$, the cycle $G$ consists of two directed paths connecting $u_m$ and $u_M$, call them $G_{m,M}=(u_m,u_{m+1},\dots,u_M)$ and $G_{M,m}=(u_M,u_{M+1},\dots,u_m)$, where indices are modulo $n$. Let $\Gamma_{H,m,M}$ and $\Gamma_{H,M,m}$ be the restrictions of $\Gamma_H$ to the vertices that belong to $G_{m,M}$ and $G_{M,m}$, respectively. The following claim is the key ingredient for the proof of the theorem.
	
	\begin{myclaim} \label{cl:cycles-two-paths}
		$G$ has an upward planar drawing extending $\Gamma_H$ in which the vertices with the smallest and largest $y$-coordinate are $u_m$ and $u_M$, respectively, if and only if:
		\begin{enumerate}[(1)]
			\item $G_{m,M}$ has an upward planar drawing extending $\Gamma_{H,m,M}$ in which the vertices with the smallest and largest $y$-coordinate are $u_m$ and $u_M$, respectively;
			\item $G_{M,m}$ has an upward planar drawing extending $\Gamma_{H,M,m}$ in which the vertices with the smallest and largest $y$-coordinate are $u_m$ and $u_M$, respectively; and
			\item either $\mathcal P(u_{M})=[u_{M-1},u_{M+1}]$ and $\mathcal S(u_{m})=[u_{m+1},u_{m-1}]$, or $\mathcal P(u_{M})=[u_{M+1},u_{M-1}]$ and $\mathcal S(u_{m})=[u_{m-1},u_{m+1}]$.
		\end{enumerate}
	\end{myclaim}  
	
	\begin{proof}
		Suppose that $G$ has an upward planar drawing $\Gamma_G$ extending $\Gamma_H$ in which the vertices with the smallest and largest $y$-coordinate are $u_m$ and $u_M$, respectively. Restricting $\Gamma_G$ to the vertices and edges of $G_{m,M}$ yields an upward planar drawing $\Gamma_{G,m,M}$ of $G_{m,M}$ that extends $\Gamma_{H,m,M}$ and in which $u_m$ and $u_{M}$ are the vertices with the smallest and largest $y$-coordinate, respectively. This proves the necessity of Condition~(1). The necessity of Condition~(2) is proved analogously. In order to prove the necessity of Condition (3) suppose, for a contradiction, that $\mathcal P(u_{M})=[u_{M-1},u_{M+1}]$ and $\mathcal S(u_{m})=[u_{m-1},u_{m+1}]$ (the case in which $\mathcal P(u_{M})=[u_{M+1},u_{M-1}]$ and $\mathcal S(u_{m})=[u_{m+1},u_{m-1}]$ can be treated analogously). Let $h_{M}$ and $h_{m}$ be the horizontal lines passing through $u_{M}$ and $u_{m}$ in $\Gamma_{G}$, respectively. The path $G_{m,M}$ divides the strip delimited by $h_{M}$ and $h_{m}$ into two regions $R_l$ and $R_r$, to the left and to the right of $G_{m,M}$, respectively. Since $\mathcal P(u_{M})=[u_{M-1},u_{M+1}]$ and $\mathcal S(u_{m})=[u_{m-1},u_{m+1}]$, and by the upward planarity of $\Gamma_{G}$, it follows that $u_{m-1}$ and $u_{M+1}$ are in $R_l$ and $R_r$, respectively. Hence, the path $G_{M+1,m-1}$ crosses the boundary of such regions. Since $u_m$ and $u_M$ have the smallest and largest $y$-coordinate among the vertices in $\Gamma_{G}$, and by the upwardness of $\Gamma_{G}$, it follows that $G_{M+1,m-1}$ does not cross $h_{M}$ or $h_{m}$, hence it crosses $G_{m,M}$, a contradiction to the planarity of $\Gamma_{G}$.

		
		The proof of the sufficiency is similar to the ones of \cref{cl:paths-no-endv,cl:paths-one-endv,cl:paths-two-endv}. Namely, we start from upward planar drawings $\Gamma_{G,m,M}$ and $\Gamma_{G,M,m}$ of $G_{m,M}$ and $G_{M,m}$ extending $\Gamma_{H,m,M}$ and $\Gamma_{H,M,m}$, respectively, in which the vertex with the smallest $y$-coordinate is $u_{m}$ and the vertex with the largest $y$-coordinate is $u_{M}$. These drawings exist by Conditions~(1) and~(2).  
		
		We then modify $\Gamma_{G,m,M}$ and $\Gamma_{G,M,m}$ so that the following properties hold:
		\begin{enumerate}[(i)]
			\item $u_{M}$ is at the same point in $\Gamma_{G,m,M}$ and $\Gamma_{G,M,m}$; and
			\item $u_{m}$ is at the same point in $\Gamma_{G,m,M}$ and $\Gamma_{G,M,m}$.
		\end{enumerate} 
		If $u_{M}\in V(H)$, then Property~(i) is already satisfied by $\Gamma_{G,m,M}$ and $\Gamma_{G,M,m}$; on the other hand, if $u_{M}\notin V(H)$, then we redraw $u_{M}$ and parts of its incident edges so to let $u_{M}$ be at the same point in $\Gamma_{G,m,M}$ and $\Gamma_{G,M,m}$. Property~(ii) is satisfied analogously.
		
		Gluing together $\Gamma_{G,m,M}$ and $\Gamma_{G,M,m}$ results in an upward drawing $\Gamma'_{G}$ of $G$ that extends $\Gamma_H$ and in which $u_m$ and $u_M$ are the vertices with the smallest and largest $y$-coordinate, respectively. However, $\Gamma'_{G}$ might contain crossings. Hence, we redraw the curves representing the edges of $G_{m,M}$ and $G_{M,m}$, while leaving the position of every vertex unaltered. This is done by defining two internally-disjoint regions $R_{m,M}$ and $R_{M,m}$ containing in their interiors the points at which the vertices of $G_{m,M}$ and $G_{M,m}$ are placed in $\Gamma'_{G}$, respectively, except for $u_m$ and $u_{M}$ which are on the boundary of both $R_{m,M}$ and $R_{M,m}$. These regions are delimited by the horizontal lines $h_m$ and $h_M$ through $u_m$ and $u_M$, and by a $y$-monotone curve $\gamma$ connecting $u_{m}$ with $u_{M}$. The curve $\gamma$ is defined so that $R_{m,M}$ is to the left of $R_{M,m}$ in the case in which $\mathcal P(u_{M})=[u_{M-1},u_{M+1}]$ and $\mathcal S(u_{m})=[u_{m+1},u_{m-1}]$, or so that $R_{m,M}$ is to the right of $R_{M,m}$ in the case in which $\mathcal P(u_{M})=[u_{M+1},u_{M-1}]$ and $\mathcal S(u_{m})=[u_{m-1},u_{m+1}]$. One of the two cases happens, because of Condition~(3). We now redraw the edges of $G_{m,M}$ and $G_{M,m}$ inside $R_{m,M}$ and $R_{M,m}$, respectively. The drawing $\Gamma_G$ obtained as the union of the constructed drawings of $G_{m,M}$ and $G_{M,m}$ is the desired drawing of $G$. This completes the proof of the claim. 
	\end{proof}
	
	From a computational point of view, we act as follows. 
	
	First we compute, for every possible pair of vertices $(u_m,u_M)$ with $m,M\in \{1,\dots,n\}$ and with $m\neq M$, whether there are upward planar drawings of $G_{m,M}$ and $G_{M,m}$ extending $\Gamma_{H,m,M}$ and $\Gamma_{H,M,m}$, respectively, in which the vertex with the smallest $y$-coordinate is $u_m$ and the vertex with the largest $y$-coordinate is $u_M$. This can be done by considering the $2n$-vertex path $(u_1,u_2,\dots,u_n,u_{n+1}=u_1,u_{n+2}=u_2,\dots,u_{2n}=u_n)$ and by setting up a dynamic programming table with entries $t(u_i,u_j,u_{m'},u_{M'})$, for all the indices $i,j,m',M'\in \{1,\dots,2n\}$ such that $i \leq m' \leq j$ and $i \leq M' \leq j$, with $i\neq j$, $m'\neq M'$, and $j-i\leq n$. The values of the entries of this table can be computed in total $O(n^4)$ time as in the proof of~\cref{th:algo-UPE-FUE-path}. 
	
	Now, for each of the $O(n^2)$ pairs of vertices $(u_m,u_M)$ with $m,M\in \{1,\dots,n\}$ and with $m\neq M$, we exploit~\cref{cl:cycles-two-paths} in order to check whether $G$ has an upward planar drawing extending $\Gamma_H$ in which the vertices with the smallest and largest $y$-coordinate are $u_m$ and $u_M$, respectively. Concerning Conditions~(1) and~(2), we check in $O(1)$ time whether $t(u_m,u_M,u_m,u_M)=$ {\sc true} and $t(u_M,u_{n+m},u_{n+m},u_M)=$ {\sc true} (if $m<M$) or whether $t(u_M,u_m,u_m,u_M)=$ {\sc true} and $t(u_m,u_{n+M},u_m,u_{n+M})=$ {\sc true} (if $m>M$). Condition~(3) can also be trivially checked in $O(1)$ time. This concludes the proof of the theorem.
\end{proof}

It turns out that directed paths and cycles are much easier to handle in the case in which they do not come with a given upward embedding, as in the next theorem.

\begin{restatable}{theorem}{TheoremAlgoUPEpaths}\label{th:algo-UPE-paths}
	The {\sc UPE} problem can be solved in $O(n)$ time for instances $\langle G, H, \Gamma_H \rangle$ such that $G$ is an $n$-vertex directed path or cycle, $H$ contains no edges, and no two vertices share the same $y$-coordinate in $\Gamma_H$.
\end{restatable}

\begin{proof}
	Suppose first that $G$ is a directed path. We partition $G$ into $k$ monotone paths $G_i=(u^i_1,u^i_2,\dots,u^i_{h(i)})$, for some integer $k\geq 1$. Assume that the edge $(u^1_1,u^1_2)$ is exiting $u^1_1$ and entering $u^1_2$, the other case being symmetric. Then for every odd $i$ we have that $u^i_{h(i)}=u^{i+1}_{h(i+1)}$ is a sink, while for every even $i$ we have that $u^i_1=u^{i+1}_1$ is a source. Our algorithm is based on the following characterization.
	
	\begin{myclaim} \label{cl:path-no-emb}
		There exists an upward planar drawing of $G$ extending $\Gamma_H$ if and only if, for every $i$, $j$, and $j'$ with $j<j'$ such that $u^i_j,u^i_{j'}\in V(H)$, it holds true that $y(u^i_j)<y(u^i_{j'})$ in $\Gamma_H$. 
	\end{myclaim}
	
	\begin{proof}
		The necessity is trivial. Indeed, if there are indices $i$, $j$, and $j'$ with $j<j'$ such that $u^i_j,u^i_{j'}\in V(H)$ and such that $y(u^i_j)>y(u^i_{j'})$, then the monotone path $(u^i_1,u^i_2,\dots,u^i_{h(i)})$ cannot be upward in any drawing of $G$ extending $\Gamma_H$. 
		
		For the sufficiency, we construct an upward planar drawing $\Gamma_G$ of $G$ extending $\Gamma_H$ by drawing one monotone path $G_i=(u^i_1,u^i_2,\dots,u^i_{h(i)})$ at a time. Roughly speaking, this is done by drawing $G_i$ ``to the right'' of what has been drawn so far. More precisely, for any $i=1,\dots,k$, consider the path $G_{1,\dots,i}=(G_1 \cup \dots \cup G_i)$ and let the end-vertex of $G_{1,\dots,i}$ different from $u^1_1$ be the {\em last} vertex of $G_{1,\dots,i}$. Denote by $\Gamma_{H,i}$ the restriction of $\Gamma_H$ to the vertices in $G_{1,\dots,i}$. We show how to construct an upward planar drawing $\Gamma_i$ of $G_{1,\dots,i}$ that extends $\Gamma_{H,i}$ and such that the last vertex of $G_{1,\dots,i}$ is visible from the right (meaning that a half-line starting at the last vertex and directed rightwards does not intersect $\Gamma_i$ other than at its starting point). Assume that $i$ is odd, the case in which $i$ is even is similar. Then $u^{i-1}_1=u^i_1$ is a source and is visible from the right in $\Gamma_{i-1}$ (the latter condition is vacuously true if $i=1$). We proceed similarly to the proofs of~\cref{cl:paths-no-endv,cl:paths-one-endv,cl:paths-two-endv}. Namely,  we consider a $y$-monotone curve $\gamma_i$ passing through $u^i_1$, having the vertices in $V(G_{1,\dots,i-1})$ -- as they are placed in $\Gamma_{i-1}$ -- to the left (except for $u^i_1$) and the vertices in $V(H)\cap (V(G_i)\cup \dots \cup V(G_k))$ to the right (except for $u^i_1$). We redraw parts of the edges of $G_{1,\dots,i-1}$ so that $\Gamma_{i-1}$ entirely lies to the left of $\gamma_i$, except at $u^i_1$; note that this is possible because $u^i_1$ is visible from the right. If $u^i_1\notin V(H)$ we further modify $\Gamma_{i-1}$ by moving $u^{i}_1$ downwards along $\gamma_i$ and by extending the edge $(u^{i-1}_1,u^{i-1}_2)$ downwards as well (while keeping it to the left of $\gamma_i$), so that $y(u^i_1)$ is smaller than the $y$-coordinate of every vertex different from $u^i_1$ in $V(H)\cap V(G_i)$. If $u^i_1\in V(H)$, then it is the condition of the claim that guarantees that $y(u^i_1)$ is smaller than the $y$-coordinate of every vertex in $V(H)\cap V(G_i)$. We next construct an upward planar drawing $\mathcal G_i$ of $G_i$, possibly intersecting $\Gamma_{i-1}$, extending the restriction of $\Gamma_H$ to the vertices in $V(H)\cap V(G_i)$. This is done by drawing $G_i$ as a $y$-monotone curve passing through the vertices in $V(H)\cap V(G_i)$ and by then placing the vertices in $V(G_i)$ not in $V(H)$ at suitable points on this curve. Finally, we redraw parts of the edges of $G_i$ so that $\mathcal G_i$ entirely lies to the right of $\gamma_i$, except at $u^i_1$. Now there is no crossing between the edges in $\Gamma_{i-1}$ and those in $\mathcal G_i$; further, $u^{i}_{h(i)}$ is visible from the right. Hence, when $i=k$ this algorithm constructs an upward planar drawing $\Gamma_G$ of $G$ extending $\Gamma_H$.
	\end{proof}

	
	


	The condition in \cref{cl:path-no-emb} can be easily checked in $O(n)$ time. Indeed, each path $(u^i_1,u^i_2,\dots,u^i_{h(i)})$ can be independently traversed from $u^i_1$ to $u^i_{h(i)}$ while keeping track of the $y$-coordinate of the last encountered vertex in $V(H)$: If the $y$-coordinate of any vertex in $V(H)$ is smaller than the $y$-coordinate of the previous vertex in $V(H)$, then the condition is not satisfied; vertices not in $V(H)$ are ignored. 
	
	We now turn our attention to the case in which $G$ is a cycle. We again partition $G$ into $k$ monotone paths $G_i=(u^i_1,u^i_2,\dots,u^i_{h(i)})$, for some integer $k\geq 2$. Then the characterization stated in~\cref{cl:path-no-emb} applies to this case as well. Whether $\langle G, H, \Gamma_H \rangle$ satisfies the characterization can be checked in $O(n)$ time, as for directed paths. The necessity of the characterization can be proved as in~\cref{cl:path-no-emb}. The sufficiency can also be proved similarly to~\cref{cl:path-no-emb}, however with one more ingredient. Namely, we have to find a source $u^i_1$ and a sink $u^j_{h(j)}$ (possibly $i=j$) which are going to be the vertices with the smallest and largest $y$-coordinate in the upward planar drawing of $G$ extending $\Gamma_H$ which we are going to construct. Any source $u^i_1$ of $G$ not in $H$ can in fact be selected for the task; if all the sources of $G$ are in $H$, then the source $u^i_1$ with the smallest $y$-coordinate can be chosen instead. The sink $u^j_{h(j)}$ can be chosen analogously. We now split $G$ into two directed paths $G_1$ and $G_2$ between $u^i_1$ and $u^j_{h(j)}$. We independently construct upward planar drawings $\Gamma_{G_1}$ and $\Gamma_{G_2}$ of $G_1$ and $G_2$ extending $\Gamma_{H_1}$ and $\Gamma_{H_2}$, respectively, in which the vertex with the smallest $y$-coordinate is $u^i_1$ and the vertex with the largest $y$-coordinate is $u^j_{h(j)}$ -- $\Gamma_{H_1}$ and $\Gamma_{H_2}$ are the restrictions of $\Gamma_H$ to the vertices in $V(H)\cap V(G_1)$ and in $V(H)\cap V(G_2)$, respectively. The drawings $\Gamma_{G_1}$ and $\Gamma_{G_2}$ can be constructed as in the proof of~\cref{cl:path-no-emb}. We then modify $\Gamma_{G_1}$ and $\Gamma_{G_2}$ as in~\cref{cl:cycles-two-paths} so that:
	\begin{enumerate}[(i)]
		\item $u^i_1$ is at the same point in $\Gamma_{G_1}$ and $\Gamma_{G_2}$; and 
		\item $u^j_{h(j)}$ is at the same point in $\Gamma_{G_1}$ and $\Gamma_{G_2}$. 
	\end{enumerate}
	Gluing together $\Gamma_{G_1}$ and $\Gamma_{G_2}$ results in an upward drawing $\Gamma'_{G}$ of $G$ that extends $\Gamma_H$ and in which $u^i_1$ and $u^j_{h(j)}$ are the vertices with the smallest and largest $y$-coordinate, respectively. However, $\Gamma'_{G}$ might contain crossings. Hence, we redraw the curves representing the edges of $G_1$ and $G_2$, while leaving the position of every vertex unaltered. This is done again as in the proof of~\cref{cl:cycles-two-paths}. 
\end{proof}

\section{Conclusions and Open Problems}\label{se:conclusions}

In this paper we introduced and studied the {\sc Upward Planarity Extension} ({\sc UPE}) problem, which takes in input an upward planar drawing $\Gamma_H$ of a subgraph $H$ of a directed graph $G$ and asks whether an upward planar drawing of $G$ exists which coincides with $\Gamma_H$ when restricted to the \mbox{vertices and edges of $H$.}

We proved that the {\sc UPE} problem is NP-complete, even if $G$ has a prescribed upward embedding and $H$ contains all the vertices and no edges. Conversely, the problem can be solved efficiently for upward planar $st$-graphs.     

Several questions are left open by our research. We cite our favorite two. 

First, is it possible to solve the {\sc UPE-FUE} problem in polynomial time for instances $\langle G,H,\Gamma_H\rangle$ such that $H$ contains no edges and no two vertices have the same $y$-coordinate in $\Gamma_H$? We proved that if any of the two conditions is dropped, then the {\sc UPE-FUE} problem is NP-hard, however we can only provide a positive answer to the above question if we further assume that $G$ is a directed path or cycle. 

Second, are the {\sc UPE} and {\sc UPE-FUE} problems polynomial-time solvable for directed paths and cycles?  Even with the assumption that $H$ contains no edges and no two vertices have the same $y$-coordinate in $\Gamma_H$, answering the above question in the affirmative was not a trivial task.

\medskip
\noindent \paragraph{\textbf{Acknowledgments.}} \cref{le:characterization-upward-plane} comes from a research session the third author had with Ignaz Rutter, to which our thanks go.

\bibliographystyle{abbrv}
\bibliography{bibliography}

\begin{thebibliography}{10}

\bibitem{Angelini:2018:WPE:3266298.3239561}
P.~Angelini, G.~{Da Lozzo}, G.~{Di Battista}, V.~{Di Donato}, P.~Kindermann,
  G.~Rote, and I.~Rutter.
\newblock Windrose planarity: Embedding graphs with direction-constrained
  edges.
\newblock {\em ACM Trans. Algorithms}, 14(4):54:1--54:24, 2018.

\bibitem{DBLP:journals/algorithmica/AngeliniLBF17}
P.~Angelini, G.~{Da Lozzo}, G.~{Di Battista}, and F.~Frati.
\newblock Strip planarity testing for embedded planar graphs.
\newblock {\em Algorithmica}, 77(4):1022--1059, 2017.

\bibitem{adf-tppeg-15}
P.~Angelini, G.~{Di Battista}, F.~Frati, V.~Jel{\'{\i}}nek,
  J.~Kratochv{\'{\i}}l, M.~Patrignani, and I.~Rutter.
\newblock Testing planarity of partially embedded graphs.
\newblock {\em {ACM} Trans.\ Algorithms}, 11(4):32:1--32:42, 2015.

\bibitem{DBLP:journals/algorithmica/BertolazziBD02}
P.~Bertolazzi, G.~{Di Battista}, and W.~Didimo.
\newblock Quasi-upward planarity.
\newblock {\em Algorithmica}, 32(3):474--506, 2002.

\bibitem{BertolazziBLM94}
P.~Bertolazzi, G.~{Di Battista}, G.~Liotta, and C.~Mannino.
\newblock Upward drawings of triconnected digraphs.
\newblock {\em Algorithmica}, 12(6):476--497, 1994.

\bibitem{DBLP:journals/siamcomp/BertolazziBMT98}
P.~Bertolazzi, G.~{Di Battista}, C.~Mannino, and R.~Tamassia.
\newblock Optimal upward planarity testing of single-source digraphs.
\newblock {\em {SIAM} J. Comput.}, 27(1):132--169, 1998.

\bibitem{DBLP:journals/cj/BinucciD16}
C.~Binucci and W.~Didimo.
\newblock Computing quasi-upward planar drawings of mixed graphs.
\newblock {\em Comput. J.}, 59(1):133--150, 2016.

\bibitem{DBLP:journals/comgeo/Brandenburg14}
F.~Brandenburg.
\newblock Upward planar drawings on the standing and the rolling cylinders.
\newblock {\em Comput. Geom.}, 47(1):25--41, 2014.

\bibitem{DBLP:conf/soda/BrucknerR17}
G.~Br{\"{u}}ckner and I.~Rutter.
\newblock Partial and constrained level planarity.
\newblock In P.~N. Klein, editor, {\em {SODA} 2017}, pages 2000--2011. {SIAM},
  2017.

\bibitem{ccc-pl-17}
S.~Chaplick, M.~Chimani, S.~Cornelsen, G.~{Da Lozzo}, M.~N{\"{o}}llenburg,
  M.~Patrignani, I.~G. Tollis, and A.~Wolff.
\newblock Planar l-drawings of directed graphs.
\newblock In F.~Frati and K.-L. Ma, editors, {\em {GD} 2017}, volume 10692 of
  {\em LNCS}, pages 465--478. Springer, 2017.

\bibitem{cdk-crp-14}
S.~Chaplick, P.~Dorbec, J.~Kratochv{\'{\i}}l, M.~Montassier, and J.~Stacho.
\newblock Contact representations of planar graphs: Extending a partial
  representation is hard.
\newblock In D.~Kratsch and I.~Todinca, editors, {\em {WG} 2014}, volume 8747
  of {\em LNCS}, pages 139--151. Springer, 2014.

\bibitem{cfk-epr-13}
S.~Chaplick, R.~Fulek, and P.~Klav{\'{\i}}k.
\newblock Extending partial representations of circle graphs.
\newblock In S.~K. Wismath and A.~Wolff, editors, {\em {GD} 2013}, volume 8242
  of {\em LNCS}, pages 131--142. Springer, 2013.

\bibitem{cgg-pvrep-18}
S.~Chaplick, G.~Guspiel, G.~Gutowski, T.~Krawczyk, and G.~Liotta.
\newblock The partial visibility representation extension problem.
\newblock {\em Algorithmica}, 80(8):2286--2323, 2018.

\bibitem{ddf-upm-18}
G.~{Da Lozzo}, G.~{Di Battista}, F.~Frati, M.~Patrignani, and V.~Roselli.
\newblock Upward planar morphs.
\newblock In T.~C. Biedl and A.~Kerren, editors, {\em {GD} 2018}, volume 11282
  of {\em LNCS}, pages 92--105. Springer, 2018.

\bibitem{dt-aprad-88}
G.~{Di Battista} and R.~Tamassia.
\newblock Algorithms for plane representations of acyclic digraphs.
\newblock {\em Theor. Comput. Sci.}, 61:175--198, 1988.

\bibitem{DBLP:journals/siamcomp/BattistaT96}
G.~{Di Battista} and R.~Tamassia.
\newblock On-line planarity testing.
\newblock {\em {SIAM} J. Comput.}, 25(5):956--997, 1996.

\bibitem{dtt-arsdpud-92}
G.~{Di Battista}, R.~Tamassia, and I.~G. Tollis.
\newblock Area requirement and symmetry display of planar upward drawings.
\newblock {\em Discrete {\&} Comput. Geometry}, 7:381--401, 1992.

\bibitem{DBLP:journals/jgt/Fiala03}
J.~Fiala.
\newblock {NP} completeness of the edge precoloring extension problem on
  bipartite graphs.
\newblock {\em Journal of Graph Theory}, 43(2):156--160, 2003.

\bibitem{GargT01}
A.~Garg and R.~Tamassia.
\newblock On the computational complexity of upward and rectilinear planarity
  testing.
\newblock {\em {SIAM} J. Comput.}, 31(2):601--625, 2001.

\bibitem{jkr-ktt-13}
V.~Jel{\'{\i}}nek, J.~Kratochv{\'{\i}}l, and I.~Rutter.
\newblock A {K}uratowski-type theorem for planarity of partially embedded
  graphs.
\newblock {\em Comput. Geom.}, 46(4):466--492, 2013.

\bibitem{DBLP:conf/gd/JungerLM98}
M.~J{\"{u}}nger, S.~Leipert, and P.~Mutzel.
\newblock Level planarity testing in linear time.
\newblock In S.~Whitesides, editor, {\em GD'98}, volume 1547 of {\em LNCS},
  pages 224--237. Springer, 1998.

\bibitem{kkk-fgp-12}
P.~Klav{\'{\i}}k, J.~Kratochv{\'{\i}}l, T.~Krawczyk, and B.~Walczak.
\newblock Extending partial representations of function graphs and permutation
  graphs.
\newblock In L.~Epstein and P.~Ferragina, editors, {\em {ESA} 2012}, volume
  7501 of {\em LNCS}, pages 671--682. Springer, 2012.

\bibitem{kko-epr-17}
P.~Klav{\'{\i}}k, J.~Kratochv{\'{\i}}l, Y.~Otachi, I.~Rutter, T.~Saitoh,
  M.~Saumell, and T.~Vyskocil.
\newblock Extending partial representations of proper and unit interval graphs.
\newblock {\em Algorithmica}, 77(4):1071--1104, 2017.

\bibitem{kko-scg-15}
P.~Klav{\'{\i}}k, J.~Kratochv{\'{\i}}l, Y.~Otachi, and T.~Saitoh.
\newblock Extending partial representations of subclasses of chordal graphs.
\newblock {\em Theor. Comput. Sci.}, 576:85--101, 2015.

\bibitem{kko-ig-17}
P.~Klav{\'{\i}}k, J.~Kratochv{\'{\i}}l, Y.~Otachi, T.~Saitoh, and T.~Vyskocil.
\newblock Extending partial representations of interval graphs.
\newblock {\em Algorithmica}, 78(3):945--967, 2017.

\bibitem{KlemzR17}
B.~Klemz and G.~Rote.
\newblock Ordered level planarity, geodesic planarity and bi-monotonicity.
\newblock In F.~Frati and K.~L. Ma, editors, {\em {GD} 17}, volume 10692 of
  {\em LNCS}, pages 440--453. Springer, 2017.

\bibitem{DBLP:journals/jgt/KratochvilS97}
J.~Kratochv{\'{\i}}l and A.~Seb{\"{o}}.
\newblock Coloring precolored perfect graphs.
\newblock {\em Journal of Graph Theory}, 25(3):207--215, 1997.

\bibitem{p-epsd-06}
M.~Patrignani.
\newblock On extending a partial straight-line drawing.
\newblock {\em Int. J. Found. Comput. Sci.}, 17(5):1061--1070, 2006.

\bibitem{DBLP:journals/cj/RextinH17}
A.~Rextin and P.~Healy.
\newblock Dynamic upward planarity testing of single source embedded digraphs.
\newblock {\em Comput. J.}, 60(1):45--59, 2017.

\bibitem{DBLP:journals/siamcomp/Tamassia87}
R.~Tamassia.
\newblock On embedding a graph in the grid with the minimum number of bends.
\newblock {\em {SIAM} J. Comput.}, 16(3):421--444, 1987.

\end{thebibliography}

\end{document}